\documentclass[a4paper,11pt]{article}
\usepackage{amsmath,amsthm,mathrsfs,amsfonts,dsfont}
\usepackage{geometry,pdflscape,graphicx,wrapfig,color,hyperref,tikz,bbm}
\usepackage[capposition=top]{floatrow}
\usepackage[comma,square]{natbib}

\hypersetup{colorlinks,citecolor = blue}
\graphicspath{{images/}}
\makeatletter

\newcommand{\Rmnum}[1]{\expandafter\@slowromancap\romannumeral #1@}
\newcommand{\indep}{\rotatebox[origin=c]{90}{$\models$}}
\newtheorem{prop}{Proposition}
\newtheorem{lem}{Lemma}
\newtheorem{thm}{Theorem}
\newtheorem{corol}{Corollary}
\newtheorem{remk}{Remark}

\theoremstyle{definition}
\newtheorem{assum}{Assumption}

\makeatother
\title{Model-Free Approaches to Discern Non-Stationary Microstructure Noise and Time-Varying Liquidity in High-Frequency Data\footnote{
The authors thank the referees for valuable suggestions and insights which make significant improvements. The authors also benefited much from discussions with Yingying Li, Xinhua Zheng, Yoann Potiron, Markus Bibinger. This research was funded by National Science Foundation Grant DMS 14-07812. All comments are gratefully welcomed.}}

\author{Richard Y. Chen\footnote{Richard Y. Chen is a Ph.D. student in Statistics at the University of Chicago. Email: \href{mailto:yrchen@uchicago.edu}{yrchen@uchicago.edu}.}, Per A. Mykland\footnote{Per A. Mykland is the Robert M. Hutchins Distinguished Service Professor, Department of Statistics, the University of Chicago. Email: \href{mailto:mykland@pascal.uchicago.edu}{mykland@pascal.uchicago.edu}.}}
\date{May, 2017}

\begin{document}
\maketitle

\begin{abstract}
In this paper, we provide non-parametric statistical tools to test stationarity of microstructure noise in general hidden It\^o semimartingales, and discuss how to measure liquidity risk using high-frequency financial data. In particular, we investigate the impact of non-stationary microstructure noise on some volatility estimators, and design three complementary tests by exploiting edge effects, information aggregation of local estimates and high-frequency asymptotic approximation. The asymptotic distributions of these tests are available under both stationary and non-stationary assumptions, thereby enable us to conservatively control type-I errors and meanwhile ensure the proposed tests enjoy the asymptotically optimal statistical power. Besides, it also enables us to empirically measure aggregate liquidity risks by these test statistics. As byproducts, functional dependence and endogenous microstructure noise are briefly discussed. Simulation with a realistic configuration corroborates our theoretical results, and our empirical study indicates the prevalence of non-stationary microstructure noise in New York Stock Exchange.
\bigskip

\noindent\textbf{Keywords}: Microstructure, High-Frequency Tests, Statistical Powers, Stable Central Limit Theorems, Non-Stationarity, Volatility, Liquidity
\vspace{3mm}

\noindent\textbf{JEL classification}: C12, C13, C14, C58
\end{abstract}

\newpage
\section{Introduction}
The introduction of high-tech trading mechanisms into markets, for example, electronic communication networks (ECNs) and other electronic trading platforms, provides an opportunity for speculators and market makers to take advantage of speed in trading and market making, and this technological innovation also brings new regulatory challenges. The subsequent high-frequency trading results in a huge amount of high-frequency transaction and quotation data, which in particular opens two potential gates for research in theoretical and empirical asset pricing: one is estimation methodology using high-frequency data, since practitioners and researchers can get access to the big data and estimate variables of interest with greater accuracy; the other is a ``frog eyes' view'' on market microstructure, since low-latency data offers a valuable chance to investigate trading behaviors with a higher resolution than ever before.

Correspondingly, this paper's contributions to the literature are twofold: \textit{i)} one is stationarity test of microstructure noise, we study the estimation problem when using high-frequency data with non-stationary noises, and then test non-stationarity in microstructure noise via several complementary model-free approaches; \textit{ii)} the other one is on empirical market microstructure, since the microstructure noise can capture some information about market quality and liquidity, we estimate noise levels as measures of time-varying bid-ask spreads, risk aversions of market participants, etc., and detect short-term liquidity variations.

\subsection{Literature review}
The high-frequency finance practice motivates two clearly distinct and closely related researches:

One is more accurate estimation in financial econometrics, to name a few but not all, the estimation of integrated volatilities, quadratic covariances, the activities of jumps, the leverage effects, the volatility of volatility, the lead-lag effect. This stream of research started from \cite{j94,jp98} from the perspective of stochastic calculus, and \cite{fn96,e00,z01,abdl01,bs02} in the context of econometrics. Now, the high-frequency financial econometrics has already developed into a considerably influential research field with numerous prominent scholars and there are already monographs on this area: \cite{js03,jp12} developed probabilistic tools for high-frequency financial data analysis, \cite{aj14} provided an excellent overview in econometric literature, \cite{h12} is a good account from a financial standpoint. There are also academic chapters concisely reviewing high-frequency financial econometrics: \cite{re10,mz12,j12}.

The other one is the study of market microstructure. Low-latency data allows financial practitioners and researchers to look at the financial markets at a higher resolution level, for example, one can know the bid/ask dynamics within each second, one can also know the order flow through the limit order book. The market microstructure theory studies how the latent demand and latent supply of market participants are ultimately translated into prices by studying the specific market structure in detail. The cornerstone papers include \cite{gm85,k85}, both of them are using (pesudo)\footnote{To say it ``pseudo'' because in the model considered in \cite{k85}, the market maker does not aim to maximize their utility, instead his or her objective is only to guarantee market clearing.} game-theoretical argument in information economics. More comprehensive books include \cite{o95,h07}. However, when looking closely at the transaction or quotation prices, one can find that the price is no longer an It\^o semimartingale, not even random walk. For this reason, according to market microstructure theory \citep{o03}, the semimartingale model in classical asset pricing theory \citep{hp81,ds94} is not a photographic depiction of the real prices of financial assets, yet it is still a fairly good approximation to asset prices when the sampling frequency is relatively low, and that is the reason the literature suggests using at most 5-minute subsampling.

Some estimation methods for integrated volatility using noisy high-frequency financial data have already been well established: \textit{i)} \cite{zma05} found the first consistent estimator (two-time scale realized volatility) using subsampling and averaging in the presence of i.i.d. market microstructure noise and \cite{z06} gave a multi-scale version with the optimal rate of convergence $n^{\frac{1}{4}}$, \cite{lm07} discussed the robustness of TSRV to noise assumptions in general, \cite{kl08} generalized the TSRV to the model with endogenous and diurnal noise and put forward a modified version of TSRV which we shall use in this paper. Later, \cite{amz11} generalized the model to allowing correlated noises under stationary and strong-mixing conditions; \textit{ii)} \cite{b08} provided a kernel-based estimator under the model in which the noise process is temporarily dependent and stationary and possibly linearly correlated with the latent It{\^o} process, their inference is also robust to endogenous spacing; \textit{iii)} \cite{j09} designed a generalized version of the pre-averaging approach \citep{pv09}, under a Markovian noise model which allows arbitrary fashion of noise but without correlation between noise and the latent process; \textit{iv)} Motivated by the likelihood method from \cite{amz05}, \cite{x10} established quasi-maximum likelihood method (QMLE) in the estimation of integrated volatility; \textit{v)} \cite{bhmr14} developed the local generalized method of moments to estimate quadratic covariation using noisy high-frequency data.

Many estimators of integrated volatilities using high-frequency noisy data were developed under the assumption that the microstructure noise is stationary.
However, the literature in empirical finance, such as \cite{ap88,h93,ab97,gjf99}, has already shown in 1990s that markets exhibit systematic intra-day patterns. Therefore, allowing heteroskedasticity and non-stationary in microstructure noise in integrated volatility estimation is of particular importance in application. Particularly, \cite{kl08} used a parametric model to describe the diurnal pattern in microstructure noise. \cite{ay09} used the estimates of noise variance in high-frequency data to measure the market liquidity from June 1996 to December 2005. There is other related research in the literature, \cite{acd09} studied the changes in microstructure noise due to sampling frequency, \cite{bry13} derived the optimal sampling frequency in terms of finite-sample forecast mean squared error in linear forecast model with non-stationary microstructure noise.


\subsection{Structure of this paper}
Section \ref{model} describes our model and assumptions; after showing non-stationarity effect on the two scale estimator, complementary statistical tests are designed to detect microstructure noise stationarity based on high-frequency asymptotics, the asymptotic distributions under both null and alternative hypotheses and their implications for testing are shown in section \ref{Test-1}, \ref{test-23}, \ref{discuss}; section \ref{alternative} introduces an aggregate measure of liquidity risk and studies its estimation problem; relation between volatility and variance of microstructure noise, as well as endogenous microstructure noise are discussed in section \ref{modelpp}; section \ref{simul} and \ref{empirical} contain our simulation and empirical analysis; section \ref{conclu} concludes. Some proofs are given in the Appendix.

\section{The model and assumptions}\label{model}
\subsection{Model setup}\label{setup}
Firstly, we have a filtered probability space $\left(\Omega^{(0)},\mathcal{F}^{(0)},\left\{\mathcal{F}^{(0)}_t\right\}_{t\ge 0},\mathbb{P}^{(0)}\right)$ on which a latent It\^o semimartingale $\{X_t\}_{t\ge 0}$ is adapted, and can be described by
\begin{equation}\label{X}
X_t=X_0+\int_0^tb_s\mathrm{d}s+\int_0^t\sigma_s\mathrm{d}W_s+J_t
\end{equation}
where $\{b_t\}_{t\ge0}$ is the drift, $\sigma^2_t$ is the spot volatility in financial terminology (for example, its dynamics can be described by Heston model \citep{he93}); $\{W_t\}_{t\ge 0}$ is a 1-dimensional Wiener process; $J_t$ is a jump process which is described in subsection \ref{assumptions}.

Secondly, we have another filtered probability space $\left(\Omega^{(1)},\mathcal{F}^{(1)},\left\{\mathcal{F}^{(1)}_t\right\}_{t\ge 0},\mathbb{P}^{(1)}\right)$ on which the observable process $\{Y_t\}_{t\ge 0}$ is adapted. Then, we can define the market microstructure noise process $\{e_t\}_{t\ge 0}$\footnote{Although the noise is immaterial outside the observation times, it does not harm to assume there exist such a noise process in continuous time.}, as the difference between the latent and observable processes:
\begin{equation}\label{error}
e_t\equiv Y_t-X_t
\end{equation}
besides we define
\begin{equation}\label{Z}
Z_t\equiv E_{\mathbb{P}^{(1)}}(Y_t|\mathcal{F}^{(0)})=X_t+E_{\mathbb{P}^{(1)}}(e_t|\mathcal{F}^{(0)})
\end{equation}
we call $\{Z_t\}_{t\ge 0}$ the ``\textit{estimable latent process}'' because we can indeed estimate it from the actual observations via, for example, pre-averaging \citep{pv09,j09,jpv10,mz16a}. It is natural to assume the process $\{Z_t\}_{t\ge 0}$ is an It\^o semimartingale, for example, if we assume $Z_t=f(X_t)$ for some $f(\cdot)\in\mathcal{C}^2(\mathbb{R})$ \citep{lm07} then $\{Z_t\}_{t\ge0}$ is an It\^o semimartingale\footnote{The definition (\ref{Z}) suggests the possibility of our inability to recover the latent process $\{X_t\}$  from the noisy observations $\{Y_t\}$, since $Z_t$ does not necessarily equal to $X_t$. More strikingly, as later discussed, this allows correlation between the microstructure noise and the latent process.}. Based upon $\{Z_t\}_{t\ge0}$, we can define a noise process $\{\epsilon_t\}_{t\ge 0}$ of another form, which is not necessarily the difference between the observable process $\{Y_t\}_{t\ge 0}$ and the latent process $\{X_t\}_{t\ge 0}$, instead defined theoretically via
\begin{equation}\label{epsilon}
\epsilon_t\equiv Y_t-Z_t=e_t-E_{\mathbb{P}^{(1)}}(e_t|\mathcal{F}^{(0)})
\end{equation}
we call $\{\epsilon_t\}_{t\ge 0}$ the ``\textit{distinguishable noise}'', which can be disentangled from the estimable latent process $\{Z_t\}_{t\ge 0}$ \citep{br06b}.

Thirdly, we have a Markov kernel to provide a connection between the processes $\{X_t\}_{t\ge 0}$ and $\{Y_t\}_{t\ge 0}$ 
, namely $Q_t(\omega^{(0)},\mathrm{d}y): (\Omega^{(0)},\mathcal{F}^{(0)})\times(\mathbb{R},\mathcal{B}(\mathbb{R})) \longmapsto[0,1]$, i.e., conditional on the whole latent process $X$, there exists a probability measure on the space $(\mathbb{R},\mathcal{B}(\mathbb{R}))$.

Thus, all the relevant process, either latent or observable, 
can be defined on the extended filtered probability space $\left(\Omega,\mathcal{F},\{\mathcal{F}_t\}_{t\ge 0},\mathbb{P}\right)$\footnote{This model combines the features of the models in \cite{lm07} and \cite{jp12} (or \cite{j09,jpv10}), and is endowed with a additional feature that $\epsilon_{t_i}$'s are not defined as the differences between the observations $Y_{t_i}$'s and the latent values $X_{t_i}$'s but the differences between the observations $Y_{t_i}$'s and the values we can actually recover which are $Z_{t_i}$'s.} where
\begin{equation}
\left\{\begin{array}{ll}\Omega\equiv \Omega^{(0)}\times\Omega^{(1)},\text{  } \mathcal{F}\equiv \mathcal{F}^{(0)}\otimes\mathcal{F}^{(1)}\\
\mathcal{F}_t\equiv \bigcap_{s>t}\mathcal{F}_s^{(0)}\otimes\mathcal{F}_s^{(1)}\\
\mathbb{P}(\mathrm{d}\omega^{(0)},\mathrm{d}\omega^{(1)})\equiv \mathbb{P}^{(0)}(\mathrm{d}\omega^{(0)})\cdot \otimes_{t\ge 0}Q_t(\omega^{(0)},\mathrm{d}y_t(\omega^{(1)}))\end{array}\right.
\end{equation}

Moreover, define $$g_t(\omega^{(0)})=\int_{\mathbb{R}}\left|y-Z_t(\omega^{(0)})\right|^2Q_t(\omega^{(0)},\mathrm{d}y)$$
i.e., $g_{t_i}= E(\epsilon^2_{t_i}|\mathcal{F}^{(0)})$. By this definition, $\{g_t\}_{t\ge0}$ is also a stochastic process. Note that $g_t$ could depend on more than one latent random variables, i.e., it is possible that $g_t(\omega^{(0)})=g_t(X_t(\omega^{(0)}),Z_t(\omega^{(0)}),\sigma^2_t(\omega^{(0)}),\cdots)$ for each $t$. 
In Section \ref{test-23} and Section \ref{alternative} regarding some behaviors in the presence of non-stationary microstructure noise, we pose specific restrictions on the process $\{g_t\}_{t\ge 0}$, and let it be an It\^o diffusion, and use asymptotic properties to show asymptotically optimal power and measure liquidity in high-frequency data.

\subsection{Observational notation}\label{ObsNotation}
This subsection can be skipped at the first reading. Please be advised to go back to this subsection when encounter the observational notation in later sections.

Suppose we focus on a finite interval $[0,T]$ on which ultra-high frequency data is recorded. Define $\mathcal{G}$ to be the finest time grid whence all the observations were obtained. Suppose we have $n$ observations after the reference starting point 0, then $\mathcal{G}$ can be written as
\begin{equation}\label{G}
    \mathcal{G}\equiv\{t_0=0,t_1,t_2,\cdots,t_n\}
\end{equation}
We sometimes do sparse sampling, typically start from the $k$-th observation and take one sample from every $K$ observations. Formally, we define sub-grids $\mathcal{G}^{(K,k)}$'s indexed by $k=0,\cdots,K-1$ for each $K\in\mathbb{N}_+$: 
\begin{equation}\label{G_K}
    \mathcal{G}^{(k)}=\mathcal{G}^{(K,k)}\equiv\left\{t_k,t_{k+K},t_{k+2K},\cdots,t_{k+\left(\left\lfloor n/K\right\rfloor-1\right)K}\right\}, \text{ where } k=0,1,\cdots,K-1
\end{equation}
To analyze the edge effect\footnote{The edge effect is a pervasive phenomenon in non-parametric high-frequency econometrics. Verbally stated, edge effect is ``information phasing in and phasing out at the edges of time intervals'', which is caused by inhomogeneous usage of data. Although undesirable, this feature is inevitable in estimation.} and the modified TSRV, we need more observational notation:
\begin{equation}\label{subgrid}
	\left.\begin{array}{ll}
		\mathcal{G}'^{(k)}=\mathcal{G}'^{(K,k)}&\equiv\{\min{\mathcal{G}^{(K,k)}}+1, \min{\mathcal{G}^{(K,k)}}+2, \cdots, \max{\mathcal{G}^{(K,k)}}\}\\
		\mathcal{G}''^{(k)}=\mathcal{G}''^{(K,k)}&\equiv\{\min{\mathcal{G}^{(K,k)}}+1, \min{\mathcal{G}^{(K,k)}}+2, \cdots, \max{\mathcal{G}^{(K,k)}}-1\}\\
		\mathcal{G}^{(\min)}=\mathcal{G}^{(K,\min)}&\equiv\{\min{\mathcal{G}}^{(K,1)},\min{\mathcal{G}}^{(K,2)},\cdots,\min{\mathcal{G}}^{(K,K)}\}\\
		\mathcal{G}^{(\min)}=\mathcal{G}^{(K,\max)}&\equiv\{\max{\mathcal{G}}^{(K,1)},\max{\mathcal{G}}^{(K,2)},\cdots,\max{\mathcal{G}}^{(K,K)}\}
	\end{array}\right\}
\end{equation}
thus, we have $|\mathcal{G}^{(\min)}|=|\mathcal{G}^{(\max)}|=K$ and the following relationships
\begin{equation*}
\bigcup_{k=1}^K\mathcal{G}^{(k)}=\mathcal{G}^{(\min)}\bigcup\left(\bigcup_{k=1}^K\mathcal{G}'^{(k)}\right)=\mathcal{G}^{(\min)}\bigcup\left(\bigcup_{k=1}^K\mathcal{G}''^{(k)}\right)\bigcup\mathcal{G}^{(\max)}
\end{equation*}
sometimes, we will also denote by $\mathcal{H}_i$ the $i$-th time point in a given grid $\mathcal{H}$, for example, $\mathcal{G}^{(\min)}_i=\min{\mathcal{G}}^{(K,i)}=\min{\mathcal{G}}^{(i)}$, $\mathcal{G}'^{(k)}_i=\mathcal{G}''^{(k)}_i$.\footnote{The time grids defined in (\ref{G_K}) and (\ref{subgrid}) depend on the tuning parameter $K$ which should be more properly written as $K_n$, however, the dependence on $n$ will be suppressed in the observational notation, for the sake of readability and notational ease. Nonetheless, it is important to keep this implicit dependence on $n$ in mind.}


In order to define some of our tests in Section \ref{test-23}, we need to introduce some shrinking moving windows and local sampling grids. Later, we will partition the fixed time interval $[0,\mathcal{T}]$ (in application, $\mathcal{T}$ could be 5 business days or some longer periods) into $r_n$ (depends on $n$ and $r_n\to\infty$) sub-intervals $(T_i,T_{i+1}]$'s, such that each $(T_{i-1},T_i]$ contains $K_n$ observations, i.e. $T_i=t_{iK_n}$, $0=T_0\le T_1\le T_2\le \cdots \le T_{r_n}$ and $r_n=\left\lfloor n/K_n\right\rfloor$. We also let $\mathcal{S}_i$ denote the shrinking sampling grid $\{t_{(i-1)K_n},\cdots,t_{iK_n}\}$ over $[T_{i-1},T_i]$, i.e., $|\mathcal{S}_i|=K_n$, $\mathcal{S}_i=\mathcal{G}\cap[T_{i-1},T_i]$, $\bigcup_{i=1}^{r_n}\mathcal{S}_i=\{t_i\in\mathcal{G}:t_i\le T_{r_n}\}$.

\subsection{Assumptions}\label{assumptions}
Beyond the model setup in subsection \ref{setup}, we have to make the following \textit{identification assumption} in order to achieve identifiability and estimability:
\begin{equation}\label{addass}
\mathrm{d}Z_t\equiv \mathrm{d}X_t=b_t\,\mathrm{d}t+\sigma_t\,\mathrm{d}W_t+J_t
\end{equation}
otherwise all the estimation methods will break down \citep{j09}. Note that under the identification assumption (\ref{addass}), we have $\{e_t\}_{t\ge0}$ and $\{\epsilon_t\}_{t\ge0}$ are identical, and there is no correlation between noise and the latent process.

\noindent As a sum-up, the following assumptions will be needed for various results:
\begin{assum}\label{diffusion}
	\textbf{Diffusion part of It\^o semimartingale.} The underlying model is (\ref{X}), $\{b_t\}_{t\ge0}$, $\{\sigma_t\}_{t\ge 0}$ and $\{W_t\}_{t\ge 0}$ are adapted, $\{b_t\}_{t\ge0}$ and $\{\sigma_t\}_{t\ge 0}$ are c\`adl\`ag processes and locally bounded.
\end{assum}
\begin{assum}\label{jump}
	\textbf{Jumps of It\^o semimartingale.} $J_t=\int_0^t\int_{\mathbb{R}}x\mathds{1}_{\{|x|\le1\}}(\mu-\nu)(\mathrm{d}s,\mathrm{d}x)+\int_0^t\int_{\mathbb{R}}x\mathds{1}_{\{|x|>1\}}\mu(\mathrm{d}s,\mathrm{d}x)$ with $\mu$ being a Poisson random measure on $\mathbb{R}_+\times \mathbb{R}$ and $\nu$ being the predictable compensator of $\mu$ in the sense that $(\mu-\nu)((0,t],A)$ is a local martingale for $\forall t>0,\forall A\in \mathcal{B}(\mathbb{R})$. One could write $\nu(\mathrm{d}t,\mathrm{d}x)=\mathrm{d}t\otimes\lambda(\mathrm{d}x)$ where $\lambda$ is a $\sigma$-finite measure on $\mathbb{R}$.
\end{assum}
\begin{assum}\label{finite.jump}
	\textbf{Finite jumps of It\^o semimartingale.} On top of \textbf{Assumption \ref{jump}}, assume $\exists$ a function $\Gamma$ on $\mathbb{R}$ such that $\int_{\mathbb{R}}\Gamma(x)\lambda(\mathrm{d}x)<\infty$ where $\Gamma\ge1$.
\end{assum}
\begin{assum}\label{identification}
	\textbf{Identifiable hidden It\^o semimartingale.} The underlying process is (\ref{X}); and we have the identifiability assumption (\ref{addass}).
\end{assum}
\begin{assum}\label{independence}
	\textbf{Conditional independence}. Conditional on the latent variable(s), the observations $Y_{t_i}$'s at different times are independent, i.e., $Y_{t_i} \indep Y_{t_j}$ for $i\ne j$. This assumption simplifies the proof substantially\footnote{An interpretation of this assumption is that the market microstructure effects occurred at different times are independent if the market participants know the latent efficient prices.}.
\end{assum}
\begin{assum}\label{bounded}
	\textbf{Locally boundedness of microstructure effect.} $\forall l>0$, and $\forall\alpha>0$, $\exists M_{(\alpha,l)}$, such that $E\left(|\epsilon_{t_i}|^{\alpha}|\mathcal{F}^{(0)}\right)\le M_{(\alpha,l)}$, when $X_{t_i}\in[-l,l]$, $\sigma^2_{t_i}\in(0,l]$.
\end{assum}
\begin{assum}\label{hfasym}
	\textbf{Possibly irregular observational grid with shrinking mesh}. The sampling times can be irregular, but independent of the latent process. The Mesh of the grid $\mathcal{G}$ goes to zero, more specifically, $\max_{1\le i\le n}\Delta t_i=O_p\left(\frac{1}{n}\right)$.
\end{assum}

Based on some of these assumptions, we provide results involving various modes of stochastic convergences. It is necessary to clarify our notation for these convergence modes: $\overset{\mathbb{P}}{\longrightarrow}$ means convergence in probability, $\overset{\mathcal{L}}{\longrightarrow}$ means convergence in law (convergence in distribution, weak convergence), $\overset{\mathcal{L}-s}{\longrightarrow}$ means stable convergence in law\footnote{The concept `` stable convergence in law'' may appear unfamiliar for some readers, please refer to \cite{mz12} or chapter 2 in \cite{jp12} for definition.}.





\section{Testing stationarity/non-stationarity: the first test}\label{Test-1}
In this article, we are considering testing the null hypothesis that the market microstructure noise is stationary:
\begin{equation*}
H_0: \{\epsilon_t\}_{t\ge0} \text{ is stationary} \longleftrightarrow H_1: \{\epsilon_t\}_{t\ge0} \text{ is non-stationary}
\end{equation*}
and we concern the following questions:
\begin{itemize}
	\item Could we find any \textit{non-parametric test} to tell the stationarity of microstructure noise?
	\item Is any stationarity test valid in terms of controlling type-I error?
	\item Is it asymptotically optimal in that its statistical power is the largest in asymptotics?
\end{itemize}

\subsection{Prelude: non-stationarity and its remedy in estimation}\label{endeff}
In this subsection, we divert our focus to the estimation of integrated volatility (or continuous quadratic variation in the terminology of stochastic calculus) using high-frequency data contaminated by (possibly non-stationary) market microstructure noise. Our first test statistic was inspired by this.

Two-time scale realized volatility estimator (TSRV) \citep{zma05} is the first consistent estimator of integrated volatility using noisy high frequency financial data. In this article, we define $[Y,Y]_{\mathcal{H}}$ as the realized variance of process $\{Y_t\}$ computed on a given sampling grid $\mathcal{H}$. The TSRV is defined as follows:
\begin{equation}\label{TSRV.def}
\widehat{\langle X,X\rangle}^{(TSRV,K_n)}_T\equiv[Y,Y]^{(avg,K_n)}_T-\frac{n-K_n+1}{nK_n}[Y,Y]_{\mathcal{G}}
\end{equation}
where, according to the notation introduced in subsection \ref{ObsNotation},
\begin{equation*}
\begin{array}{ll}
      [Y,Y]^{(avg,K_n)}&=\frac{1}{K_n}\sum^{K_n-1}_{k=0}[Y,Y]_{\mathcal{G}^{(k)}}\\
      {[Y,Y]}_{\mathcal{G}}&=\sum_{i=1}^n(Y_{t_i}-Y_{t_{i-1}})^2\\
      {[Y,Y]}_{\mathcal{G}^{(k)}}&=\sum_{t_i\in\mathcal{G}'^{(k)}}(Y_{t_i}-Y_{t_{i-K}})^2\text{, for } k=0,\cdots,K-1
\end{array}
\end{equation*}
The optimal choice for the tuning parameter is $K_n=O(n^{\frac{2}{3}})$\footnote{
A caveat in application is to choose $K_n$ such that $n-\lfloor n/K_n\rfloor K$ is sufficiently small, in order to reduce the edge effect.},
which results in the best possible order of TSRV. In the identical fashion, we can define $[\epsilon,\epsilon]_{\mathcal{G}}$, $[\epsilon,\epsilon]^{(avg,K_n)}_T$ and $[Z,Z]_T^{(avg,K_n)}$.

The intuition behind the design of $\widehat{\langle X,X\rangle}_T^{(TSRV,K_n)}$ is sub-sampling and averaging: each $[Y,Y]_{\mathcal{G}^{(k)}}$ is computed on a sparser grid hence mitigate the microstructure effect, hence their average $[Y,Y]^{(avg,K_n)}_T$ should be more closer to $\langle X,X\rangle_T$; the second term $[Y,Y]_{\mathcal{G}}$ is a good proxy to the noise variance, hence it is to offset the bias due to the noise in $[Y,Y]^{(avg,K_n)}_T$.


The TSRV was originally designed under the setting where microstructure noises are stationary; however, under non-stationary microstructure noises, TSRV has a bias term produced by edge effect due to the following lemma: 

\begin{lem}\label{lem_edge}
Under the \textbf{Assumption \ref{diffusion}, \ref{finite.jump}, \ref{identification}, \ref{independence}, \ref{bounded}, \ref{hfasym}}, we have
\begin{equation}\label{edge_avg}
[Y,Y]_T^{(avg,K_n)}-[Z,Z]^{(avg,K_n)}_T=\underbrace{\frac{2}{K_n}\sum^K_{k=1}\sum_{t_i\in\mathcal{G}''^{(k)}}g_{t_i}+\frac{1}{K_n}\sum_{t_i\in\mathcal{G}^{(\min)}}g_{t_i}+\frac{1}{K_n}\sum_{t_i\in\mathcal{G}^{(\max)}}g_{t_i}}_{\text{bias in } [Y,Y]_T^{(avg,K)} \text{ due to non-stationary noise}}+o_p(1)
\end{equation}
\end{lem}
From \textbf{Lemma \ref{lem_edge}}, we can see {the noise in each time point does not contribute equally to the bias in the averaged realized variance $[Y,Y]_T^{(avg,K_n)}$}. In the first and last $K_n$ sample points, the conditional second moments of noises are multiplied by the factor $\frac{1}{K_n}$, in contrast, the conditional second moments of noises in the middle of the sample points are multiplied by the factor $\frac{2}{K_n}$. However, the noise correction term $[Y,Y]_{\mathcal{G}}$ in (\ref{TSRV.def}) acts as if $g_{t_i}$'s all have the same contribution to the noise component in $[Z,Z]^{(avg,K_n)}_T$. The modification to the TSRV and the first two tests are motivated by the inhomogeneity of utilization of information at the two edges of the time interval $[0,T]$.

To the best of our knowledge, \cite{kl08} is the first study which considered the edge effect in TSRV due to the non-stationary microstructure noise, and they redefined the TSRV by $[Y,Y]^{(avg,K_n)}_T-\frac{n-K_n+1}{nK_n}[Y,Y]_T^{\{n\}}$ where $$[Y,Y]_T^{\{n\}}=\frac{1}{2}\left(\sum_{i=1}^{n-K_n}(Y_{t_{i+1}}-Y_{t_i})^2+\sum^{n-1}_{i=K_n}(Y_{t_{i+1}}-Y_{t_i})^2\right)$$
under a parametric model which incorporates the diurnal and endogenous measurement error.\footnote{The model upon which \cite{kl08} was based is
\begin{equation*}
\begin{array}{ll}
      \mathrm{d}X_t&=\mu_t\,\mathrm{d}t+\sigma_t\,\mathrm{d}W_t\\
      Y_{t_i}&=X_{t_i}+\epsilon_{t_i}\\
      \epsilon_{t_i}&=u_{t_i}+v_{t_i}\\
      u_{t_i}&=\delta \gamma_n(W_{t_i}-W_{t_{i-1}})\\
      v_{t_i}&=m(t_i)+n^{-\frac{\alpha}{2}}\omega(t_i)\,e_{t_i}, \alpha\in[0,1/2)
\end{array}
\end{equation*}
where $e\indep X$, i.i.d., $E(e)=0$.} In the following, we used this design to attack the non-stationarity problem under the general hidden It\^o semimartingale model given in Section \ref{model}.

In this paper, we call the new TSRV consisting of the modified version of realized variance in \cite{kl08} as ``sample-weighted TSRV'', which is defined as
\begin{equation*}
\widehat{\langle X,X\rangle}^{(WTSRV,K_n)}_T=[Y,Y]_T^{(avg,K_n)}-\frac{1}{K_n}[Y,Y]^{\{n\}}_T
\end{equation*}
The sample-weighted TSRV enjoys the following asymptotic property under the general model in Section \ref{model}:

\begin{thm}\label{thm1}
Suppose there are $n$ observations in the finite time interval $[0,T]$. When we take $K_n=cn^{2/3}$ for some constant $c$, under the \textbf{Assumption \ref{diffusion}, \ref{finite.jump}, \ref{identification}, \ref{independence}, \ref{bounded}, \ref{hfasym}}, we have
\begin{equation}
n^{1/6}\left(\widehat{\langle X,X\rangle}^{(WTSRV,K_n)}_T-[Z,Z]_T\right)\overset{\mathcal{L}-s}{\longrightarrow}\mathcal{MN}\left(0,\frac{8}{Tc^2}\int_0^Tg_t^2\mathrm{d}t+\frac{4cT}{3}\int_0^T\sigma^4_t\,\mathrm{d}t\right)
\end{equation}
\end{thm}


The theorem tells us the sample-weighted TSRV in non-stationary noise setting enjoys the same asymptotic property as those of traditional TSRV in stationary noise setting \citep{zma05,lm07,amz11}, in that the asymptotic distribution as well as the convergence rate remains unchanged; in other words, the asymptotic property of the sample-weighted TSRV is invariant with respect to non-stationary market microstructure noise.

\subsection{The first test $N(Y,K_n)^n_T$}\label{test-1}
Assuming $H_0$ is true, both of the asymptotic distributions of the original TSRV and the sample-weighted TSRV are mixed normals. So, the asymptotic distribution of the difference between two different versions (after proper scaling) is also a mixed normal. Therefore, we can test the null $H_0$ based on the asymptotic behavior of their difference $D(Y,K_n)^n_T=\sqrt{K_n}\left(\widehat{\langle X,X\rangle}^{(WTSRV,K_n)}_T-\widehat{\langle X,X\rangle}^{(TSRV,K_n)}_T\right)$, note that 
\begin{multline}\label{Di}
	D(Y,K_n)^n_T=\frac{n-2(K_n-1)}{2nK_n^{1/2}}[Y,Y]_{\mathcal{G}^{(\min)}}+\frac{n-2(K_n-1)}{2nK_n^{1/2}}[Y,Y]_{\mathcal{G}^{(\max)}}\\-\frac{K_n-1}{nK_n^{1/2}}[Y,Y]_{\mathcal{G}/(\mathcal{G}^{(\min)}\cup\mathcal{G}^{(\max)})}
\end{multline}

The first test statistic $N(Y,K_n)_T^n$ is designed as follows:
\begin{equation}\label{test1}
N(Y,K_n)^n_T\equiv\left\{\begin{array}{ll}
\frac{D(Y,K_n)^n_T}{\sqrt{\frac{1}{n}[Y;4]_{\mathcal{G}}-\frac{3}{2n^2}[Y,Y]_{\mathcal{G}}^2}}, & \text{ if } [Y;4]_{\mathcal{G}}-\frac{3}{2n}[Y,Y]_{\mathcal{G}}^2\ne0\\
0, & \text{ if } [Y;4]_{\mathcal{G}}-\frac{3}{2n}[Y,Y]_{\mathcal{G}}^2=0
\end{array}\right.
\end{equation}
where $[Y;4]_{\mathcal{G}}=\sum_{i=1}^n(Y_{t_i}-Y_{t_{i-1}})^4$ is the sample quarticity based on the observation $Y_{t_i}$'s.

Our first test statistic has the following asymptotic property:
\begin{thm}\label{thm2}
If the noise process is stationary, under the \textbf{Assumption \ref{diffusion}, \ref{jump}, \ref{independence}, \ref{bounded}, \ref{hfasym}}, as long as $K_n\to\infty$ but $K_n=o(n)$,
\begin{equation}\label{N.null}
N(Y,K_n)^n_T\overset{\mathcal{L}}{\longrightarrow}N(0,1)
\end{equation}
\end{thm}
We use this result to test the stationarity of the market microstructure noise in subsection \ref{empitest} (Figure \ref{EmpriTest1}).

The denominator of the test statistic (\ref{N.null}), namely $\frac{1}{n}[Y;4]_{\mathcal{G}}-\frac{3}{2n^2}([Y,Y]_{\mathcal{G}})^2$ is actually an estimator of $2E(\epsilon^4|\mathcal{F}^{(0)})$. This is formally introduced in (\ref{est.epsilon4}), which is not only used in the first test statistic but also used in the second test statistic in subsection \ref{TestV}. It is interesting in its own right, hence we here give the result:
\begin{lem}\label{lem_epsilon4}
	If we define the process $h_t(\omega^{(0)})\equiv E(\epsilon^4_t|\mathcal{F}^{(0)})(\omega^{(0)})$, then under the \textbf{Assumption \ref{diffusion}, \ref{jump}, \ref{independence}, \ref{bounded}, \ref{hfasym}}, we have
	\begin{equation}\label{eps4}
	\frac{1}{n}[Y;4]_{\mathcal{G}}=\frac{2}{T}\int_0^Th_t\,\mathrm{d}t+\frac{6}{T}\int_0^Tg^2_t\,\mathrm{d}t+O_p\left(\frac{1}{\sqrt{n}}\right)
	\end{equation}
\end{lem}
\begin{remk}\label{remk.est.epsilon4}
	Based on \textbf{Lemma \ref{lem_epsilon4}}, if the noise is stationary, $\frac{1}{2n}[Y,Y,Y,Y]_{\mathcal{G}}|\mathcal{F}^{(0)}\overset{\mathbb{P}}{\longrightarrow}E(\epsilon^4|\mathcal{F}^{(0)})+3\left(E(\epsilon^2|\mathcal{F}^{(0)})\right)^2$,
	so a natural estimate of $E(\epsilon^4|\mathcal{F}^{(0)})$ is
	\begin{eqnarray}\label{est.epsilon4}
		\widehat{E(\epsilon^4|\mathcal{F}^{(0)})}
		&=&\frac{1}{2n}[Y;4]_{\mathcal{G}}-\frac{3}{4n^2}[Y,Y]_{\mathcal{G}}^2
	\end{eqnarray}
\end{remk}
\begin{remk}\label{N.alternative}
	We now investigate the behavior of our first test statistic when microstructure noise is non-stationary. Since $\frac{1}{2n}[Y,Y]_{\mathcal{G}}=\frac{1}{T}\int_0^Tg_t\,\mathrm{d}t+o_p(1)$, hence
	\begin{equation}\label{asym.epsilon4}
		\widehat{E(\epsilon^4|\mathcal{F}^{(0)})}\overset{\mathbb{P}}{\longrightarrow}
		D_T\equiv\left\{\begin{array}{ll}
		   E(\epsilon^4|\mathcal{F}^{(0)}) & \{\epsilon_t\}_{t\ge0} \text{ is stationary}\\
		   \frac{1}{T}\int_0^Th_t\,\mathrm{d}t+\frac{3}{T}\int_0^Tg^2_t\,\mathrm{d}t-\frac{3}{T^2}\left(\int_0^Tg_t\,\mathrm{d}t\right)^2    & \{\epsilon_t\}_{t\ge0} \text{ is non-stationary}
		\end{array}\right.
	\end{equation}
	Since we assume local boundedness of noise variance, $D_T$ is finite almost surely, regardless of noise stationarity. From the proof in subsection \ref{prfthm1}, we know
	\begin{equation}\label{tsdif2}
		N(Y,K_n)^n_T=\sqrt{K_n}\times\frac{\overline{g}^{(\text{start})}+\overline{g}^{(\text{end})}-2\overline{g}^{(\text{middle})}}{D_T}+O_p(1)
	\end{equation}
		where
	\begin{equation}\label{avgeps}
	\begin{array}{ll}
	\overline{g}^{(\text{start})}=&\frac{1}{K_n}\sum_{t_i\in\mathcal{G}^{(\max)}}g_{t_i}\\
	\overline{g}^{(\text{end})}=&\frac{1}{K_n}\sum_{t_i\in\mathcal{G}^{(\max)}}g_{t_i}\\
	\overline{g}^{(\text{middle})}=&\frac{1}{n+1-2K_n}\sum^{K_n}_{k=1}\sum_{t_i\in\mathcal{G}''_k}g_{t_i}
	\end{array}
	\end{equation}
	Since $K_n=O\left(n^{2/3}\right)$ in our setup, $N(Y,K_n)_T$ explodes when the noise is not stationary. Thus, the type-\Rmnum{2} error of this test is asymptotically negligible.
\end{remk}

Followed from \textbf{Theorem \ref{thm2}} and \textbf{Remark \ref{N.alternative}}, we have
\begin{corol}\label{corol.powerN}
	Assume $\{g_t\}_{t\ge0}$ and $\{h_t\}_{t\ge0}$ are c\`adl\`ag processes\footnote{The term ``c\`adl\`ag'' (French acronym of ``continue à droite, limite à gauche'') describes the property of a function that is everywhere right-continuous and has left limits everywhere, for example, a Brownian motion (sample path are continuous almost surly), L\'evy processes (countably many jump discontinuities).} on $[0,T]$ with $0$, $T$ being continuity points almost surly, additional we have the \textbf{Assumption \ref{diffusion}, \ref{jump}, \ref{independence}, \ref{bounded}, \ref{hfasym}}, and let $K_n\to\infty$, $K_n/n\to0$. When the noise process is non-stationary,
	\begin{equation}
		N(Y,K_n)^n_T-\sqrt{K_n}\times\frac{g_0+g_T-\frac{1}{T}\int_0^T\,g_t\mathrm{d}t}{\frac{1}{T}\int_0^Th_t\,\mathrm{d}t+\frac{3}{T}\int_0^Tg^2_t\,\mathrm{d}t-\frac{3}{T^2}\left(\int_0^Tg_t\,\mathrm{d}t\right)^2}\overset{\mathcal{L}}{\longrightarrow}N(0,1)
	\end{equation}
    on the event that $\int_0^Th_t\,\mathrm{d}t+3\int_0^Tg^2_t\,\mathrm{d}t-\frac{3}{T}\left(\int_0^Tg_t\,\mathrm{d}t\right)^2\ne0$.
\end{corol}

\begin{remk}\label{rv.in.test}
    
    
    The test statistic $N(Y,K_n)^n_T$ can disclose potential non-stationarity in the market microstructure noise via two edges of the mesh $\mathcal{G}^{(\min)}$, $\mathcal{G}^{(\max)}$ and the middle of the mesh $\mathcal{G}/(\mathcal{G}^{(\min)}\cup \mathcal{G}^{(\max)})$. We can show there are, in latter subsections, schemes which are able not only to reflect the heterogeneity between two edges and the middle, but also to capture almost all of the information about the non-stationarity in the data, however, inevitably with more computational cost. We will discuss these schemes in Section \ref{test-23}.	
\end{remk}


\section{The second and third tests}\label{test-23}
\subsection{The general idea}\label{SecT2}
The second and third tests are designed as an attempt to effectively utilize all the information relevant to noise stationarity contained in the data, in contrast to the first test $N(Y,K_n)^n$ (see \textbf{Remark \ref{rv.in.test}}). The basic idea of the second and third tests is to conduct local tests on sub-intervals and then combine evidences from all the local tests.

To straighten the idea, recall the observational notation in subsection \ref{ObsNotation}, we partition the fixed time interval $[0,\mathcal{T}]$ into $r_n$ sub-intervals $(T_i,T_{i+1}]$'s, such that each $(T_{i-1},T_i]$ contains $K_n$ observations. Similar to the definition of the first test statistic, but instead of the whole interval $[0,\mathcal{T}]$, the second test uses local test statistic defined on a moving window of the form $(T_{i-1},T_{i-1+s_n}]\subset[0,\mathcal{T}]$ with a suitable window length $s_n$ (in terms of the number of subintervals $(T_i,T_{i+1}]$'s): 
\begin{equation}
D(Y,K_n,s_n)^n_i\equiv\frac{(s_n-2)K_n+2}{2s_nK_n^{3/2}}\left([Y,Y]_{\mathcal{G}_i}+[Y,Y]_{\mathcal{G}_{i+s_n}}\right)-\frac{K_n-1}{s_nK_n^{3/2}}[Y,Y]_{\cup_{k=2}^{s_n-1}\mathcal{G}_{i+k}}
\end{equation}

Then, we use the overlapping window to calculate the quantity $U(Y,K_n,u)^n_{\mathcal{T}}$, which depends on the process $Y$, the stage of statistical experiment $n$, tuning parameter $K_n$ and $s_n<\left\lfloor\frac{n}{K_n}\right\rfloor$, and a number $u>0$:
\begin{equation}\label{gene.U1}
    U(Y,K_n,s_n,u)^n_{\mathcal{T}}\equiv \frac{1}{\lfloor n/K_n\rfloor-s_n+1}\sum^{\lfloor n/K_n\rfloor-s_n+1}_{i=1}\left|D(Y,K_n,s_n)^n_i\right|^u
\end{equation}
Similarly, we also define a quantity based on non-overlapping windows:
\begin{equation}\label{gene.U2}
	U'(Y,K_n,s_n,u)_\mathcal{T}^n\equiv\frac{1}{\lfloor n/(s_nK_n)\rfloor}\sum^{\lfloor n/(s_nK_n)\rfloor}_{i=1}\left|D(Y,K_n)_{(i-1)s_n+1}\right|^u
\end{equation}

\subsection{The second test $V(Y,K_n,s_n,2)^n_{\mathcal{T}}$}\label{TestV}
We designed our second test statistic by
\begin{equation}\label{test2}
V(Y,K_n,s_n,2)^n_{\mathcal{T}}\equiv\left\{\begin{array}{ll}
	\sqrt{n/K_n}\left(U(Y,K_n,s_n,2)^n_{\mathcal{T}}-\frac{1}{n}[Y;4]_{\mathcal{G}}-\frac{3}{2n^2}[Y,Y]_{\mathcal{G}}^2\right)/\widehat{\eta}, & \text{ if } \widehat{\eta}\ne0\\
	0, & \text{ if } \widehat{\eta}=0
	\end{array}\right.
\end{equation}
where
\begin{equation}\label{est.eta}
	\widehat{\eta}^2
	=\frac{6}{n^2}[Y;4]_{\mathcal{G}}^2-\frac{21}{n^3}[Y;4]_{\mathcal{G}}\cdot[Y,Y]_{\mathcal{G}}^2+\frac{39}{2n^4}[Y,Y]_{\mathcal{G}}^4
\end{equation}
We have the following result in regard to the asymptotic property of $V(Y,K_n,2)^n_{\mathcal{T}}$:
\begin{thm}\label{thm3}\textbf{\textit{($V(Y,K_n,s_n,2)_{\mathcal{T}}^n$ under the null)}}
Under the \textbf{Assumption \ref{diffusion}, \ref{jump}, \ref{independence}, \ref{bounded}, \ref{hfasym}}, assume the noise process is stationary, and choose the tuning parameters such that $K_n/n\to0$, $K_n/n^{1/3}\to\infty$, $s_n\to\infty$, $s_n$. 
Then the test statistic $V(Y,K_n,2)^n_{\mathcal{T}}$ has the following asymptotic property:
\begin{equation}\label{asym.V}
V(Y,K_n,s_n,2)_{\mathcal{T}}^n\overset{\mathcal{L}}{\longrightarrow}N(0,1)
\end{equation}
on the event that $E(\epsilon^4|\mathcal{F})^2-E(\epsilon^4|\mathcal{F})E(\epsilon^2|\mathcal{F})^2+E(\epsilon^2|\mathcal{F})^4\ne0$.
\end{thm}
We use this result to test the stationarity of the market microstructure noise in subsection \ref{empitest} (Figure \ref{EmpriTest2}).

We can also define another quantity $V^{\prime}(Y,K_n,s_n,2)_\mathcal{T}^n$ based on non-overlapping intervals
\begin{equation*}
V'(Y,K_n,s_n,2)^n_{\mathcal{T}}\equiv\left\{\begin{array}{ll}
\sqrt{n/(s_nK_n)}\left(U'(Y,K_n,s_n,2)^n_{\mathcal{T}}-\frac{1}{n}[Y;4]_{\mathcal{G}}-\frac{3}{2n^2}[Y,Y]_{\mathcal{G}}^2\right)/\widehat{\eta}, & \text{ if } \widehat{\eta}\ne0\\
0, & \text{ if } \widehat{\eta}=0
\end{array}\right.
\end{equation*}
Following from \textbf{Theorem \ref{thm3}}, the asymptotic property of $V'(Y,K_n,s_n,2)_\mathcal{T}^n$ can be derived.
\begin{corol}\label{corollary1}
Under the same conditions as in \textbf{Theorem \ref{thm3}}, assume the noise is stationary:
\begin{equation}
V'(Y,K_n,s_n,2)_{\mathcal{T}}^n\overset{\mathcal{L}}{\longrightarrow}N(0,1)
\end{equation}
on the event that $E(\epsilon^4|\mathcal{F})^2-E(\epsilon^4|\mathcal{F})E(\epsilon^2|\mathcal{F})^2+E(\epsilon^2|\mathcal{F})^4\ne0$.
\end{corol}
\begin{remk}
	It is a little bit surprising when we compare \textbf{Corollary \ref{corollary1}} with \textbf{Theorem \ref{thm3}}, since the limiting mixed normals of $U(Y,K_n,s_n,2)^n_{\mathcal{T}}$ and $U'(Y,K_n,s_n,2)^n_{\mathcal{T}}$ have the same asymptotic variance which can be consistently estimated by $\widehat{\eta}$, though the convergence rate of the former is lower. However, the results only demonstrate the limiting behaviors. $V'(Y,K_n,s_n,2)^n_\mathcal{T}$ required less computation, while $V(Y,K_n,s_n,2)^n_\mathcal{T}$ is more accurate in terms of asymptotic approximation because of its higher rate of convergence.
\end{remk}

\subsection{The third test $\overline{V}(Y,K_n,2)^n$}\label{equi.test}
There is a moderate edge effect in the second test statistic (\ref{test2}) (coming from the first $s_nK_n$ and the last $s_nK_n$ observations). Motivated by \textbf{Remark \ref{rv.in.test}} regarding the first test statistic (\ref{test1}), we can design a complementary test statistic $\overline{V}(Y,K_n,2)^n$ (defined by (\ref{test3})) with the same asymptotic property with $V(Y,K_n,s_n,2)^n$ when the noise is stationary, yet has a smaller edge effect in finite sample. However, we should keep $V(Y,K_n,s_n,2)^n$ in our toolbox - although $\overline{V}(Y,K_n,2)^n$ offers better approximation when noise is stationary, we will see $V(Y,K_n,s_n,2)^n$ has more statistical power as indicated in Figure \ref{compare_NVV}.

The key component of the third test statistic is
\begin{equation}\label{gene.U3}
   \overline{U}(Y,K_n,2)^n_\mathcal{T}\equiv\frac{1}{4n}\sum^{\lfloor n/K_n\rfloor-1}_{i=1}\left|[Y,Y]_{\mathcal{S}_{i+1}}-[Y,Y]_{\mathcal{S}_i}\right|^2
\end{equation}
where each $\mathcal{S}_i$ denotes the shrinking sampling grid $\{t_{(i-1)K_n},\cdots,t_{iK_n}\}$ over $[T_{i-1},T_i]$ (recall the observational notation in subsection \ref{ObsNotation}), and $[Y,Y]_{\mathcal{S}_i}$ is the realized variance of process $Y$ on the grid $\mathcal{S}_i$. Our third test statistic is defined as
\begin{equation}\label{test3}
	\overline{V}(Y,K_n,2)_\mathcal{T}^n\equiv\left\{\begin{array}{ll}
	\sqrt{n/K_n}\left(\overline{U}(Y,K_n,2)^n_{\mathcal{T}}-\frac{1}{n}[Y;4]_{\mathcal{G}}-\frac{3}{2n^2}[Y,Y]_{\mathcal{G}}^2\right)/\widehat{\eta}, & \text{ if } \widehat{\eta}\ne0\\
	0, & \text{ if } \widehat{\eta}=0
	\end{array}\right.
\end{equation}
where $\widehat{\eta}$ was defined in (\ref{est.eta}).

\begin{thm}\label{thm4} \textbf{\textit{($\overline{V}(Y,K_n,2)_\mathcal{T}^n$ under the null)}}
Under the \textbf{Assumption \ref{diffusion}, \ref{jump}, \ref{independence}, \ref{bounded}, \ref{hfasym}}, assume the noise process is stationary, suppose $K_n\to\infty$, $K_n/n\to0$ 
then the test statistic $\overline{V}(Y,K_n,2)^n_\mathcal{T}$ has the following asymptotic property:
\begin{equation}
\overline{V}(Y,K_n,2)_\mathcal{T}^n\overset{\mathcal{L}}{\longrightarrow}\mathcal{N}(0,1)
\end{equation}
on the event that $E(\epsilon^4|\mathcal{F})^2-E(\epsilon^4|\mathcal{F})E(\epsilon^2|\mathcal{F})^2+E(\epsilon^2|\mathcal{F})^4\ne0$.
\end{thm}

\subsection{Optimal statistical powers}\label{VV.alternative}
How $V(Y,K_n,s_n,2)^n$ and $\overline{V}(Y,K_n,2)^n$ behave when the noise is non-stationary determine their statistical powers. If the test statistics tend to be large when the microstructure noise is non-stationary, they can easily detect non-stationarity.

The behaviors of $U(Y,K_n,s_n,2)^n$ and $\overline{U}(Y,K_n,2)^n$ largely indicate the behaviors of $V(Y,K_n,s_n,2)^n$ and $\overline{V}(Y,K_n,2)^n$. We investigate in this subsection the asymptotic behaviors of $U(Y,K_n,s_n,2)^n$ and $\overline{U}(Y,K_n,2)^n$ when microstructure noise is non-stationary under a slightly strengthened setting, we need 2 more assumptions on top of those assumptions in subsection \ref{assumptions}:
\begin{assum}\label{regular}
	\textbf{Regular sampling.} The sample grid is equi-distant over the interval $[0,\mathcal{T}]$.
\end{assum}
\begin{assum}\label{Ito.variance}
	\textbf{Noise variance process is It\^o.} $\{g_t\}_{t\ge0}$ is an It\^o diffusion (in time):
	\begin{equation}\label{g.Ito}
	g_t=\int_0^t\mu^{(g)}_s\,\mathrm{d}s+\int_0^t\sigma^{(g)}_s\,\mathrm{d}B_s
	\end{equation}
	where $\{\mu^{(g)}_t\}_{t\ge0}$ is locally bounded, optional and c\`adl\`ag, $\{B_t\}_{t\ge0}$ is a standard Brownian motion, $\{\sigma^{(g)}_t\}_{t\ge0}$ is a locally bounded It\^o diffusion.
\end{assum}
As described in subsection \ref{ObsNotation}, we partition the whole time interval into $r_n$ disjoint sub-intervals $(T_{i-1},T_i]$ for $i=1,2,\cdots,r_n$ such that in each sub-interval we have $K_n$ observations, particularly we have $T_0=0$ and $\mathcal{T}-T_{r_n}=o(1)$, $r_nK_n/n\to1$. Since \textbf{Assumption \ref{regular}} is assumed, we let $\Delta T=T_i-T_{i-1},\,\forall i=1,2,\cdots,r_n$.

\begin{thm}\label{thm5.0}\textbf{\textit{($U(Y,K_n,s_n,2)^n$ under the alternative)}}
	Assume \textbf{Assumption \ref{diffusion}, \ref{jump}, \ref{independence}, \ref{bounded}, \ref{hfasym}} as well as \textbf{Assumption \ref{regular}, \ref{Ito.variance}}. Let 
	$\frac{K_n}{n^{1/2}}\to\infty$, $s_n\to\infty$ and $\frac{s_nK_n}{n}\to0$ but $\frac{s_n^{2/5}K_n}{n^{3/5}}\to\infty$. Then, we have
	\begin{equation}\label{UH0}
	\sqrt{\frac{n}{K_n}}\left(\frac{n}{s_nK_n^2}U(Y,K_n,s_n,2)^n_{\mathcal{T}}-E^{(1)}_n-E^{(2)}_n-E^{(3)}_n\right)\overset{\mathcal{L}-s}{\longrightarrow}\mathcal{MN}\left(0,\frac{2{\mathcal{T}}}{9}\int_0^{\mathcal{T}}(\sigma^{(g)}_t)^4\,\mathrm{d}t\right)
	\end{equation}
	on the event that $\{\sigma^{(g)}_t\}_{t\in[0,\mathcal{T}]}$ is non-vanishing, where
	\begin{eqnarray*}
		E^{(1)}_n&=&\frac{(s_n-2)^2}{3s_n^2}\int_0^\mathcal{T}(\sigma^{(g)}_t)^2\,\mathrm{d}t\\
		E^{(2)}_n&=&-\frac{s_nK_n}{6n}\left[(\sigma^{(g)}_0)^2+(\sigma^{(g)}_\mathcal{T})^2\right]\mathcal{T}\\
		E^{(3)}_n&=&\frac{2n}{s_nK_n^2{\mathcal{T}}}\int_0^{\mathcal{T}}h_t\,\mathrm{d}t, \hspace{5mm} h_t(\omega^{(0)})\equiv E(\epsilon^4_t|\mathcal{F}^{(0)})(\omega^{(0)})
	\end{eqnarray*}
\end{thm}

\begin{thm}\label{thm5}\textbf{\textit{($\overline{U}(Y,K_n,2)^n$ under the alternative)}}
	Assume \textbf{Assumption \ref{diffusion}, \ref{jump}, \ref{independence}, \ref{bounded}, \ref{hfasym}} as well as \textbf{Assumption \ref{regular}, \ref{Ito.variance}}.
	Let $\frac{K_n}{n^{1/2}}\to\infty$ and $\frac{K_n}{n^{2/3}}\to0$. 
	Then, we have
	\begin{equation}\label{UH}
	\sqrt{\frac{n}{K_n}}\left(\frac{n}{K_n^2}\overline{U}(Y,K_n,2)_{\mathcal{T}}^n-\overline{E}^{(1)}-\overline{E}^{(2)}_n\right)\overset{\mathcal{L}-s}{\longrightarrow}\mathcal{MN}\left(0,\frac{2{\mathcal{T}}}{3}\int_0^{\mathcal{T}}(\sigma^{(g)}_t)^4\,\mathrm{d}t\right)
	\end{equation}
	on the event that $\{\sigma^{(g)}_t\}_{t\in[0,\mathcal{T}]}$ is non-vanishing, where
	\begin{eqnarray*}
		\overline{E}^{(1)}  &=&\frac{2}{3}\int_0^\mathcal{T}(\sigma^{(g)}_t)^2\,\mathrm{d}t\\
		\overline{E}^{(2)}_n&=&\frac{2n}{K_n^2{\mathcal{T}}}\int_0^{\mathcal{T}}h_t\,\mathrm{d}t, \hspace{5mm} h_t(\omega^{(0)})\equiv E(\epsilon^4_t|\mathcal{F}^{(0)})(\omega^{(0)})
	\end{eqnarray*}
\end{thm}
	\textbf{Theorem \ref{thm3}} and \textbf{\ref{thm4}} provide us the asymptotic distributions of $V(Y,K_n,s_n,2)^n$ and $\overline{V}(Y,K_n,2)^n$ under the stationarity hypothesis, which aid us to control the type-\Rmnum{1} error. On the other hand, \textbf{Theorem \ref{thm5.0}} and \textbf{\ref{thm5}} reveal asymptotic behaviors of $V(Y,K_n,s_n,2)^n$ and $\overline{V}(Y,K_n,2)^n$ under the alternative hypothesis by respectively analyzing $U(Y,K_n,s_n,2)^n$ and $\overline{U}(Y,K_n,2)^n$.
Since the moments of noise are locally bounded, $\frac{1}{n}[Y;4]_{\mathcal{G}}-\frac{3}{2n^2}[Y,Y]_{\mathcal{G}}^2$ and $\widehat{\eta}$ are always finite. Following \textbf{Theorem \ref{thm5.0}}, \textbf{\ref{thm5}}, we have the following corollary:
\begin{corol}\label{corol.powerVVbar}
	Assume \textbf{Assumption \ref{diffusion}, \ref{jump}, \ref{independence}, \ref{bounded}, \ref{hfasym}} as well as \textbf{Assumption \ref{regular}, \ref{Ito.variance}}. Adopt the choice of tuning parameters in \textbf{Theorem \ref{thm5.0}}, \textbf{\ref{thm5}}, we have
	\begin{eqnarray*}
		V(Y,K_n,s_n,2)^n       &=&O_p\left(s_n\cdot\frac{K_n}{n^{1/2}}\cdot K_n^{1/2}\right)\\
		\overline{V}(Y,K_n,2)^n&=&O_p\left(\frac{K_n}{n^{1/2}}\cdot K_n^{1/2}\right)
	\end{eqnarray*}
	on the event $\widehat{\eta}\ne0$. Besides, we have
	$$N(Y,K_n)^n=O_p(K_n^{1/2})$$
	on the event $[Y;4]_{\mathcal{G}}-\frac{3}{2n}[Y,Y]_{\mathcal{G}}^2\ne0$.
\end{corol}
Recall the choices of tuning parameter, $K_n\to\infty$ for $N(Y,K_n)^n$, $K_n/n^{1/2}\to\infty$ for $V(Y,K_n,s_n,2)^2$ and $\overline{V}(Y,K_n,2)^2$, their asymptotic powers attain the optimal. As in finite samples, $\overline{V}(Y,K_n,2)^n$ has more statistical power than $N(Y,K_n)^n$ by a factor of magnitude $K_n/n^{1/2}$; $V(Y,K_n,s_n,2)^n$ is more powerful than $\overline{V}(Y,K_n,2)^n$ by a factor of magnitude $s_n$.

\section{A user's guide of stationarity tests}\label{discuss}
We currently have 3 complementary tests, namely $N(Y,K_n)^n$, $V(Y,K_n,s_n,2)^n$ and $\overline{V}(Y,K_n,2)^n$, each test has its own advantages as well as disadvantages. In this subsection, we are going to discuss their strength and weakness, and how to choose the optimal test for different circumstances.
\begin{enumerate}
	\item[(1)] The first test $N(Y,K_n)^n$ divides the sample into 3 periods and compares the noise level in the middle with those on the edges. If we are interested in possible daily diurnal noise patterns, for example, let us test whether the noise level is higher in opening and closing trading hours, the best choice is to apply $N(Y,K)^n$ on 1-day high-frequency data. However, $N(Y,K_n)^n$ is not sensitive to local changes, for example, in case of a periodic change and the data sample covers several cycles, $N(Y,K_n)^n$ will likely misjudge the non-stationarity fact;
	\item[(2)] The second test uses moving local windows each containing $s_nK_n$ observation, and compares noise levels in the edges and the middle of each local window; the third test also uses local windows but compares the noise level in one local widow with those in neighboring windows. Because they conduct test locally and aggregate local evidences, $V(Y,K_n,s_n,2)^n$ and $\overline{V}(Y,K_n,2)^n$ are more powerful in detecting local noise changes which $N(Y,K_n)^n$ could probably ignore. However, if the noise transition goes very smoothly but there is a systematic paradigm shift on a global scale, $V(Y,K_n,s_n,2)^n$ and $\overline{V}(Y,K_n,2)^n$ might lead to false stationarity conclusion;
	\item[(3)] As said in subsection \ref{equi.test}, $V(Y,K_n,s_n,2)^n$ has a smaller edge effect than $\overline{V}(Y,K_n,2)^n$ hence $\overline{V}(Y,K_n,2)^n$ is more a accurate test under the null hypothesis; whereas $V(Y,K_n,2)^n$ enjoys larger statistical power (the lower right panel in Figure \ref{compare_NVV}). The intuition is that by construction the focus of $\overline{V}(Y,K_n,2)^n$ is too local although it results in the smaller edge effect, which turns into its disadvantage when the noise is non-stationary. 
\end{enumerate}
As a simulation comparison, Figure \ref{compare_NVV} shows averaged p-values computed from simulated 1-day/multi-day data with stationary/non-stationary noises. Figure \ref{ROC} shows their ROC curves. The simulation configuration is described in subsection \ref{simul.config}, and each p-values shown is the average of 3000 Monte Carlo samples.
\begin{figure}[!]
	\centering
    \caption{Averaged p-values of the 3 tests proposed.}
	\includegraphics[width=1.1\textwidth]{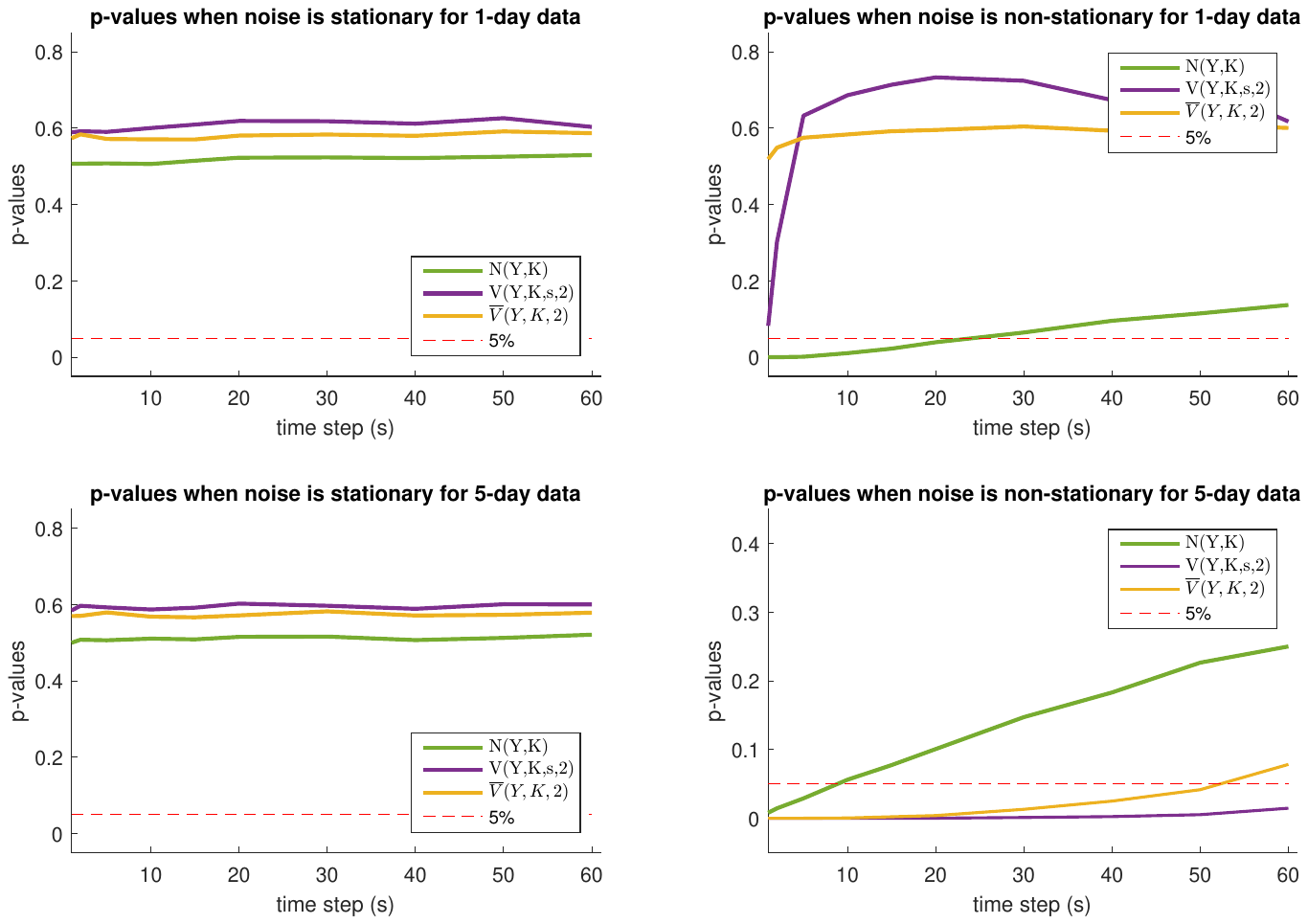}
	\label{compare_NVV}
	\floatfoot{These plots show some properties of the tests we proposed:\\
		1. For type-I error, all the tests can control their type-I errors pretty well regardless of the time span of the data, in that $N(Y,K)^n$ can very accurately control its type-I error in finite samples, $V(Y,K,s,2)^n$ and $\overline{V}(Y,K,2)^n$ are more conservative in the sense that their effective type-I error is smaller than specified;\\
		2. For type-II error or statistical power, only $N(Y,K)^n$ performs satisfactorily on 1-day data, meanwhile, $V(Y,K,s,2)$ and $\overline{V}(Y,K,2)^n$ are much better when applied to multi-day data in terms of their larger statistical powers and faster convergence rates. Consistent with \textbf{Corollary \ref{corol.powerVVbar}}, $V(Y,K_n,s_n,2)^n$ has a better statistical power in finite sample.}
\end{figure}
\begin{figure}
	\caption{ROC curves}
	\includegraphics[width=.53\textwidth]{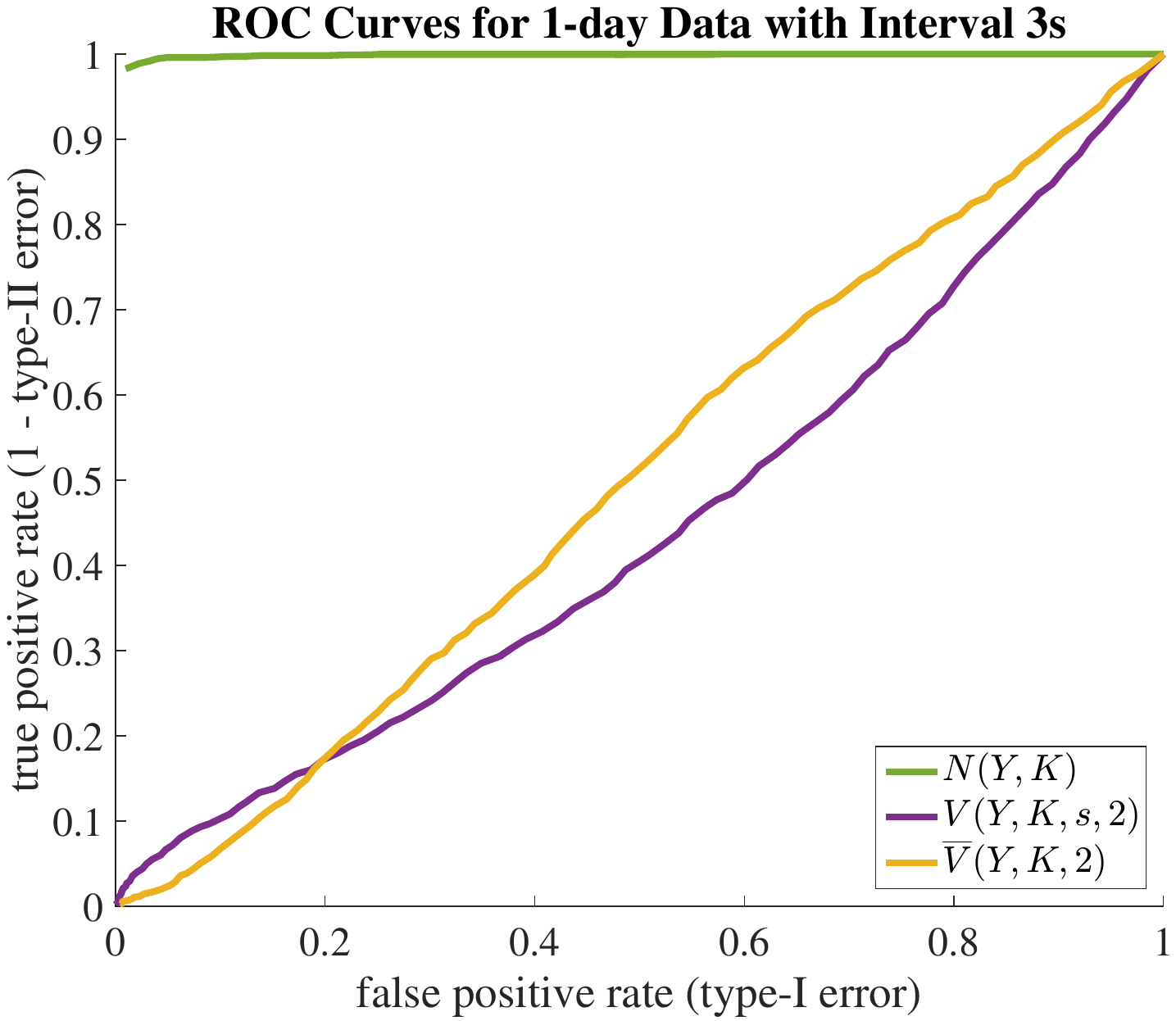}\includegraphics[width=.53\textwidth]{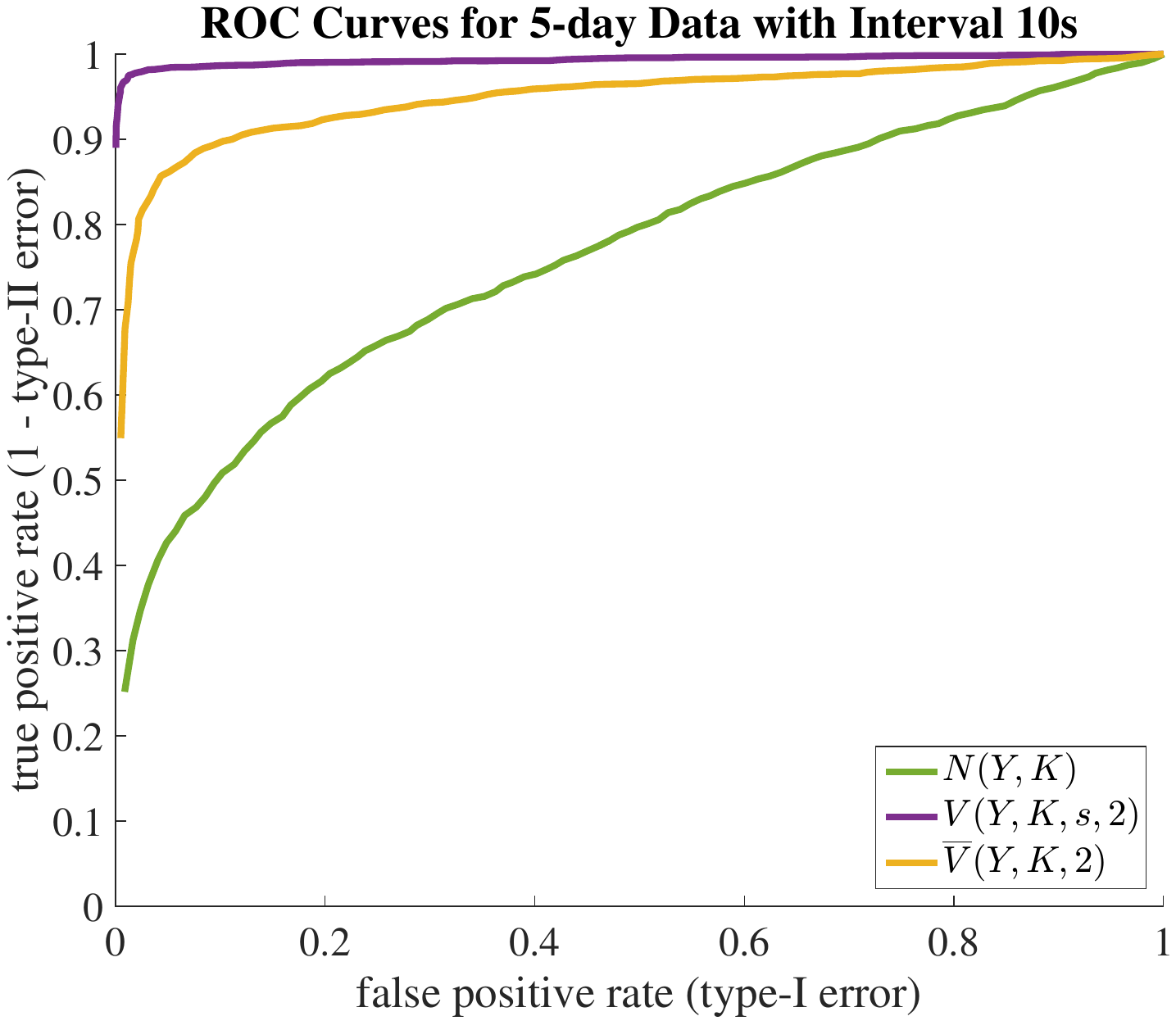}
	\label{ROC}
	\floatfoot{ROC curves show how type-II changes as type-I error varies. Once again, the ROC curves indicates that when microstructure noise exhibits a daily diurnal pattern, $N(Y,K_n)^n$ is optimal for 1-day data, $V(Y,K_n,s_n,2)^n$ is optimal for multi-day data.}
\end{figure}

As a summary, we list different considerations about the optimal choice of these tests in Table \ref{table.swt}. We suggest some choices of the tuning parameters ($K_n$, $s_n$) to balance various errors in the high-frequency approximation. Table \ref{table.rpt} shows the convergence rates and statistical powers of our tests under the suggested tuning parameters.

\begin{table}[!]
	\centering
	\begin{tabular}{l|lll}
		Test Statistics & $N(Y,K_n)^n_T$   & $V(Y,K_n,s_n,2)^n_T$ & $\overline{V}(Y,K_n,2)^n_T$\\
		\hline
		type-\Rmnum{1} error control & most accurate & least accurate    & moderately accurate\\
		Strength in detection\footnote{Evaluated in terms of statistical power.}& global change & local change   & local change (suboptimal)\\
		Length requirement& 1/multi-day data & multi-day data & multi-day data\\
		Frequency requirement\footnote{The minimal thresholds are expressed in terms of averaged time gap between consecutive observations. They are estimated from our simulation whose configuration is fairly realistic (subsection \ref{simul.config}).} & $\le$ 20s/$\le$10s & $\le$ 60s & $\le$ 50s\\
		Computational cost& relative small   & relative large       & relative large
	\end{tabular}
	\caption{Strengths and weakness of the tests}\label{table.swt}
\end{table}

\begin{table}[!]
	\centering
	\begin{tabular}{c|ccc}
		      &  $N(Y,K_n)^n$    &    $V(Y,K_n,s_n,2)^n$   &   $\overline{V}(Y,K_n,2)^n$\\
		\hline
		$K_n$ &   $n^{1/2}$      &    $n^{2/3}$            &   $n^{7/12}$\\
		$s_n$ &   ---            &    $n^{1/6}$            &   ---\\
		\hline
		rates under $H_0$       & $n^{1/4}$  & $n^{1/6}$ & $n^{5/24}$\\
		magnitudes under $H_1$  & $n^{1/4}$  & $n^{2/3}$ & $n^{3/8}$
	\end{tabular}
    \caption{Rates of convergence and statistical powers of the tests}\label{table.rpt}
\end{table}


\section{Measuring aggregate liquidity risks}\label{alternative}
\subsection{A notion of ``aggregate liquidity risk''}\label{QVg}
On one hand, $[g,g]$ as the quadratic variation of $\{g_t\}_{t\ge 0}$ over $(0,\mathcal{T}]$ is a reasonable measure of the ``aggregate'' variation of the process $\{g_t\}_{t\ge 0}$. On the other hand, microstructure noise variance $g_{t_i}$ is a measure of market quality \citep{h93}, or more specifically, market liquidity \citep{ay09}. Hence, it should not be utterly unreasonable to interpret $[g,g]_{\mathcal{T}}$ as ``\textit{aggregate liquidity risks}''. In this section, we are going to define a notion of ``\textit{aggregate liquidity risks}'' and provide an estimator with an associated CLT.


\subsection{Estimating aggregated liquidity risk $[g,g]$}
Recall \textbf{Theorem \ref{thm5}} and note that $\frac{n}{K_n^2}\int_0^{\mathcal{T}}h_t\,\mathrm{d}t\rightarrow0$, $\frac{3n}{2K_n^2}\overline{U}(Y,K_n,2)^n_{\mathcal{T}}$ is a consistent estimator of $[g,g]_{\mathcal{T}}$, i.e.
\begin{equation}\label{[g,g].est}
	\frac{3n}{2K_n^2}\overline{U}(Y,K_n,2)^n_\mathcal{T}\overset{\mathbb{P}}{\longrightarrow}[g,g]_\mathcal{T}
\end{equation}

However, We can rewrite (\ref{UH}) as
\begin{equation*}
\sqrt{\frac{n}{K_n}}\left(\frac{3n}{2K_n^2}\overline{U}(Y,K_n,2)_{\mathcal{T}}^n-[ g,g]_{\mathcal{T}}\right)-\frac{3n^{3/2}}{K_n^{5/2}{\mathcal{T}}}\int_0^{\mathcal{T}}h_t\,\mathrm{d}t\overset{\mathcal{L}-s}{\longrightarrow}\mathcal{MN}\left(0,\frac{3{\mathcal{T}}}{2}\int_0^{\mathcal{T}}(\sigma^{(g)}_t)^4\,\mathrm{d}t\right)
\end{equation*}
depending the relation between the number of blocks and the number of observations within each block, we have different second-order properties. If we let $\frac{K_n}{n^{3/5}}\to\infty$, 
we have an unbiased central limit theorem for estimating $[g,g]$. Otherwise, in case $K_n\asymp n^{3/5}$ (or $K_n/n^{3/5}\to0$), we have a CLT with a finite (or diverging) bias. 

\begin{corol}\label{corollary2}
	\footnote{\textbf{Corollary \ref{corollary2}} in some sense is an extension of the ``integral-to-spot device'' in \cite{mz16a}: for a semimartingale \{$\theta_t$\} on $[0,\mathcal{T}]$, let $\Theta_{(T_i,T_{i+q}]}=\int_{T_i}^{T_{i+q}}\theta_t\,\mathrm{d}t$, and $\mathrm{QV}_q(\Theta)=\frac{1}{q}\sum^{r_n-q}_{i=q}\left(\theta_{(T_i,T_{i+q}]}-\theta_{(T_{i-q},T_i]}\right)^2$, then under some regularity conditions (to guarantee standard stable convergence plus additional restriction on edge effects), as $q\to\infty$ and $q\Delta T\to 0$, $$\frac{1}{(q\Delta T)^2}\mathrm{QV}_q(\Theta)\overset{\mathbb{P}}{\rightarrow}\frac{2}{3}[\theta,\theta]_{\mathcal{T}-}$$
	Define $G_i\equiv \int^{T_i}_{T_{i-1}}g_t\,\mathrm{d}t$ and $\widehat{G}_i\equiv \frac{\Delta T}{2K_n}[Y,Y]_{\mathcal{S}_i}$. Under some regularity conditions, according to the ``integral-to-spot device'' in \cite{mz16a}
	\begin{equation*}
	\frac{3}{2(\Delta T)^2}\sum^{\lfloor n/K_n\rfloor}_{i=1}\left(G_i-G_{i-1}\right)^2\overset{\mathbb{P}}{\longrightarrow}[g,g]_{\mathcal{T}}
	\end{equation*}
	However, we do not know the true values of $G_i$'s in application, after swapping $G_i$ for $\widehat{G}_i$,
	\begin{equation}\label{aggrisk}
	\frac{3}{2(\Delta T)^2}\sum^{\lfloor n/K_n\rfloor}_{i=1}\left(\widehat{G}_i-\widehat{G}_{i-1}\right)^2\overset{\mathbb{P}}{\longrightarrow}[g,g]_{\mathcal{T}-}+\text{(possibly additional terms)}
	\end{equation}
	Note that $\frac{n}{K_n^2}\overline{U}(Y,K_n,2)_{\mathcal{T}}^n=\frac{1}{(\Delta T)^2}\sum^{\lfloor n/K_n\rfloor-1}_{i=1}\left(\widehat{G}_i-\widehat{G}_{i-1}\right)^2$, \textbf{Corollary \ref{corollary2}} reveals the possible additional terms and provides the central limit theorem associated with (\ref{aggrisk}). Upon choosing $K_n$ appropriately, the additional terms on the right side of (\ref{aggrisk}) is zero and we have an unbiased central limit theorem.
}
	Assume \textbf{Assumption \ref{diffusion}, \ref{jump}, \ref{independence}, \ref{bounded}, \ref{hfasym}} as well as \textbf{Assumption \ref{regular}, \ref{Ito.variance}}. Let $\frac{K_n}{n^{3/5}}\to\infty$, $\frac{K_n}{n^{2/3}}\to0$, then we have
	\begin{equation}\label{UH2}
	\sqrt{\frac{n}{K_n}}\left(\frac{3n}{2K_n^2}\overline{U}(Y,K_n,2)_{\mathcal{T}}^n-[g,g]_{\mathcal{T}}\right)\overset{\mathcal{L}-s}{\longrightarrow}\mathcal{MN}\left(0,\frac{3{\mathcal{T}}}{2}\int_0^{\mathcal{T}}(\sigma^{(g)}_t)^4\,\mathrm{d}t\right)
	\end{equation}
\end{corol}


\begin{remk}
	Toward a better finite-sample performance, for example, to get a more accurate confidence interval for the aggregate liquidity risk, we suggest to use the estimate of \begin{equation*}
	\left(\underbrace{\frac{3{\mathcal{T}}}{2}\int_0^{\mathcal{T}}(\sigma^{(g)}_t)^4\mathrm{d}t}_{\text{due to discretization (non-vanishing)}}+\underbrace{\frac{54n^2}{K_n^4\mathcal{T}}\int_0^{\mathcal{T}}\left[h^2_t-h_tg^2_t+g^4_t\right]\mathrm{d}t}_{\text{due to noise (vanishing)}}\right)\times\frac{K_n}{n}
	\end{equation*}
	as the approximation to the finite-sample variance, in order to avoid the situation in which we underestimate the finite-sample variance and become overoptimistic about the accuracy of our estimate. 
	The 95\% confidence interval for our measure ``\textit{aggregated liquidity risk}'' is
	\begin{equation}
		\left[\frac{3n}{2K_n^2}\overline{U}(Y,K_n,2)_{\mathcal{T}}^n-1.96\times\widehat{\Gamma},\hspace{2mm}\frac{3n}{2K_n^2}\overline{U}(Y,K_n,2)_{\mathcal{T}}^n+1.96\times\widehat{\Gamma}\right]
	\end{equation}
	where
	\begin{multline*}
		\widehat{\Gamma}^2=\frac{27}{128\,l_n^2K_n^4}\sum_{i=1}^{\lfloor n/K_n\rfloor-l_n}\left[\sum_{j=1}^{l_n}\left([Y,Y]_{\mathcal{S}_{i+j}}-[Y,Y]_{\mathcal{S}_{i+j-1}}\right)^2\right]^2\\+\frac{27}{8\,K_n^2}\sum_{i=1}^{\lfloor n/K_n\rfloor}\left(\frac{4}{K_n^2}[Y;4]_{\mathcal{S}_i}^2-\frac{14}{K_n^3}[Y;4]_{\mathcal{S}_i}[Y,Y]_{\mathcal{S}_i}^2+\frac{13}{K_n^4}[Y,Y]_{\mathcal{S}_i}^4\right)
	\end{multline*}
	and $l_n\asymp\sqrt{n/K_n}$.
\end{remk}


\section{Noise functional dependency and model extension}\label{modelpp}


The law of microstructure noise is represented via a Markov kernel $Q_t(\omega^{(0)},\mathrm{d}y)$ for each time $t$, through which the second moment of the noise evolves according to a random function in time $g_t(\omega^{(0)})=E(\epsilon_t^2|\mathcal{F}^{(0)})(\omega^{(0)})=\int_{\mathbb{R}}|y-Z_t(\omega^{(0)})|^2\,Q_t(\omega^{(0)},\mathrm{d}y)$ on the probability space $(\Omega^{(0)},\mathcal{F}^{(0)},\{\mathcal{F}^{(0)}_t
\}_{t\ge0},\mathbb{P}^{(0)})$. The random function $g_t(\omega^{(0)})$ could depend on various latent variables, and more generally the form of $Q_t(\omega^{(0)},\mathrm{d}\omega^{(1)})$ allows a wide range of correlation structures between the efficient price process and the microstructure noise. In this section, we discuss an elementary empirical evidence about the dependence of $g_t(\omega^{(0)})$ on $\sigma_t^2(\omega^{(0)})$ and the implication of the violation of our \textit{identification assumption} i.e. \textbf{Assumption \ref{identification}}.

\subsection{Regression: market microstructure noises and spot volatilities}\label{regress}
In this subsection, we go beyond the recognition that the second moment of microstructure noise is evolving over time, document the dependence of $g_t$ on volatility $\sigma^2_t$.
In doing so, we conduct time series linear regression of $g_{t_i}$'s on the latent variables $\sigma_{t_i}^2$'s\footnote{\cite{as16} provide a theoretical underpinning for the negative correlation between volatility and liquidity, equivalently, positive correlation between $\sigma^2_t$ and $g_t$: a higher volatility indicates a higher risk that arbitrageurs might take advantage of market makers' previous orders to act against market makers; hence, it reduces the activities of liquidity provision.}. We assume that the latent market microstructure noise variance and the latent volatility are correlated:
\begin{assum}\label{g-sigma}
	With probability 1, we have $\forall t\ge0$,
	\begin{equation}\label{lin1}
	g_t=\beta\sigma^2_t +\alpha+\zeta_t
	\end{equation}
	where $\zeta_t\indep\sigma_t^2$.
\end{assum}

Since both $g_t$ and $\sigma^2_t$ are unobservable, we need some preliminary estimates for both variables. Here, we use scaled sample-weighted TSRV and realized variance calculated from local samples to estimate spot volatilities $\sigma^2_{t_i}$'s and local noise levels, respectively, i.e., $\widehat{\sigma}^2_{\tau_{i-1}}=\frac{1}{\tau_i-\tau_{i-1}}\widehat{\langle X,X\rangle}^{(WTSRV)}_{(\tau_{i-1},\tau_i]}$ and $\widehat{g}_{\tau_{i-1}}=\frac{1}{2|\mathcal{H}_i|}[Y,Y]_{\mathcal{H}_i}$, where $\{\tau_0,\tau_1,\cdots\}\subset\mathcal{G}$, $\mathcal{H}_i=\mathcal{G}\cap(\tau_{i-1},\tau_i]$, $|\mathcal{H}_i|$ is the cardinality of $\mathcal{H}_i$. Then, we can conduct linear regression on these pairs of volatility-noise estimates $(\widehat{g}_{\tau_i}, \widehat{\sigma}^2_{\tau_i})$'s by ordinary least squares:
\begin{equation}\label{lin2}
\widehat{g}_{\tau_i}=\widehat{\beta}_m\widehat{\sigma}^2_{\tau_i} +\widehat{\alpha}_m+\eta^{(m)}_{\tau_i}
\end{equation}
where $m$ is the number of observation in the small time interval $(\tau_i,\tau_{i+1}]$, and $\eta^{(m)}_t$ denotes a component in the noise variance not captured by the volatility estimator $\widehat{\sigma}^2_{\tau_i}$. Besides, we use $m$ in the subscripts of estimators $\hat{\alpha}_m$ and $\hat{\beta}_m$ to emphasize that the values of the estimators in (\ref{lin2}) depend on the sample size $n$, and the distribution of $\eta^{(m)}_t$ also depends on $m$.

\begin{lem}\label{lemmareg}
Suppose \textbf{Assumption \ref{diffusion}, \ref{finite.jump}, \ref{identification}, \ref{independence}, \ref{bounded}, \ref{hfasym}} ae well as \textbf{Assumption \ref{g-sigma}} hold, let $\min{|H_i|}\to\infty$ and $\max{|H_i|}/n\to0$, then
\begin{eqnarray*}
	\widehat{\beta}_m &\overset{\mathbb{P}}{\longrightarrow}&\beta\\
	\widehat{\alpha}_m&\overset{\mathbb{P}}{\longrightarrow}&\alpha
\end{eqnarray*}
\end{lem}


By \textbf{lemma \ref{lemmareg}}, if there is a linear relationship between the noise variance and the spot volatility, the regression (\ref{lin2}) provides consistent estimates of linear coefficients. Figure \ref{reg} shows the least square regression plots for high-frequency transaction data in April, 2013 of 6 stocks: International Business Machines (IBM), Goldman Sachs (GS), Johnson \& Johnson (JNJ), Nike, Inc. (NKE), Chevron Corporation (CVX), McDonald's (MCD).

The time series regression and empirical analysis here are preliminary. One can investigate the statistical properties of this type of linear regression in more detail. Perhaps, there are non-linear relations. These issues will be addressed in our future research.

\begin{figure}
	  \centering
	  \includegraphics[width=.52\textwidth]{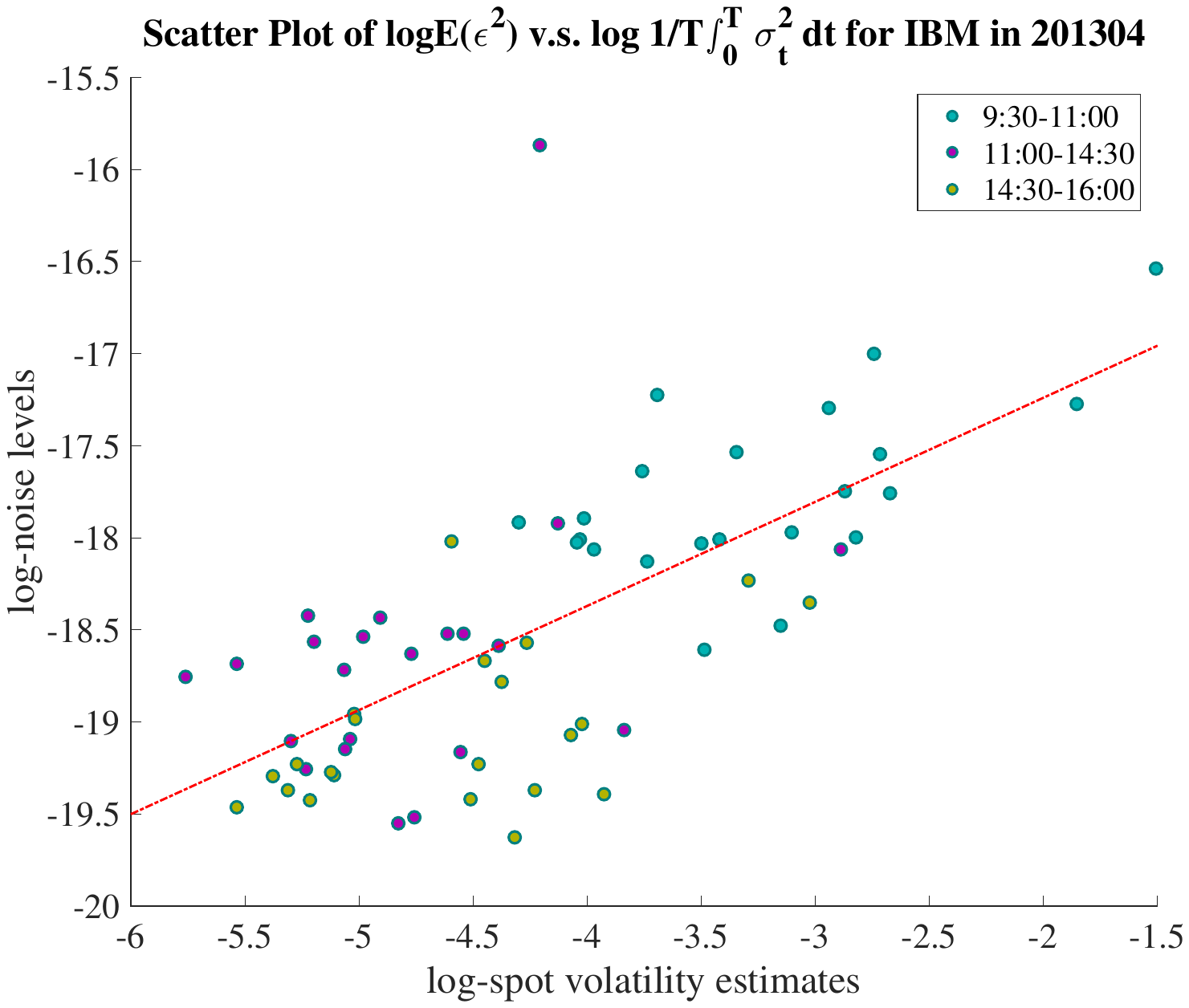}\includegraphics[width=.52\textwidth]{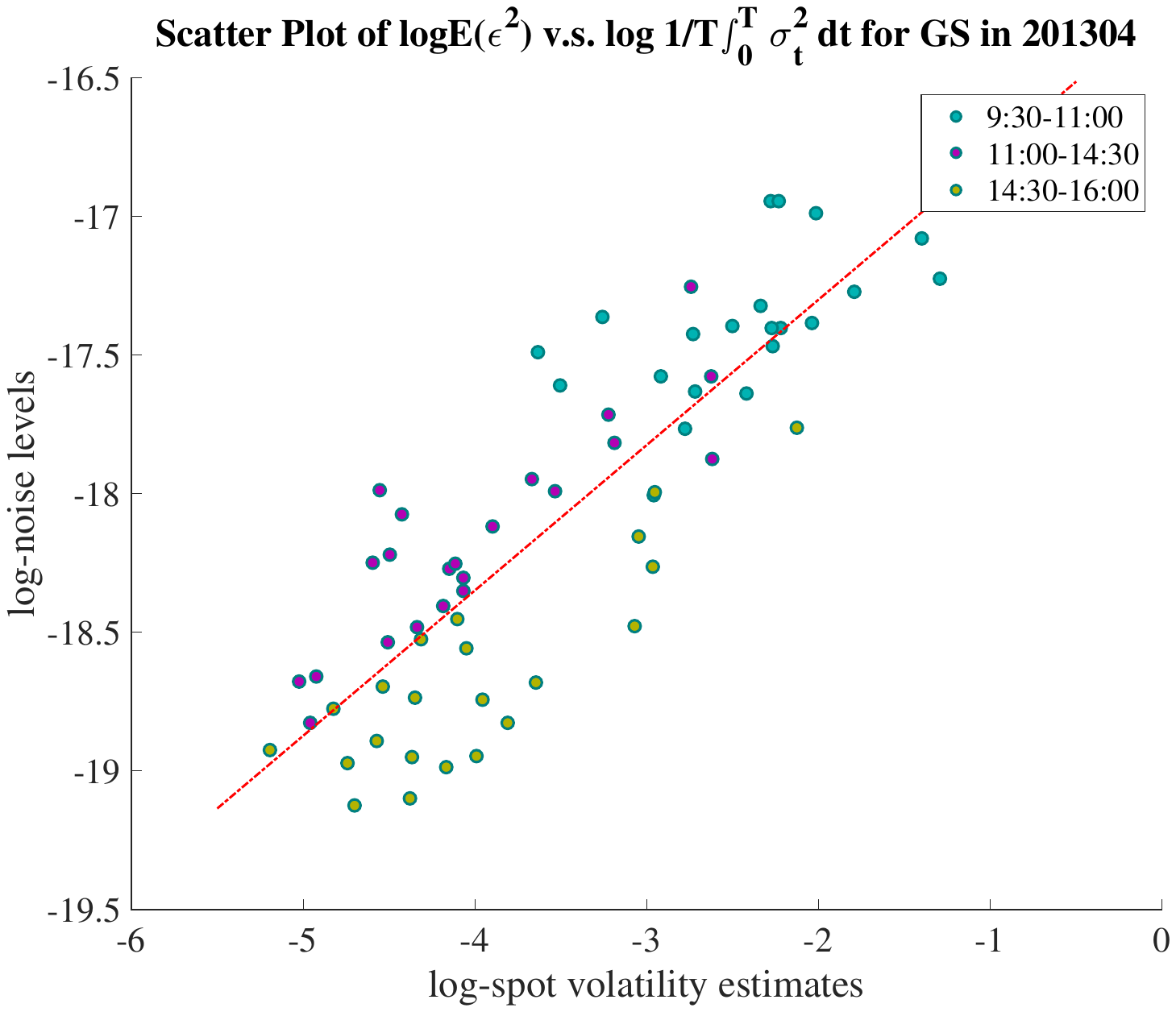}\\
	  \includegraphics[width=.52\textwidth]{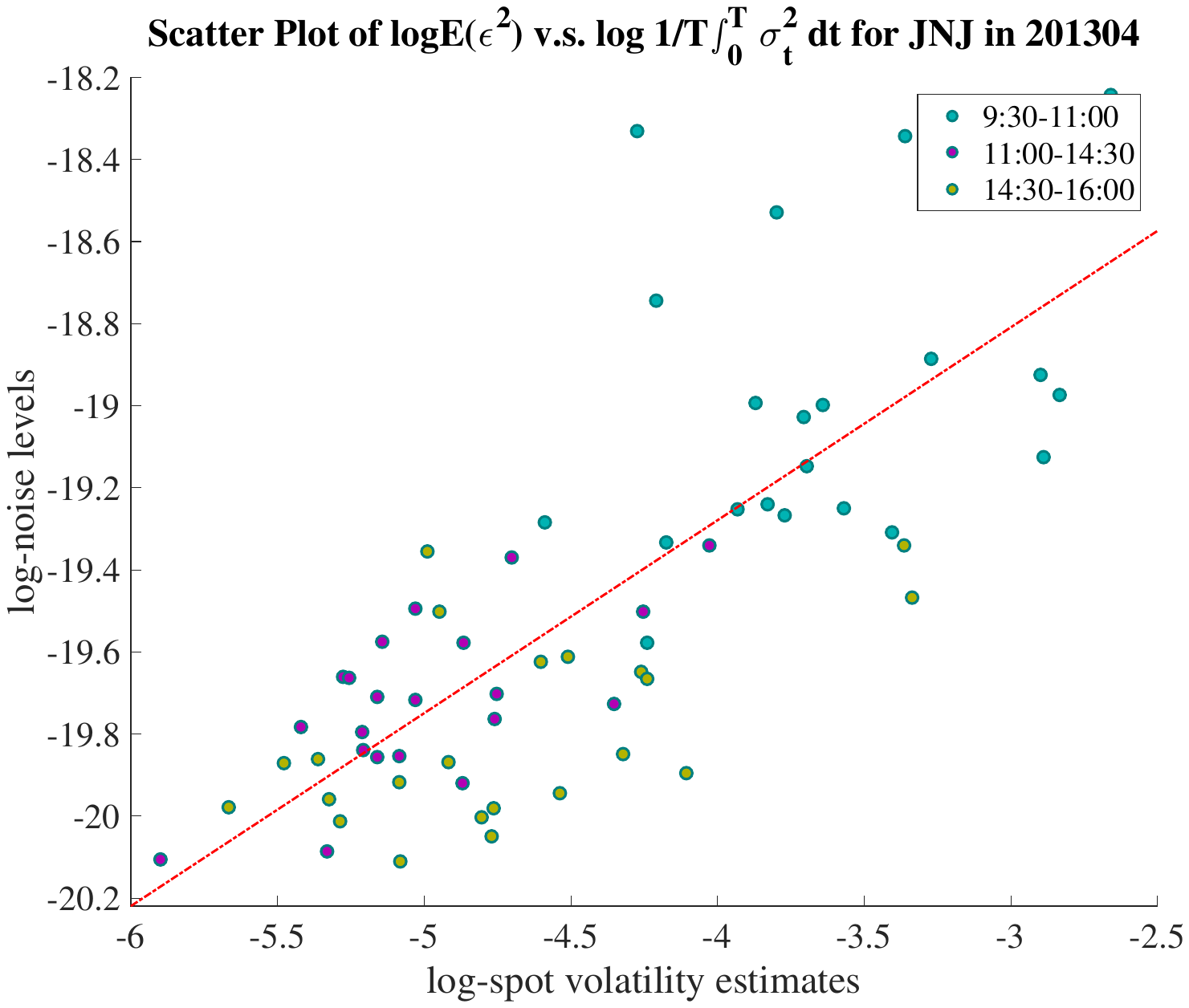}\includegraphics[width=.52\textwidth]{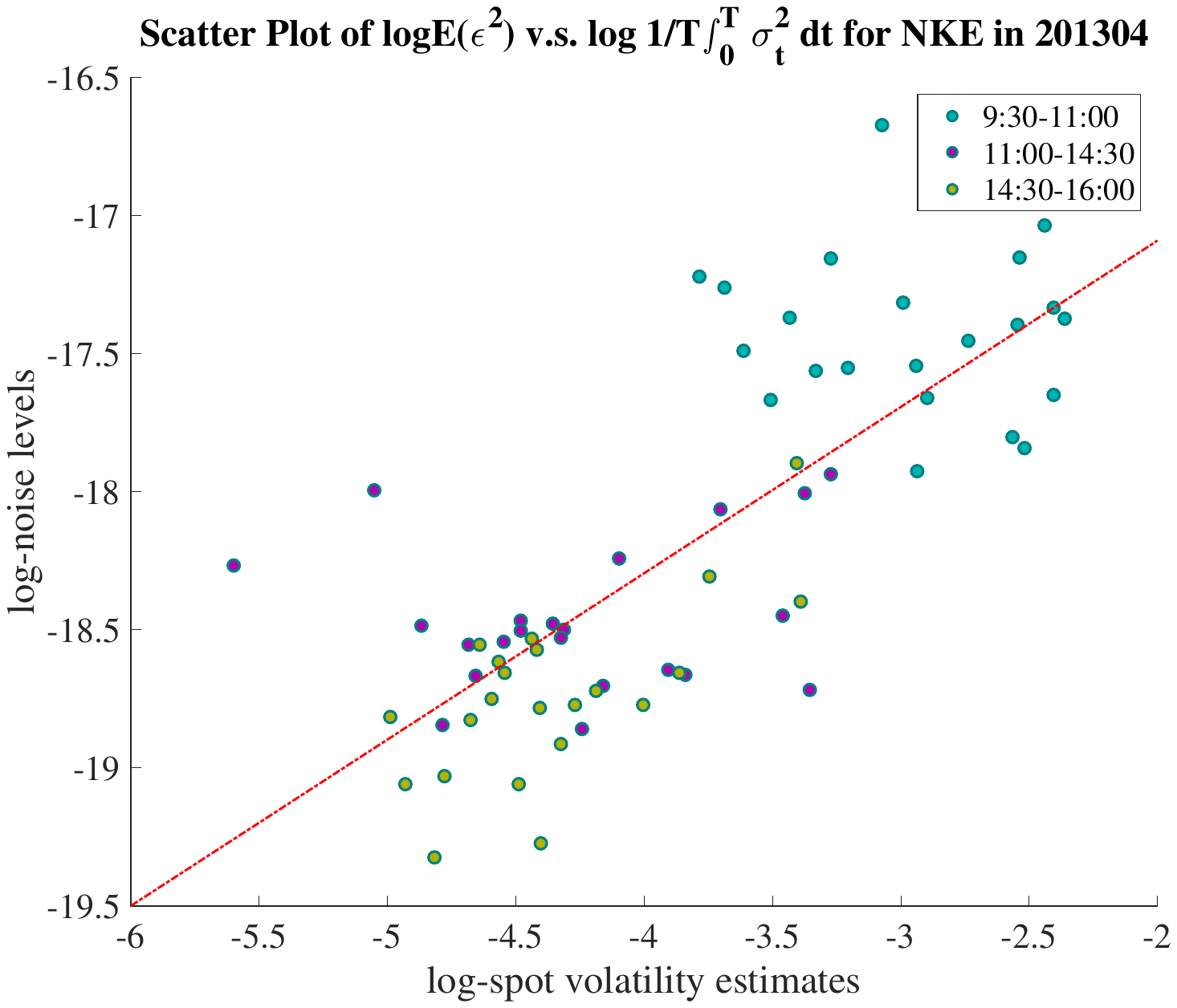}\\
	  \includegraphics[width=.52\textwidth]{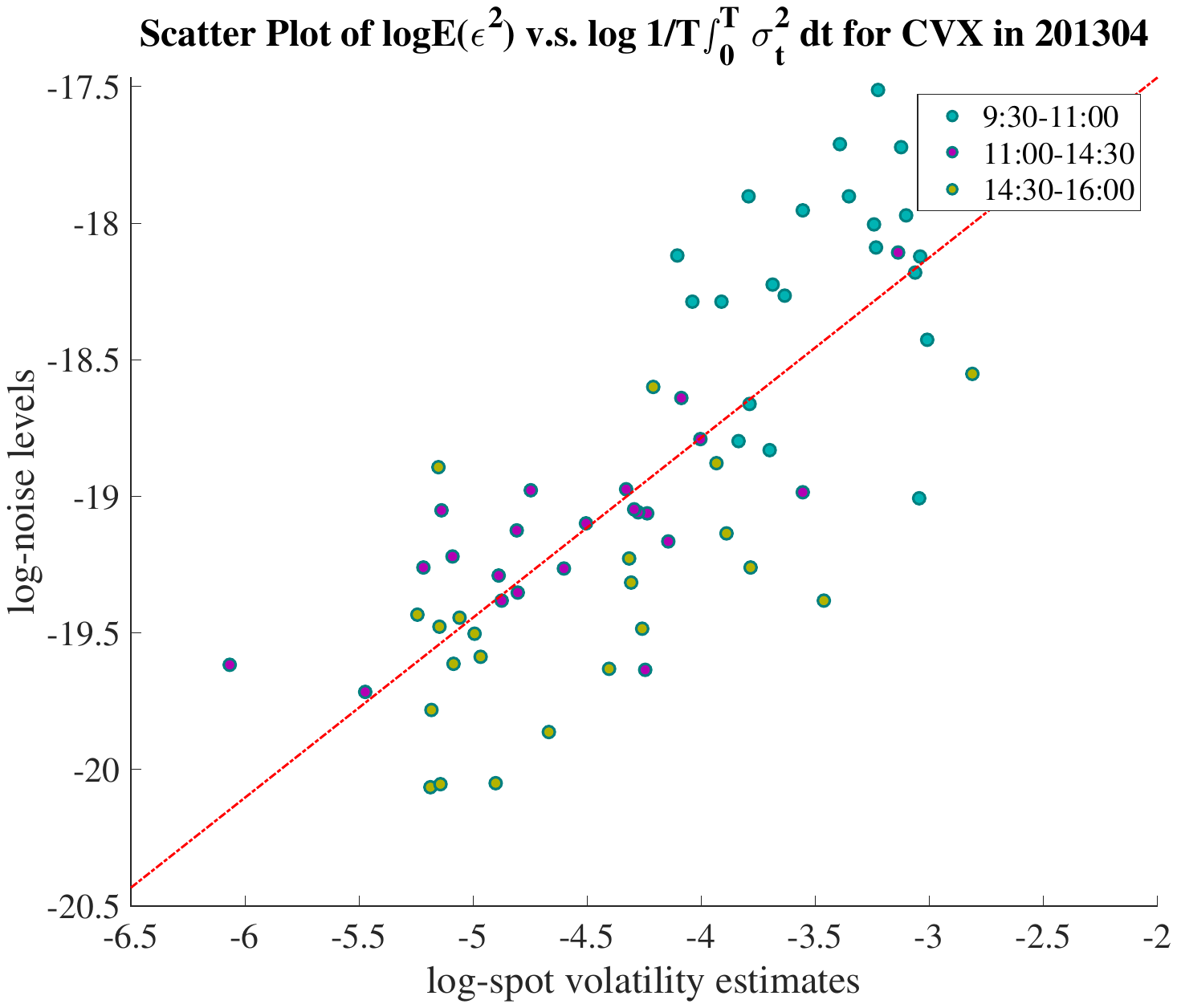}\includegraphics[width=.52\textwidth]{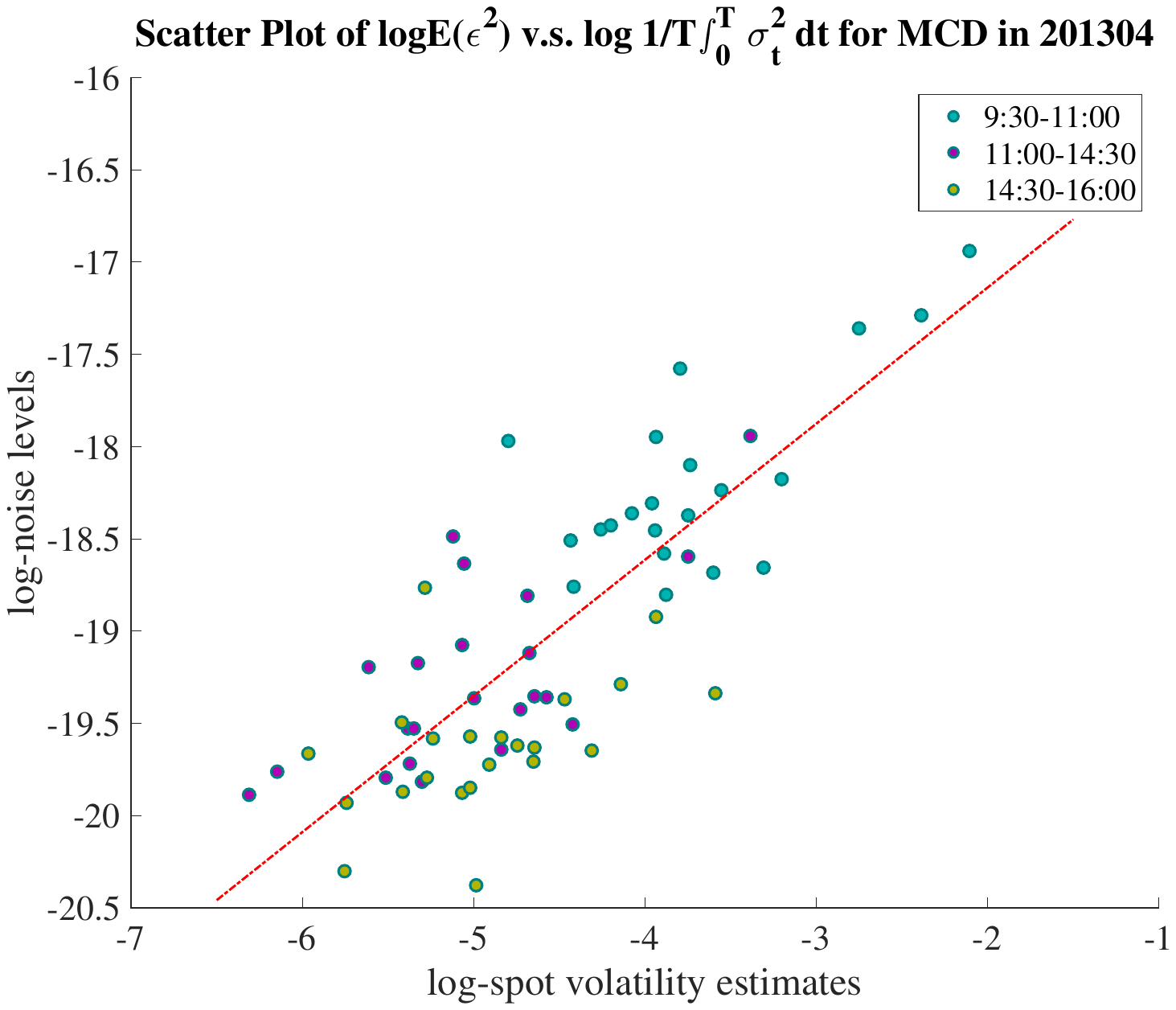}
	  \caption{Scatter plots of $\log(\widehat{g}_{\tau_i})$ against $\log(\widehat{\sigma}^2_{\tau_i})$, where each $\tau_i$ represents a particular period in each day. The \textcolor{red}{red dotted lines} are the fitted regression lines for IBM (IT), GS (finance), JNJ (medicine and pharmacy), NKE (manufacturing), CVX (energy), MCD (fast food), from left to right and top to bottom, respectively.}
	  \label{reg}
\end{figure}

\subsection{Model extension: endogenous noise}\label{modelexten}
In our model, we allow arbitrary fashion of the noise process up to the time-varying Markov kernel $Q_t(\cdot,\cdot)$ plus the identification assumption (\textbf{Assumption \ref{identification}}). As documented in \cite{j09}, the identification assumption is restrictively strong. If one is interested in the stationarity of $\{\epsilon_t\}_{t\ge0}$, our methods are valid regardless the identification assumption holds or not. However, if one is concerned about $\{e_t\}_{t\ge0}$ our methods will break down when the identification assumption is violated. Nevertheless, this extension is indispensable for empirically compatible modeling and it allows endogenous microstructure noise (noise which is correlated with the efficient price \citep{hl06}).

Note that in subsection \ref{setup}, conditioning on all latent variables, $\epsilon_t$ is a mean-zero random variable, i.e., $\int_{\mathbb{R}}(y-Z_t(\omega^{(0)}))\,Q_t(\omega^{(0)},\mathrm{d}y)=0$ since $\int_{\mathbb{R}}y\,Q_t(\omega^{(0)},\mathrm{d}y)=Z_t(\omega^{(0)})$. 
However, the conditional mean of $e_t$ is not necessarily 0 since $E(e_t|\mathcal{F}^{(0)})=E(Y_t-X_t|\mathcal{F}^{(0)})=Z_t-X_t$.

This observation enables us to, non-parametrically, introduce endogenous noise into our model. We can allow instantaneous/realized correlation between the latent process $\{X_t\}_{t\ge 0}$ and the noise process $\{e_t\}_{t\ge 0}$. 
Although $E(e_t|\mathcal{F}^{(0)})$ is not necessarily 0, we assume the unconditional mean $E_{\mathbb{P}}(e_t)$ is zero, then calculation shows
\begin{eqnarray*}
	\mathrm{Cov}(X_t,e_t)&=&E_{\mathbb{P}^{(0)}}[X_tZ_t]-E_{\mathbb{P}^{(0)}}[X^2_t]\\
    \mathrm{Cov}(Z_t,e_t)&=&E_{\mathbb{P}^{(0)}}[Z^2_t]-E_{\mathbb{P}^{(0)}}[X_tZ_t]\\
    \mathrm{Cov}(X_t,\epsilon_t)&=&0\\
    \mathrm{Cov}(Z_t,\epsilon_t)&=&0
\end{eqnarray*}

\citep{j09} assumed $Z_t=X_t$, so there is no endogenous noise in their model. However, as long as $E_{\mathbb{P}^{(0)}}[X_tZ_t]\ne E_{\mathbb{P}^{(0)}}[X^2_t]$, there is correlation between the latent process $\{X\}_{t\ge 0}$ defined by (\ref{X}) and the noise process $\{e_t\}_{t\ge 0}$ defined by (\ref{error}).


An intuitive interpretation is that $e_t$ encodes some relevant information about the processes defined on the latent probability space if it is correlated with the latent random variables $X_t$ and $Z_t$. In contrast, $\epsilon_t$ is a pure noise and conveys no useful information about the latent processes, the correlation between $\epsilon_t$ and any latent random variable is zero. For this reason, we call $e_t$ ``\textit{endogenous microstructure noise}'', and call $\epsilon_t$ ``\textit{exogenous microstructure noise}''.
\begin{center}
	\usetikzlibrary{positioning}
	\usetikzlibrary{shapes,snakes}
	\begin{tikzpicture}[xscale=12,yscale=6,>=stealth]
	\tikzstyle{s}=[rectangle, minimum size=5mm,draw,thick]
	\tikzstyle{e}=[ellipse,   minimum size=5mm,draw,thick]
	\tikzstyle{v}=[circle,    minimum size=5mm,draw,thick]
	\node[s] (y)   [draw=blue!60,fill=blue!10] {$Y_t$};
	\node[e] (x)   [draw=green!60,fill=green!10,right=of y] {$X_t$};
	\node[e] (z)   [draw=red!60,fill=red!10,right=of x] {$Z_t=E(Y_t|\mathcal{F}^{(0)})$};
	\node[v] (e)   [draw=green!60,fill=green!10,below=of x] {$e_t$};
	\node[v] (eps) [draw=red!60,fill=red!10,right=of e] {$\epsilon_t$};
	\draw[thick,->] (y) to node {} (x);
	\draw[thick,->] (y.south) .. controls +(down:1mm) and +(left:1mm) ..  (e.west);
	\draw[thick,->] (x) to node {} (z);
	\draw[thick,->,snake=snake] (e) to node[anchor=west] {information} (z);
	\draw[thick,->] (e) to node {} (eps);
	\end{tikzpicture}
\end{center}
\begin{remk}
	When one tries to estimate the integrated volatility, the quantity which is actually estimated is $\langle Z,Z\rangle_T$, not necessarily the usually desired target $\langle X,X\rangle_T$. This is discussed by \citep{lm07}. In contrast to \citep{j09}, we do not assume $\int_{\mathbb{R}}y\,Q_t(\omega^{(0)},\mathrm{d}y)=X_t(\omega^{(0)})$. In other words, in the case where $Z_t\ne X_t$, the integrated volatility $\langle X,X\rangle_T$ is not identifiable; however, if we are satisfied with estimating $\langle Z,Z\rangle_T$, then we are able to introduce some conditional correlation between the efficient price and the microstructure noise.
\end{remk}

One conceptual finding from the model extension is the informational content in microstructure noise $\{e_t\}_{t\ge 0}$ with respect to the efficient price in financial term (or latent process in statistical term) which is modeled as an It\^o semimartingale. The interpretation comes from market microstructure theory \citep{o95,o03}. As in the classical asset pricing theory, we take the price as given and exogenous, and conduct trading and hedging strategies, portfolio allocation and risk management. But, the price discovery and price formation depend on the behaviors of market participants, no price will be produced without investment activities of various market participants. It is the balance between demand and supply from investors, it is the psychology of people in the market, it is the synthesis of microscopic effects of beliefs and behaviors of market participants, that determine the prices. Thus, the efficient price should be an endogenous process in the financial market. It is one of striking difference between asset pricing and market microstructure theory: the classical asset pricing theory assumes frictionless and competitive market in which people do not have to worry about the price impact and liquidity constraint. While, in market microstructure theory, the modelers need to look inside the ``black box'' of the trading processes, and take market making, price discovery, liquidity formation, inventory control, asymmetric information into account.

Since we consider the price as endogenous, which, for example, affected by transaction costs (like bid-ask spread), inventory control, discrete adjustment of price, lagged incorporation of new information, insider trading and adverse selection brought by asymmetric information, lack of liquidity caused by one or several of the factors mentioned above, the It\^o process is merely an approximation to the efficient price observed at high-frequency, at which market microstructure effects manifest itself to such extent that the accumulated noise swamps the integrated volatility of the latent It\^o process and the variation in microstructure noise dominates the total variance.

Therefore, it is reasonable (even indispensable) to extend our model to allow  endogenous microstructure noise, at least from the viewpoint of microstructure theory, and for sake of realistic modeling at low-latency and millisecond level. This topic is not the focus of this paper; in-depth discussion and treatment on endogenous microstructure noise will be addressed in our future research.

\section{Simulation}\label{simul}
\subsection{Simulation scenario}\label{simul.config}

The configuration of our simulation design is
\begin{eqnarray}\label{xiu2010}
Y_{t_i}             &=&\left\lfloor \frac{X_{t_i}+\epsilon_{t_i}}{\alpha}\right\rfloor\alpha\\
\mathrm{d}X_t       &=&\mu\,\mathrm{d}t+\sigma_{t}\,\mathrm{d}W_t+J^X_t\,\mathrm{d}N^X_t\\
\mathrm{d}\sigma^2_t&=&\kappa(\bar{\sigma}^2-\sigma_t^2)\,\mathrm{d}t+\delta\sigma_{t}\,\mathrm{d}B_t+\sigma_{t-}J^V_t\,\mathrm{d}N^V_t
\end{eqnarray}
where $E(\mathrm{d}W_t\cdot\mathrm{d}B_t)=\rho\,\mathrm{d}t$, $N^X$ and $N^V$ are Poisson processes $\indep W,B$ with parameters $\lambda_X$ and $\lambda_V$ respectively, the jump sizes satisfy $J^X\sim N(\theta_X,\nu_X)$ and $J^V_t=e^z$ with $Z\sim N(\theta_V,\nu_V)$. The stationary microstructure noise behaves as $\epsilon^{(s)}_{t_i}\overset{\text{i.i.d.}}{\sim}N(0,a_0^2)$, whereas the non-stationary microstructure noises are distributed as
\begin{equation}\label{noise.curve}
\begin{array}{ll}
    \epsilon^{(ns)}_{t_i}&= \sqrt{\frac{60}{17}}\left[\left(\frac{i}{n}-0.5\right)^2+0.2\right]^{\frac{1}{2}}\times e_{t_i}\\
    e_{t_i}&=z_i+\sum_{j=1}^M\binom{u+j-1}{j}z_{i-j}\\
    z_k&\overset{\text{i.i.d.}}{\sim}N\left(0,\omega^2\right),\hspace{4mm} \omega^2=a_1\left(\frac{1}{n}\sum_{j=1}^n\sigma_{it_j}^4\right)^\frac{1}{2}
\end{array}
\end{equation}
where $u\in(-0.5,0.5)$ and $n$ is the number of high-frequency observations in 1 business day. In (\ref{noise.curve}), the noise variance of $\{\epsilon^{(ns)}_{t_i}\}_i$ changes according to a U-curve, which means that the noise is of relatively higher levels around opening and closed hours. The U-curve is chosen such that the averaged noise variance within a day is $\omega^2$.
The noise conforms to the empirical feature that the variance of microstructure noise increases with the level of volatility \citep{br06}. The parameters are chosen so that they are consistent with \cite{ay09}:
\vspace{2mm}

\centerline{\begin{tabular}{l|cccccccccccc}
		$X$ parameters & $X_0$ & $\mu$ & $\rho$ & $\lambda_X$ & $\theta_X$ & $\nu_X$\\
		\hline
		      & $\ln(100)$ & 0.03  &  -0.6  &  6  &  0.0016    &   0.004\\
		\hline \hline
		$\sigma$ parameters & $\kappa$ & $\bar{\sigma}^2$ & $\delta$& $\lambda_V$ & $\theta_V$ & $\nu_V$\\
		\hline
		&   6     &       0.16        &    0.5   &  12  &    -5    &   0.8\\
		\hline \hline
		noise parameters & $a_0$ & $a_1$ & $\alpha$ &  $M$  &  $u$\\
		\hline
		& $5\times10^{-3}$ & $1.54\times10^{-4}$ & $1\times10^{-5}$ & 10 & 0.3
	\end{tabular}}
	\vspace{2mm}
\noindent Furthermore, $\sigma_0^2$ is sampled from the stationary distribution of Cox-Ingersoll-Ross process \citep{cir85}, i.e., $\text{Gamma}\left(\frac{2\kappa\bar{\sigma}^2}{\delta^2},\frac{\delta^2}{2\kappa}\right)$ so the unconditional mean of the volatility is $\bar{\sigma}^2$. $a_1$ is chosen such that $\mathrm{Var}(\epsilon^{(s)})=\mathrm{Var}(\epsilon^{(ns)})$ in average. We also adopted a random sampling scheme according to an inhomogeneous Poisson process $\mathrm{Poisson}(\lambda_t\times\Delta)$ where $\Delta$ is averaged sampling duration and the trading intensity evolves periodically $\lambda_t=1+0.5\times \cos(2\pi t/T)$ with $T$ being the length of 1 business day.

\subsection{Simulation results}
In Figure \ref{N}, \ref{V} and \ref{Vbar}, we show the simulation results of $N(Y,K_n)^n_T$, $V(Y,K_n,s_n,2)^n_T$ and $\overline{V}(Y,K_n,2)^n_T$ where $T$ is taken to be 1 business day (left panel in each figure) and 5 business days (right panel in each figure). For each test and each time span, the simulation is conducted in 2 different circumstances: stationary noise (upper picture in each column), U-shape noise (\ref{noise.curve}) (lower picture in each column). The plots show various empirical densities function of our proposed tests against the density of $N(0,1)$. Each group of tests were computed from 3000 sample paths with averaged sampling interval 1 second.

\begin{figure}
	\centering
	\caption{Empirical density of $N(Y,K_n)^n$}
	\centerline{\includegraphics[width=.55\textwidth]{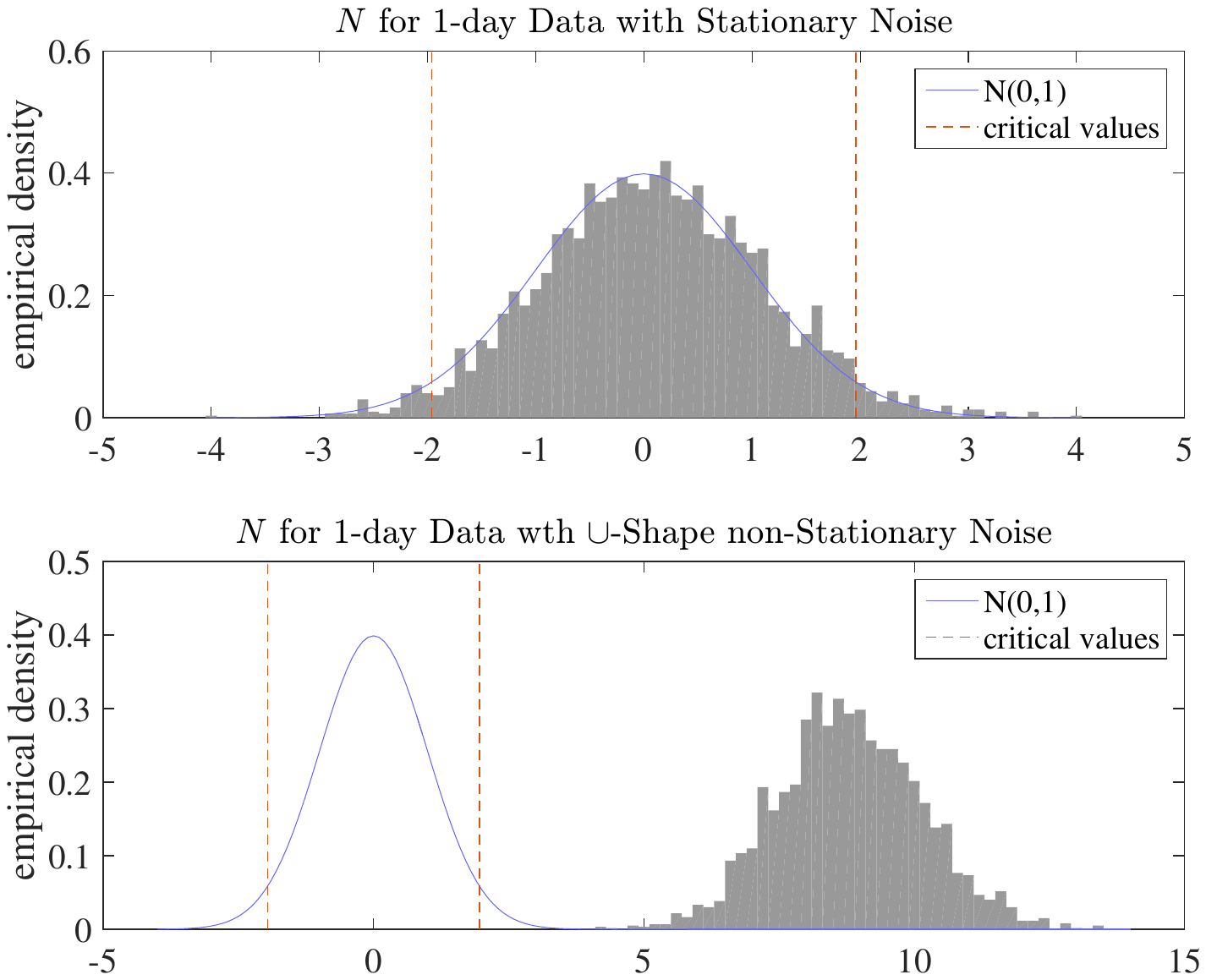}\includegraphics[width=.55\textwidth]{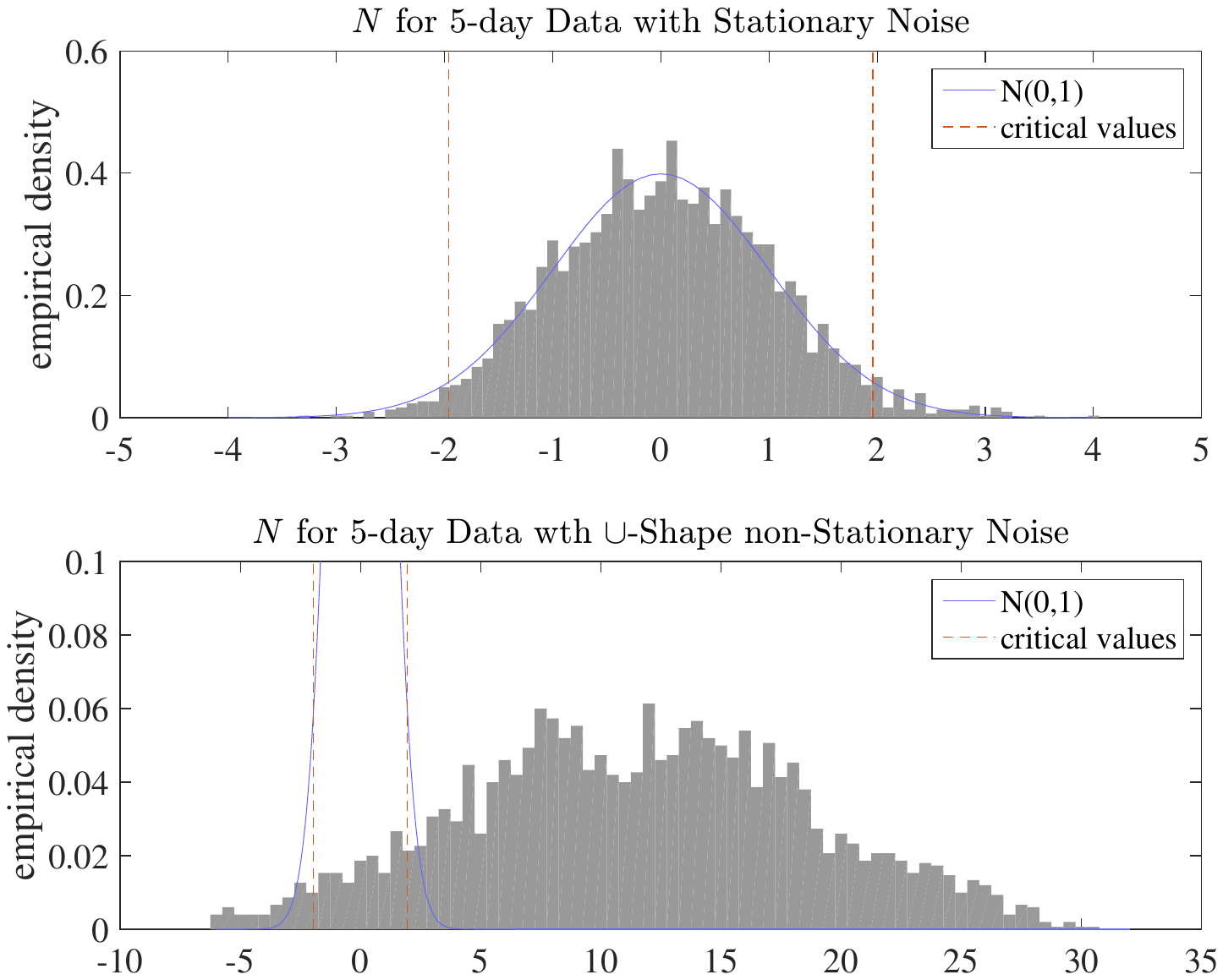}}
	\floatfoot{These plots show the empirical densities of $N(Y,K_n)^n$ when it applies to 1-day/5-day data with stationary/non-stationary noises. Compared the simulation of other tests, we can see $N(Y,K_n)^n$ converges faster to $N(0,1)$ when microstructure noise is stationary. On the other hand, if the microstructure noise is non-stationary and exhibits daily diurnal pattern, $N(Y,K_n)^n$ is the best for 1-day data.}\label{N}
\end{figure}



\begin{figure}
	\centering
	\caption{Empirical density of $V(Y,K_n,s_n,2)^n$}
	\centerline{\includegraphics[width=.55\textwidth]{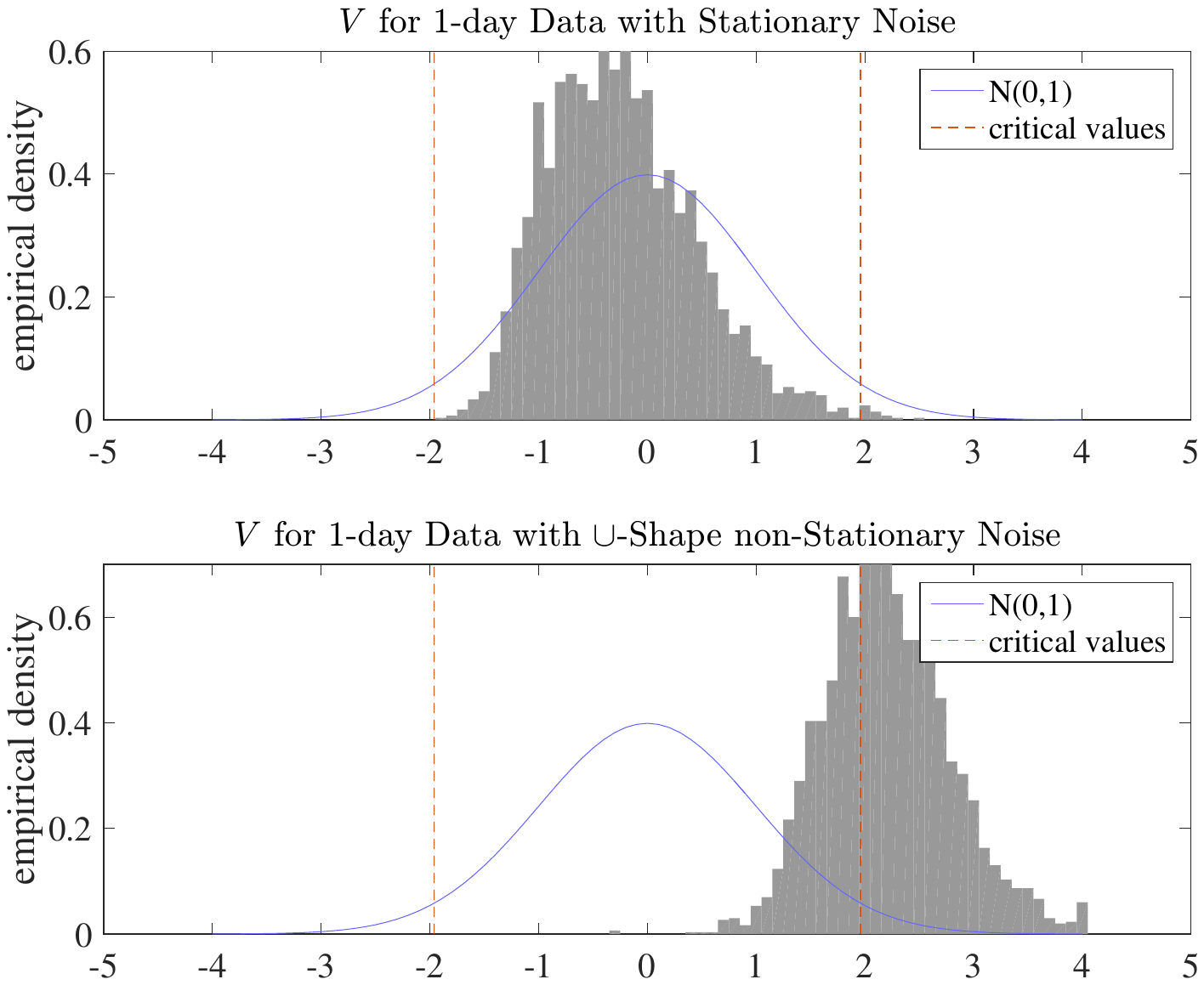}\includegraphics[width=.55\textwidth]{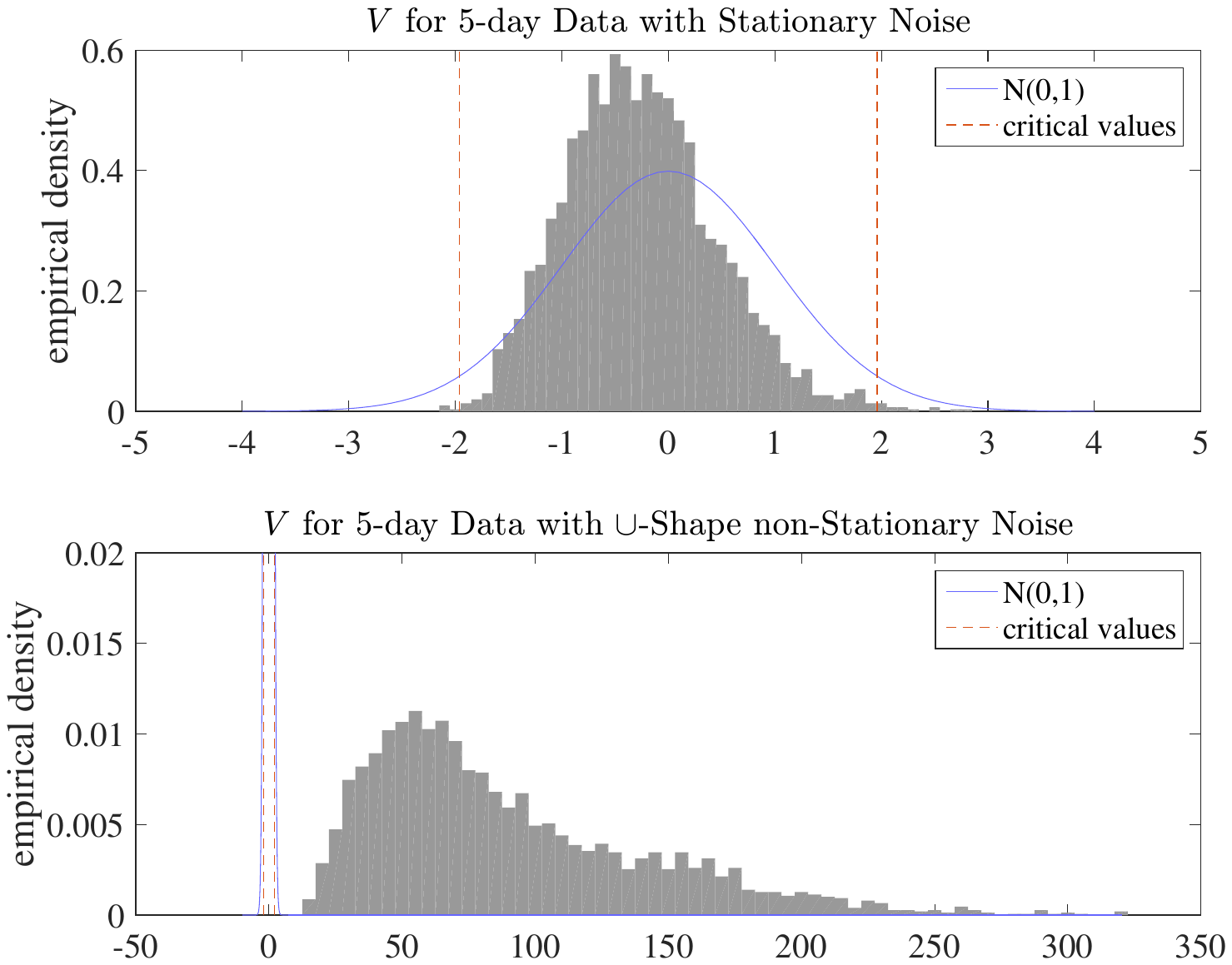}}
	\floatfoot{These plots show the empirical densities of $V(Y,K_n,s_n,2)^n$ when it applies to 1-day/5-day data with stationary/non-stationary noises. Compared the simulation of other tests, we can see $V(Y,K_n,s_n,2)^n$ is more conservative due to its relatively large edge effect when microstructure noise is stationary. On the other hand, if the microstructure noise is non-stationary and exhibits daily diurnal pattern, $V(Y,K_n,s_n,2)^n$ is the best for multi-day data and enjoys the largest statistical power.}\label{V}
\end{figure}

\begin{figure}
  \centering
  \caption{Empirical density of $\overline{V}(Y,K_n,2)^n$}
  \centerline{\includegraphics[width=.55\textwidth]{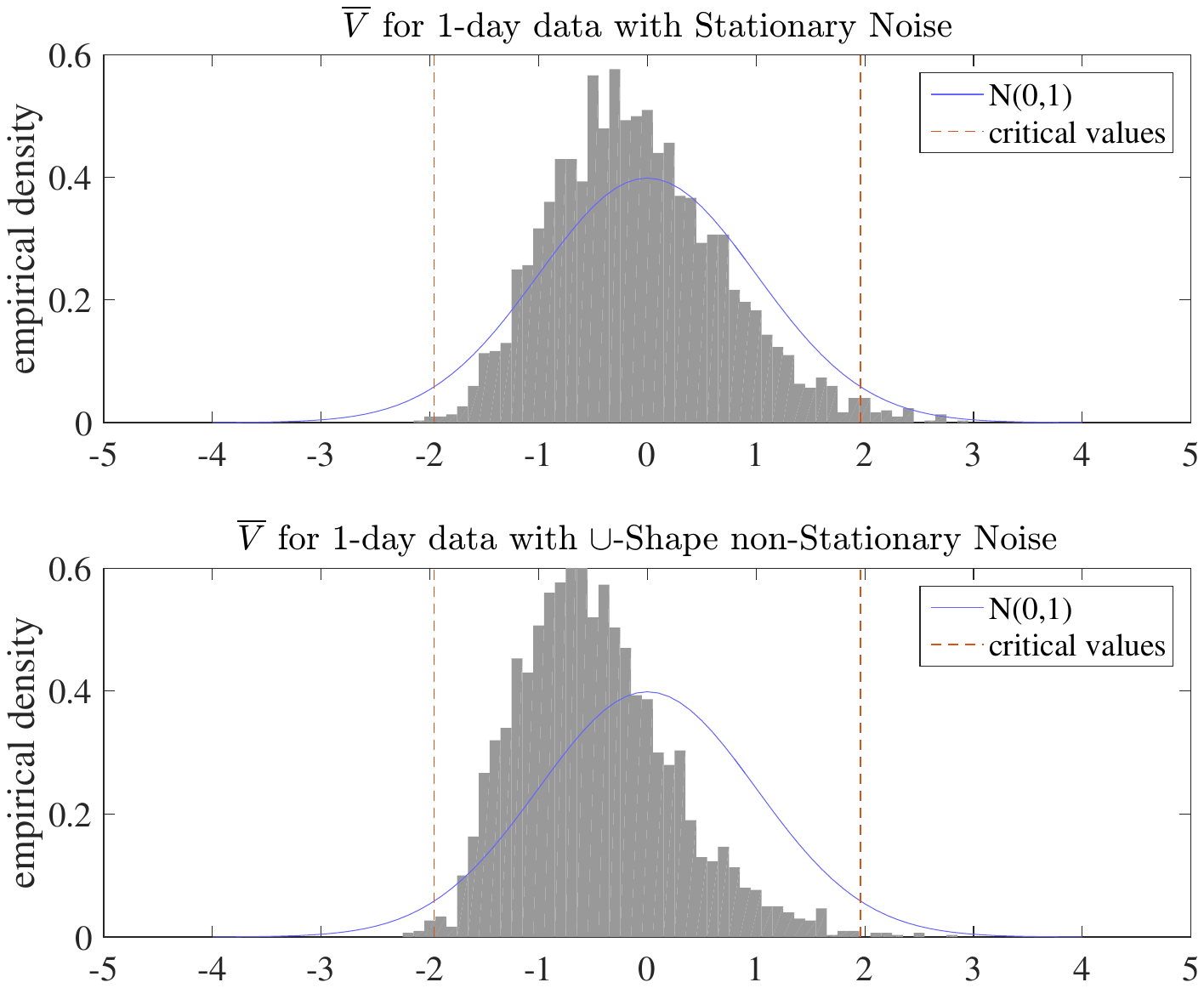}\includegraphics[width=.55\textwidth]{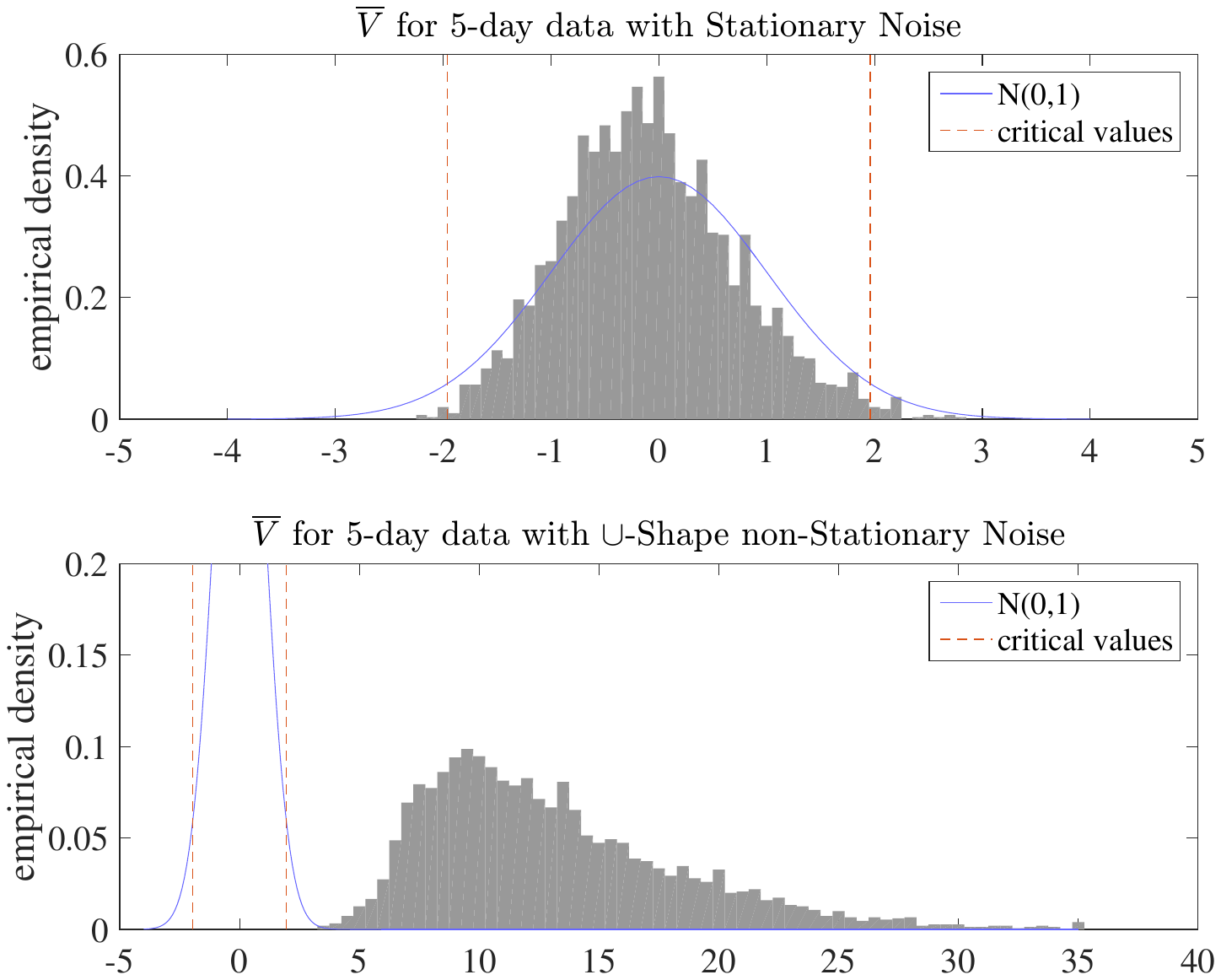}}
  \floatfoot{These plots show the empirical densities of $\overline{V}(Y,K_n,2)^n$ when it applies to 5-day/10-day data with stationary/non-stationary noises. Compared the simulation of other tests, we can see $\overline{V}(Y,K_n,2)^n$ controls type-I error more accurately than $V(Y,K_n,s_n,2)^n$ does when microstructure noise is stationary. On the other hand, if the microstructure noise is non-stationary and exhibits daily diurnal pattern, $N(Y,K_n)$ is better for multi-day data.}\label{Vbar}
\end{figure}




\section{Empirical studies}\label{empirical}
\subsection{Empirical evidence of non-stationary microstructure noise}\label{empievid}

Figure \ref{month_noise} shows daily variations of microstructure noise levels in 2008. Figure \ref{day_noise} exhibits intra-day variations in microstructure noises of individual stocks in the first 4 months of 2013.

\begin{figure}
  \centering
  \includegraphics[width=1.0\textwidth]{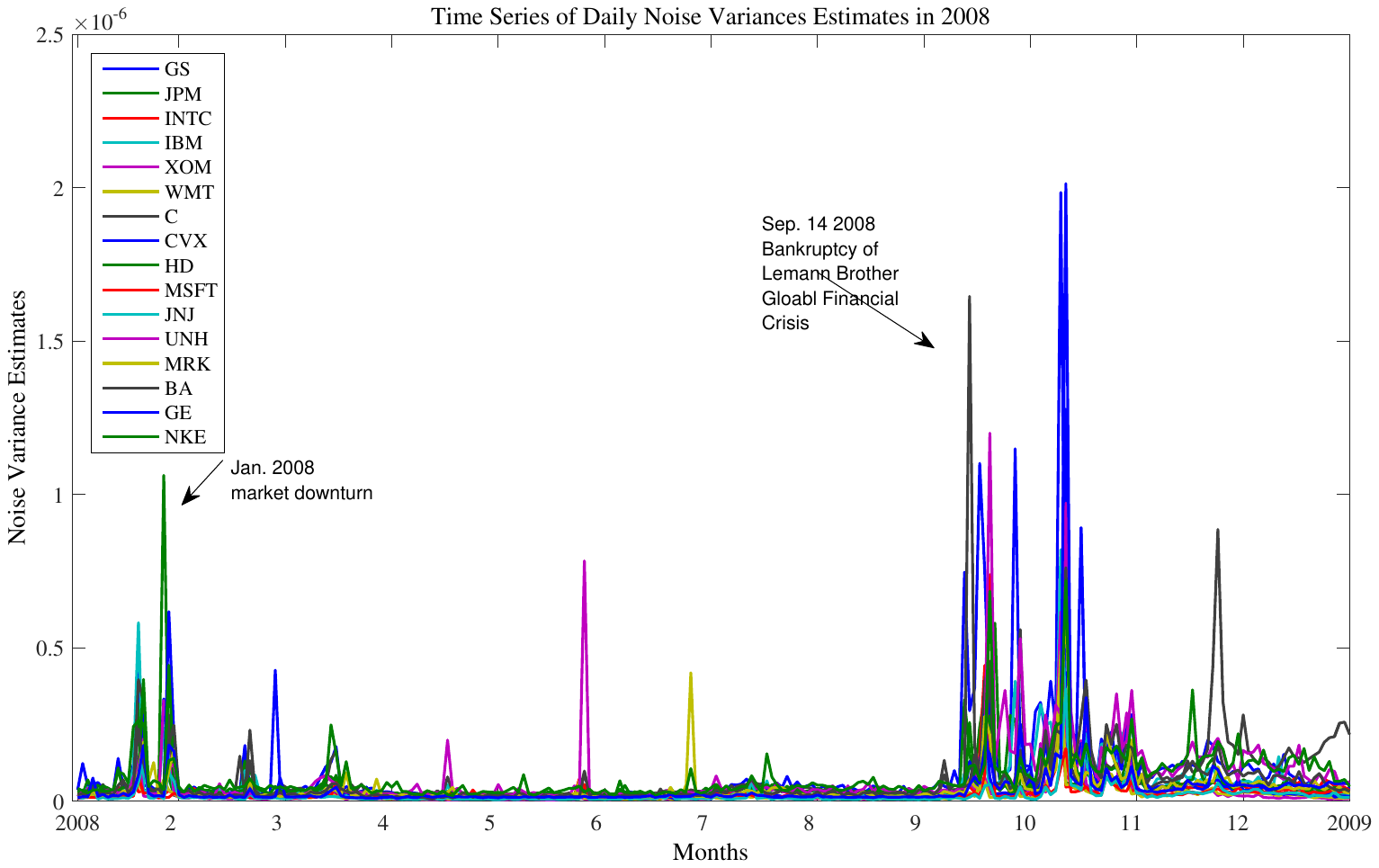}
  \caption{
  Daily noise variance estimates in 2008, with a simple event-history analysis. During the turmoil of financial crisis, the market microstructure noise surged up, the quality of the market worsened strikingly.}\label{month_noise}
\end{figure}
\begin{figure}
	\centering
	\includegraphics[width=.52\textwidth]{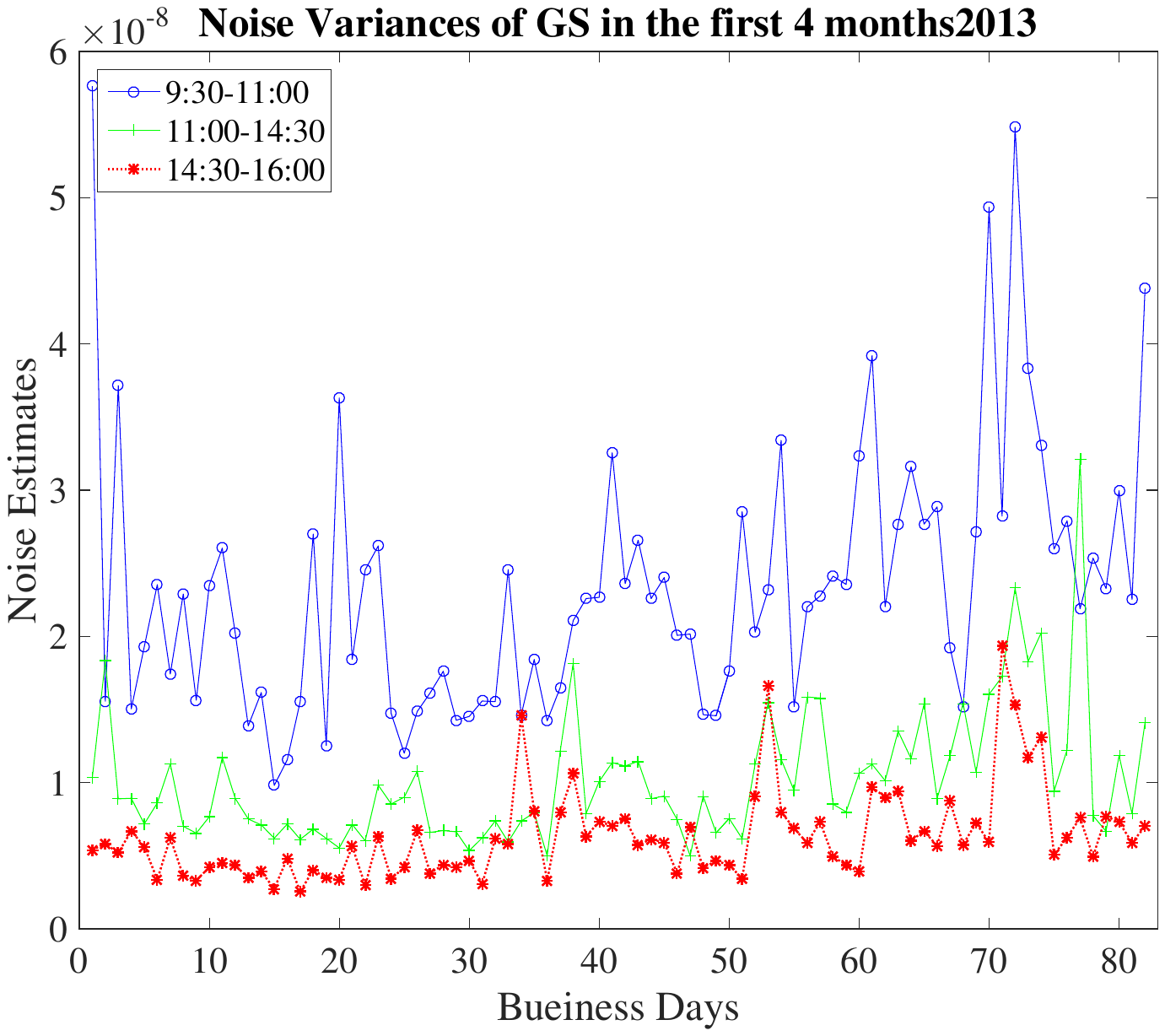}\includegraphics[width=.52\textwidth]{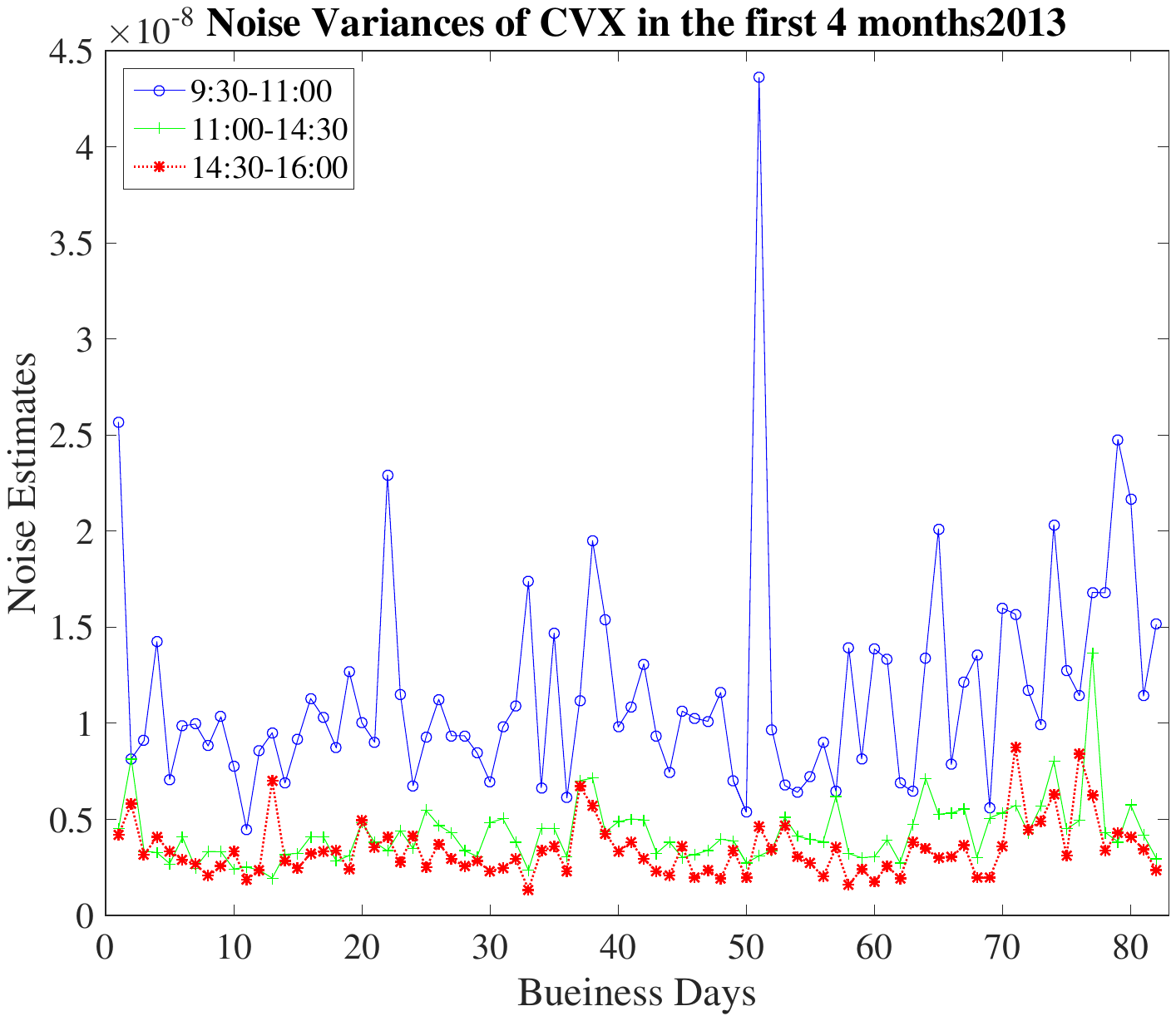}\\
	\includegraphics[width=.52\textwidth]{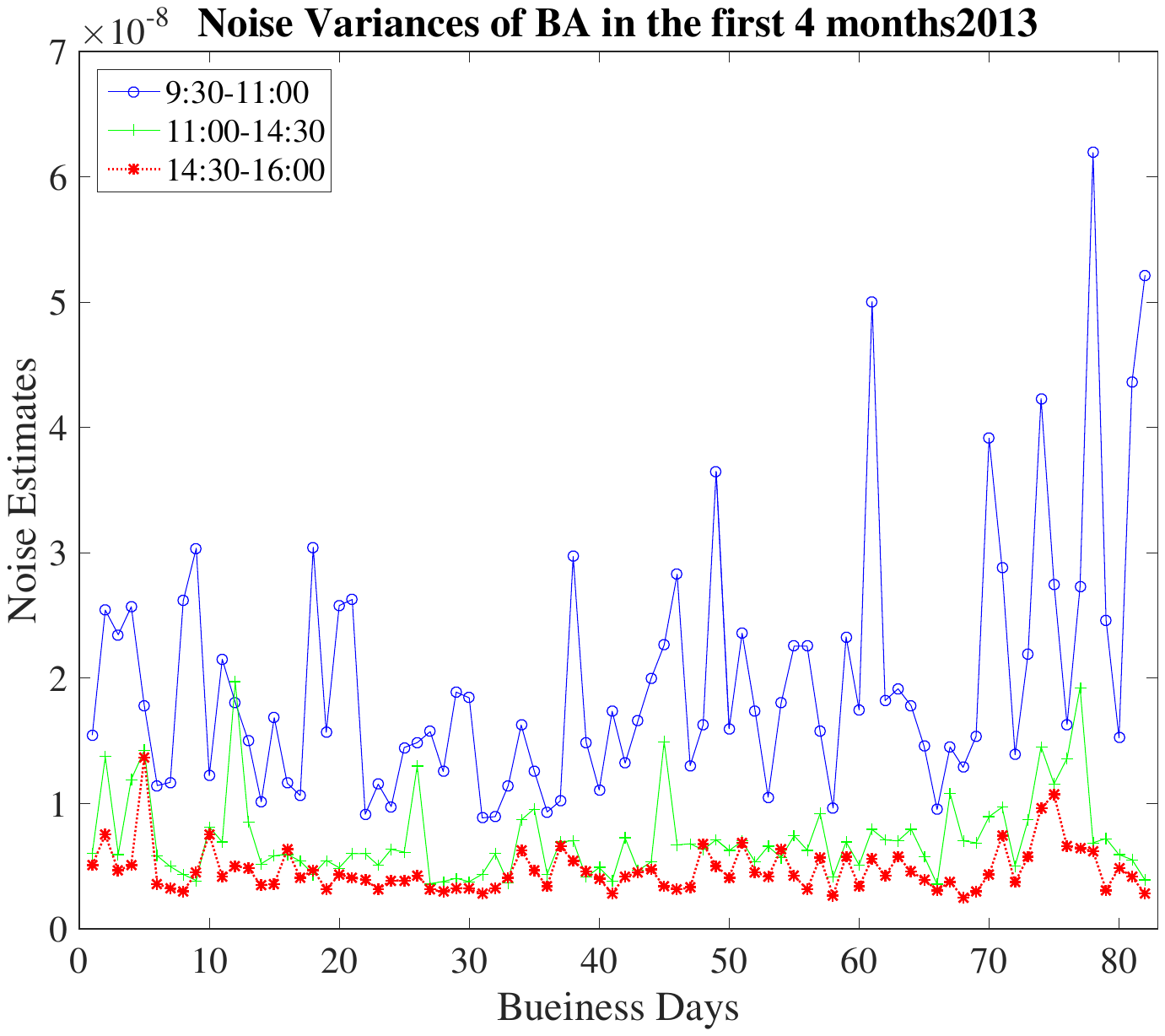}\includegraphics[width=.52\textwidth]{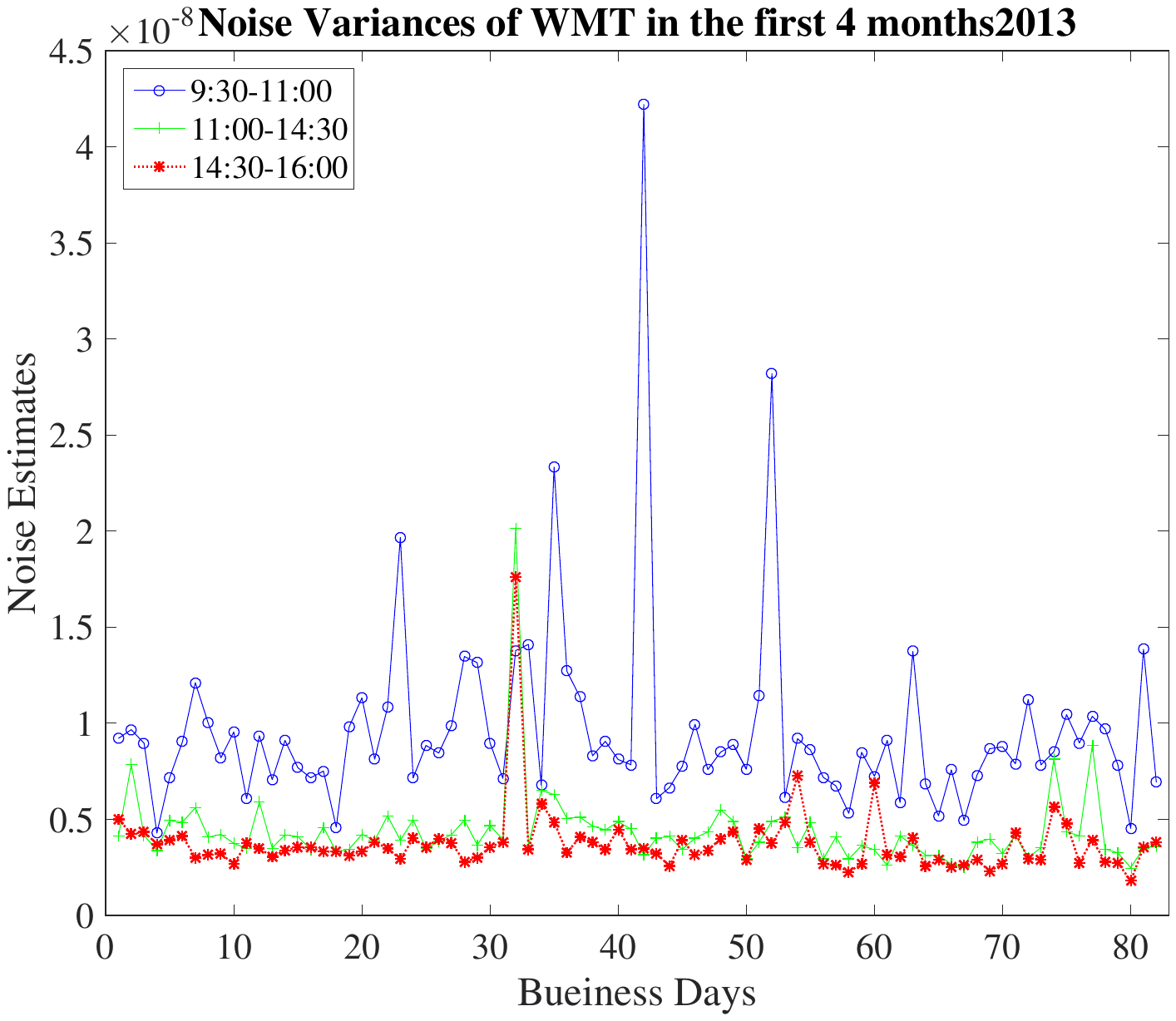}\\
	\caption{Temporal dynamics of $\widehat{E(\epsilon_t^2)}$ in different segments of trading hours 9:30-16:00 EST. For example, the \textcolor{blue}{blue line} is the time series plot of estimated noise level around morning across different business days in the first 4 months, 2013. From left to right and top to bottom, the companies are Goldman Sachs (GS), Chevron Corporation (CVX), Boeing (BA), Walmart (WMT).}\label{day_noise}
\end{figure}

\subsection{Empirical tests}\label{empitest}

In this subsection, we apply our tests onto high-frequency financial transaction data of stocks. We take several components in Dow Jones Industrial Average (DJIA30): Intel Corporation (INTC), International Business Machines Corporation (IBM), Goldman Sachs (GS), JPMorgan Chase (JPM), Exxon Mobil Corporation (XOM), General Electric (GE) and Walmart (WMT). We compute the test statistics and the p-values for these stocks during the 22 business days in April, which is shown in Table \ref{table_monthly_N}.  Besides, in Figure \ref{EmpriTest1}, we plot the whole trend of the test statistics $N(Y,K)_T^n$ during the period January 3, 2006 to December 31, 2013 as measures of liquidity.

\begin{landscape}
\begin{table}
\begin{tabular}{crl|rl|rl|rl|rl}
                 \cline{2-11}
Dates      & \multicolumn{2}{c}{\textbf{IBM}} & \multicolumn{2}{c}{\textbf{XOM}} & \multicolumn{2}{c}{\textbf{INTC}} & \multicolumn{2}{c}{\textbf{GS}} & \multicolumn{2}{c}{\textbf{GE}}\\
\cline{2-11}
yyyy-mm-dd & $N(Y,K)_T$ & p-value & $N(Y,K)_T$ & p-value & $N(Y,K)_T$ & p-value & $N(Y,K)_T$ & p-value & $N(Y,K)_T$ & p-value\\
\hline
2013-04-01 & 0.5942  & \textcolor{red}{0.2762}  & 6.0114  & 9.1947e-10 & 17.3676 & 0 & 0.9125 & \textcolor{red}{0.1807} & 6.4765  & 4.6925e-11\\
2013-04-02 & 3.8894  & 5.0246e-05 & 16.7202 & 0          & 12.3133 & 0 &  8.6813 & 0 & 26.2744 & 0\\
2013-04-03 & 6.8579  & 3.4941e-12 & 11.1238 & 0          & 12.6015 & 0 &  9.9089 & 0 &  4.5688 & 2.4529e-06\\
2013-04-04 & 4.5851  & 2.2690e-06 & 11.7737 & 0          & 11.8105 & 0 &  7.4771 & 3.7970e-14  & 8.3468 & 0\\
2013-04-05 & 8.6943  & 0          & 19.6103 & 0          & 21.9399 & 0 & 13.0797 & 0 & 7.7996 & 3.1086e-15\\
2013-04-08 & 12.0086 & 0          & 10.2720 & 0          & 19.6533 & 0 & 12.0044 & 0 & 8.7725 & 0\\
2013-04-09 & 4.4107  & 5.1507e-06 & 4.2196  & 1.0152e-05 & 14.5840 & 0 &  3.9217 & 4.3971e-05 & 2.8118 & 0.0025\\
2013-04-10 & 10.7967 & 0          & 20.3985 & 0          & 12.4934 & 0 &  1.4729 & \textcolor{red}{0.0704} & 12.1430 & 0\\
2013-04-11 & 10.5358 & 0          & 8.4332  & 0          & 19.8102 & 0 &  5.5467 & 1.4557e-08 & 7.6796 & 7.9936e-15\\
2013-04-12 &  9.8741 & 0          & 18.8744 & 0          &  9.7960 & 0 & 10.4689 & 0 & 11.3813 & 0\\
2013-04-15 &  8.6767 & 0          & 37.0635 & 0          & 11.8791 & 0 &  5.0028 & 2.8247e-07 & 6.6791 & 1.2023e-11\\
2013-04-16 & 11.5517 & 0          & 25.8213 & 0          & 11.0252 & 0 &  5.5612 & 1.3384e-08 & 16.5744 & 0\\
2013-04-17 & 11.2338 & 0          & 4.2163  & 1.2419e-05 & 20.6048 & 0 &  5.5168 & 1.7261e-08 & 13.6559 & 0\\
2013-04-18 & 15.1748 & 0          & 14.7396 & 0          & 49.2313 & 0 &  3.1477 & 8.2284e-04 & 10.6347 & 0\\
2013-04-19 & 29.7852 & 0          & 18.3013 & 0          & 10.8806 & 0 &  9.7611 & 0          & 18.5074 & 0\\
2013-04-22 & 13.4899 & 0          & 7.0150  & 1.1479e-12 & 10.0430 & 0 &  7.5659 & 1.9207e-14 & 12.9960 & 0\\
2013-04-23 & 11.0911 & 0          & 0.9798  & \textcolor{red}{0.1636} & 1.5144 & \textcolor{red}{0.065} & 0.4083 & \textcolor{red}{0.3415}  & 26.3066 & 0\\
2013-04-24 & 10.6420 & 0          & 26.4967 & 0          & 22.6824 & 0 & 9.6762 & 0 & 20.8122 & 0\\
2013-04-25 & 12.8092 & 0          & 13.4558 & 0          & 15.4190 & 0 & 9.4322 & 0 &  9.3956 & 0\\
2013-04-26 & 7.1480  & 4.4031e-13 & 14.8469 & 0          & 15.1904 & 0 & 3.0896 & 0.0010 & 6.0681 & 6.4723e-10\\
2013-04-29 & 3.4021  & 3.3438e-04 & 19.2697 & 0          &  0.8441 & \textcolor{red}{0.1993} & 9.3275 & 0 & -0.0481 & \textcolor{red}{0.4808}\\
2013-04-30 & 0.4047  & \textcolor{red}{0.3428}  & 12.8344 & 0   & -0.2676 & \textcolor{red}{0.3945} & 10.1050 & 0 & 7.4785 & 3.7637e-14\\
\hline
\end{tabular}
\caption{Test statistics $N(Y,K_n)_T$ for DJIA components ($T$ is 1 business day)}
\label{table_monthly_N}
\end{table}
\end{landscape}

\begin{landscape}
\begin{figure}
  \centering
  \includegraphics[width=1.0\textwidth]{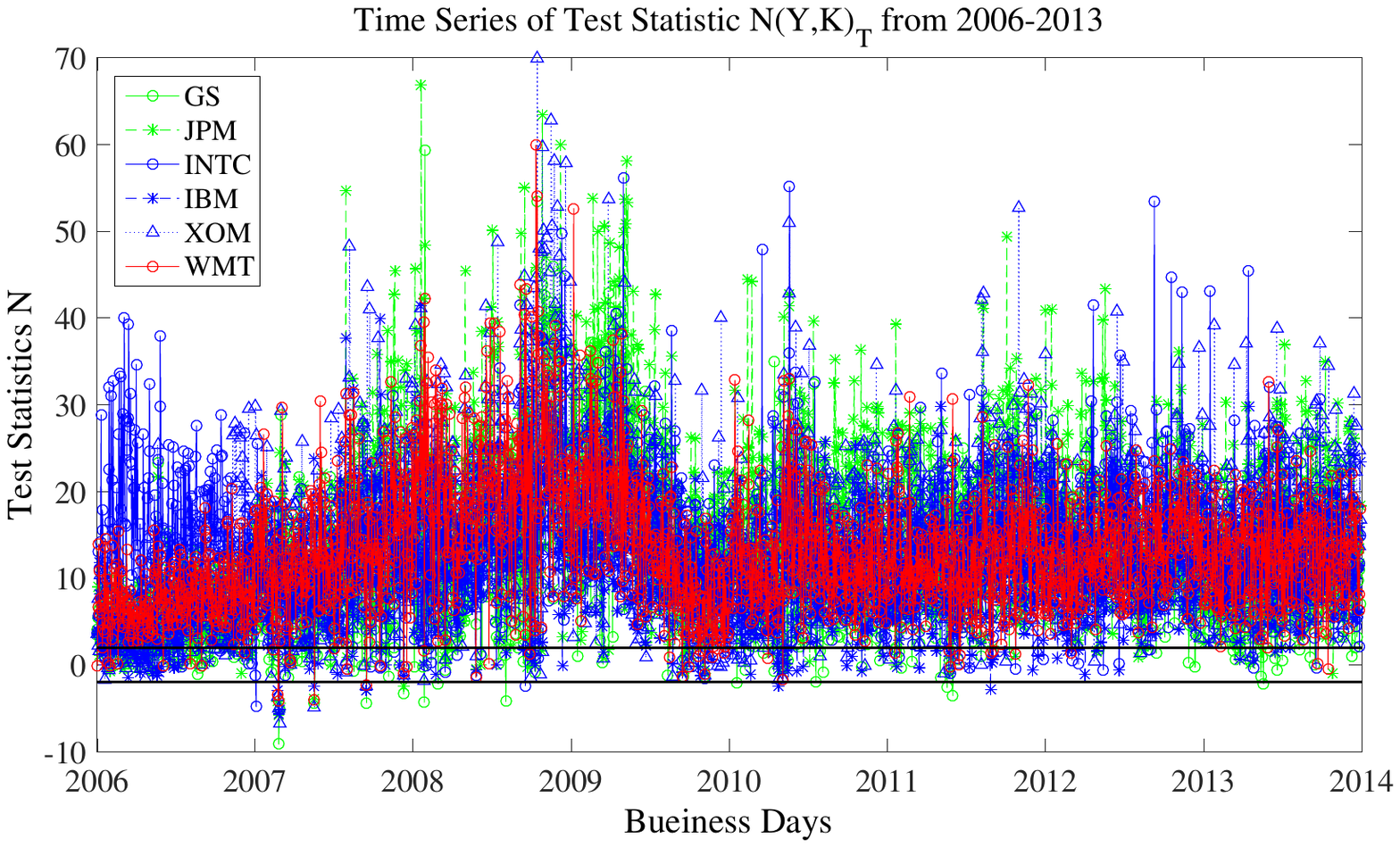}
  \caption{The six time series of the first test statistic $N(Y,K_n)^n$ computed daily using intra-day ultrahigh frequency financial data. The black horizontal lines around zero are .025\% and .975\% quantiles of the standard normal distribution.}\label{EmpriTest1}
\end{figure}
\end{landscape}

\begin{landscape}
\begin{figure}
  \centering
  \includegraphics[width=1.0\textwidth]{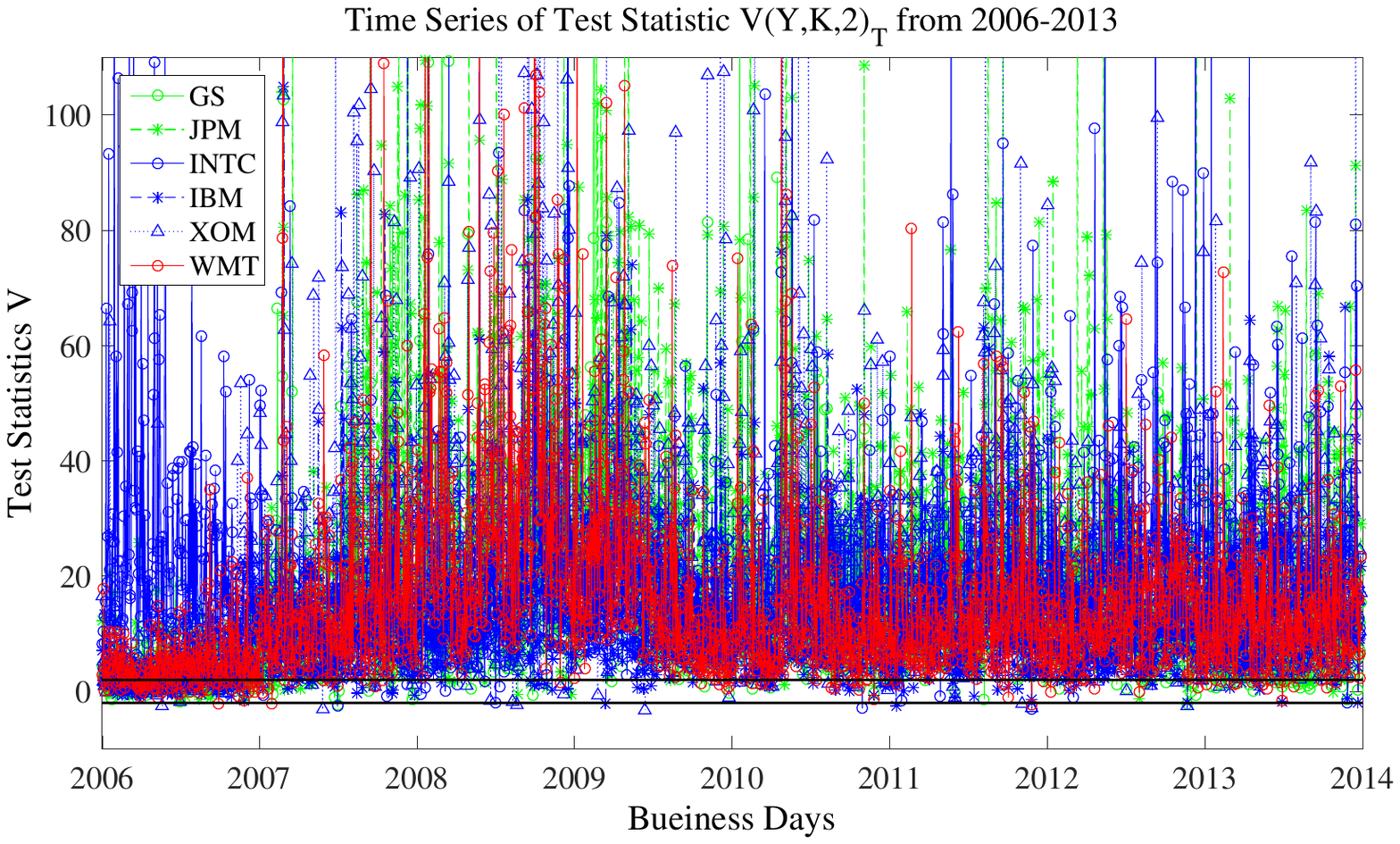}
  \caption{The six time series of the second test statistic $V(Y,K_n,s_n,2)^n$ computed daily using intra-day ultrahigh frequency financial data. The black horizontal lines around zero are .025\% and .975\% quantile of the standard normal distribution.}\label{EmpriTest2}
\end{figure}
\end{landscape}

\section{Conclusion}\label{conclu}

In this paper, we mainly concern hypothesis testing of microstructure noise stationarity in a hidden It\^o semimartingale model. The null hypothesis is that microstructure noise is stationary, and the alternative hypothesis is that microstructure noise is non-stationary with arbitrage dynamics up to a Markov kernel. Our tests work in fairly general settings where the latent It\^o semimartingale may have jumps with any degree of activity, the microstructure allows white noise and rounding error, and the observation times can be irregularly spaced.


The first test is motivated by the behavior of the two-scaled estimator (TSRV) under contamination of non-stationary noise, whose negative impact can be eliminated by a modification of TSRV \citep{kl08} under our general model. Based on the remedy for non-stationary microstructure noise, the first test $N(Y,K_n)^n$ is designed as a functional of volatility estimators, its type-I error can be controlled by associated central limit theorem under the null hypothesis. We also demonstrated that $N(Y,K_n)^n$ explodes in high-frequency asymptotics when microstructure noise is non-stationary.

Besides, we have other complementary tests, namely $V(Y,K_n,s_n,2)^n_T$ and $\overline{V}(Y,K_n,2)^n$. They are defined as functionals of $N(Y,K_n)^n$'s and realized variances, respectively, which are computed in different local time windows. $V(Y,K_n,s_n,2)^n$ and $\overline{V}(Y,K_n,2)^n$ are asymptotically equivalent and share the same convergence rate under the null hypothesis. Asymptotic approximation to $\overline{V}(Y,K_n,2)^n$ in finite sample is more accurate than that of $V(Y,K_n,s_n,2)^n$ under the null hypothesis, however, $V(Y,K_n,s_n,2)^n$ has more advantage under the alternative hypothesis in that it has a larger statistical power. Compared to $N(Y,K_n)^n$ which is more suited for 1-day data, $V(Y,K_n,s_n,2)^n_T$ and $\overline{V}(Y,K_n,2)^n$ are more suited for multi-day data. How to choose these complementary tests are discussed in detail.


Since microstructure noise could be a measure of the market quality (market liquidity, market depth, etc.) \citep{h93,o03,ay09}, our test statistics can be measures of liquidity risk. Particularly, assuming microstructure noise variance evolves like an It\^o diffusion, not only the asymptotic distributions of the test statistics under the alternative hypothesis are available, but also a notation of ``aggregate liquidity risk'' and a consistent estimator with an associated central limit theorem.

Some high-frequency financial data from NYSE are analyzed using the tests. As some DJIA components from 2006 to 2013 shows, variances of microstructure noise indeed changed both daily and intra-daily, which agrees with the empirical literature. Moreover, we find that the timing of the sudden increase in noise variance in Sep. 2008 coincided with the beginning of the global financial catastrophe triggered by the mortgage subprime crisis. The time series of our test statistics reveals a pattern which indicates increases in daily and weekly transaction costs during the financial turmoil.

\newpage
\section{Appendix}\label{proof}
\small
All the calculations are conditional on $\mathcal{F}^{(0)}$. Assuming \textbf{Proposition \ref{prop}} and \textbf{Lemma \ref{lem_RV}, \ref{lem_M}} which can also be found in \cite{zma05,lm07}:
\begin{prop}\label{prop}
	Assume that $E(|A_n||\mathcal{F}^{(0)})$ is $O_p(1)$. Then $A_n$ is $O_p(1)$.
\end{prop}
\begin{lem}\label{lem_RV}
	Under the model (\ref{X}), (\ref{Z}) and (\ref{epsilon}), and the assumptions in Section \ref{assumptions}, we have:
	\begin{eqnarray}
	[Y,Y]_{\mathcal{G}}&=&[\epsilon,\epsilon]_{\mathcal{G}}+O_p(1)\label{y2all}\\
	{[Y,Y]}_T^{(avg,K)}&=&[Z,Z]_T^{(avg,K)}+[\epsilon,\epsilon]_T^{(avg,K)}+O_p(1/\sqrt{K})
	\end{eqnarray}
\end{lem}

Besides, define $\mathcal{G}_{K+1}^{(\min)}$ as the right immediate neighbor of $\max{\mathcal{G}^{(\min)}}$ in the full grid $\mathcal{G}$, and define $\mathcal{G}_{0}^{(\max)}$ as the left immediate neighbor of $\min{\mathcal{G}^{(\max)}}$ in the full grid $\mathcal{G}$.

In order to describe the edge effect and the behaviors of some test statistics below, we need to introduce some random variables:
\begin{equation}\label{M_martingales}
	\begin{array}{ll}
	M^{(1)}_T\equiv\frac{1}{\sqrt{n}}\sum^n_{i=0}\left(\epsilon^2_{t_i}-g_{t_i}\right)\\
	M^{(2)}_T\equiv\frac{1}{\sqrt{n}}\sum^n_{i=1}\epsilon_{t_i}\epsilon_{t_{i-1}}&
	M^{(3)}_T\equiv\frac{1}{\sqrt{n}}\sum_{k=0}^{K-1}\sum_{t_i\in \mathcal{G}'^{(k)}}\epsilon_{t_i}\epsilon_{t_{i-K}}
	\end{array}
\end{equation}
and denote $h_t(\omega^{(0)})\equiv E(\epsilon^4_t|\mathcal{F}^{(0)})(\omega^{(0)})$. Note that $M^{(1)}_T$, $M^{(2)}_T$ and $M^{(3)}_T$ are the \textit{end-points of martingales} with respect to filtration $\mathcal{F}_i=\sigma(\epsilon_{t_j},j\le i; X_t, \forall t)$. By the argument from the Appendix A.2 in \cite{amz05}, we have
\begin{lem}\label{lem_M}
	$M_T^{(1)}$, $M_T^{(2)}$ and $M_T^{(3)}$ are asymptotically conditionally independent mixed normals, they have conditional variances $\frac{1}{T}\int_0^Th_t-g_t^2\,\mathrm{d}t$, $\frac{1}{T}\int_0^Tg_t^2\,\mathrm{d}t$, $\frac{1}{T}\int_0^Tg_t^2\,\mathrm{d}t$, respectively.
\end{lem}

\subsection{Robustness to jumps in noise inference}
In proving the testing theorems, namely, \textbf{Theorem \ref{thm2}}, \textbf{Theorem \ref{thm3}}, \textbf{Corollary \ref{corollary1}} and \textbf{Theorem \ref{thm4}}, we can assume the $J_t=0,\forall t>0$ in (\ref{X}) without loss of generality, as long as the noise is uncorrelated with neither the continuous part nor the jump part. 
Under \textbf{Assumptions \ref{diffusion}, \ref{jump}}, there are 3 components in the realized variance:
\begin{itemize}
	\item[(1)] finite quadratic variation of the discontinuous It\^o semimartingale $[X,X]_T=\langle X,X\rangle_T+\sum_{t\le T}|\Delta X_t|^2$, where $\Delta X_t=X_t-\lim_{s\nearrow t}
	X_s$ (a well-known result in stochastic calculus);
	\item[(2)] variation due to noise, which is of order $O_p(n)$;
	\item[(3)] asymptotically negligible terms, which are cross terms between noise, continuous martingale and jumps.
\end{itemize}
Under \textbf{Assumption \ref{independence}}, by a similar argument to those in the proof of lemma 1 in \cite{lm07}\footnote{Lemma 1 on p. 606 in \cite{lm07}.}, we have a result similar to \textbf{Lemma \ref{lem_RV}}:
\begin{equation}
[Y,Y]_{\mathcal{G}}=[\epsilon,\epsilon]_{\mathcal{G}}+\underbrace{2\sum_{i=1}^n(J_{t_i}-J_{t_{i-1}})(\epsilon_{t_i}-\epsilon_{t_{i-1}})+[X,X]_T}_{O_p(1)}+O_p(1)
\end{equation}
which suggests that normalized realized variance of the fastest time scale $\frac{1}{2n}[Y,Y]_{\mathcal{G}}$ is a consistent estimator of $E(\epsilon^2|\mathcal{F}^{(0)})$ provided the noise is stationary even if there exist jumps, i.e., \textbf{Lemma \ref{lem_RV}} still holds. For this reason, the asymptotic distributions remain the same for the test statistics, even if jump is present.

\subsection{Proof of Lemma \ref{lem_edge}}\label{deriavg}
\begin{proof} 
By our assumptions, we can write
\begin{eqnarray*}
[Y,Y]_T^{(avg,K)}
&=&[Z,Z]_T^{(avg,K)}+\frac{2\sqrt{n}}{K}\left(M_T^{(1)}-M_T^{(3)}\right)\\&&+\frac{2}{K}\sum^K_{k=1}\sum_{t_i\in\mathcal{G}''^{(k)}}g_{t_i}+\frac{1}{K}\sum_{t_i\in\mathcal{G}^{(\min)}}g_{t_i}+\frac{1}{K}\sum_{t_i\in\mathcal{G}^{(\max)}}g_{t_i}+O_p\left(\frac{1}{\sqrt{K}}\right)
\end{eqnarray*}

By the conditional Lyapunov condition, 
and \textbf{Lemma \ref{lem_M}}
\begin{equation*}
M_T^{(1)}-M_T^{(3)}\overset{\mathbb{P}}{\longrightarrow}\mathcal{MN}\left(0,\frac{1}{T}\int_0^Th_t\,\mathrm{d}t\right)
\end{equation*}
thus (\ref{edge_avg}) follows.
\end{proof}

\subsection{Proof of Theorem \ref{thm1}}\label{prfthm1}
\begin{proof}
Asymptotically\footnote{The caveat is $\mathcal{G}$ and $\mathcal{G}^{(\min)}\bigcup\left(\bigcup_{k=1}^K\mathcal{G}''^{(k)}\right)\bigcup\mathcal{G}^{(\max)}$ might not equal, the difference is $\left\{t_{\left\lfloor n/K\right\rfloor\cdot K+1},\cdots,t_n\right\}$. Whereas, upon an appropriate choice of $K$, this difference is asymptotically negligible.}, the new version of realized variance in \cite{kl08} can be written as follows:
\begin{equation*}
[Y,Y]^{\{n\}}_T=\frac{1}{2}[Y,Y]_{\mathcal{G}^{(\min)}}+[Y,Y]_{\bigcup^K_{k=1}\mathcal{G}''^{(k)}}+\frac{1}{2}[Y,Y]_{\mathcal{G}^{(\min)}}
\end{equation*}
Since for any grid $\mathcal{H}$, $[Y,Y]_{\mathcal{H}}=[Z,Z]_{\mathcal{H}}+2[Z,\epsilon]_{\mathcal{H}}+[\epsilon,\epsilon]_{\mathcal{H}}$, we have
\begin{eqnarray*}
[Y,Y]^{\{n\}}_T
=[Z,\epsilon]_{\mathcal{G}^{(\min)}}+2[Z,\epsilon]_{\bigcup^K_{k=1}\mathcal{G}''^{(k)}}+[Z,\epsilon]_{\mathcal{G}^{(\max)}}+\frac{1}{2}[\epsilon,\epsilon]_{\mathcal{G}^{(\min)}}+[\epsilon,\epsilon]_{\bigcup^K_{k=1}\mathcal{G}''^{(k)}}+\frac{1}{2}[\epsilon,\epsilon]_{\mathcal{G}^{(\max)}}+O_p(1)
\end{eqnarray*}
Note that
$[Z,\epsilon]_{\mathcal{G}^{(\min)}}+2[Z,\epsilon]_{\bigcup^K_{k=1}\mathcal{G}''^{(k)}}+[Z,\epsilon]_{\mathcal{G}^{(\max)}}\le 2[Z,\epsilon]_{\mathcal{G}}$.
Define $\Delta Z_{t_i}=Z_{t_i}-Z_{t_{i-1}}$, then
\begin{eqnarray*}
E\left(\left([Z,\epsilon]_{\mathcal{G}}\right)^2I_{\{\tau_l>T\}}|\mathcal{F}^{(0)}\right)
&=&I_{\{\tau_l>T\}}\sum^n_{i=1}\sum^n_{j=1}\Delta Z_{t_i}\Delta Z_{t_j}E\left[\left(\epsilon_{t_i}-\epsilon_{t_{i-1}}\right)\left(\epsilon_{t_j}-\epsilon_{t_{j-1}}\right)|\mathcal{F}^{(0)}\right]
\end{eqnarray*}

By assumption, the noises are mutually independent conditioning on the whole path of latent process $X$, thus
\begin{equation*}
E\left[(\epsilon_{t_i}-\epsilon_{t_{i-1}})(\epsilon_{t_j}-\epsilon_{t_{j-1}})|\mathcal{F}^{(0)}\right]=
\left\{\begin{array}{rl}
-g_{t_{i\wedge j}},     & |i-j|=1\\
 g_{t_{i-1}}+g_{t_i},   & j=i\\
 0,                     & \text{otherwise}
\end{array}\right.
\end{equation*}
So, if $\tau_l>T$, we have
\begin{multline*}
\sum^n_{i=1}\sum^n_{j=1}\Delta Z_{t_i}\Delta Z_{t_j}E\left[\left(\epsilon_{t_i}-\epsilon_{t_{i-1}}\right)\left(\epsilon_{t_j}-\epsilon_{t_{j-1}}\right)|\mathcal{F}^{(0)}\right]\\
=\sum^{n-1}_{i=0}\left(\Delta Z_{t_{i+1}}\right)^2g_{t_i}+\sum_{i=1}^n\left(\Delta Z_{t_i}\right)^2g_{t_i}-2\sum_{i=1}^{n-1}\Delta Z_{t_i}\Delta Z_{t_{i+1}}g_{t_i}
\le4M_{(2,l)}\cdot[Z,Z]_{\mathcal{G}}=O_p(1)
\end{multline*}
by \textbf{Proposition \ref{prop}} and $\mathbb{P}^{(0)}\{\tau_l>T\}\longrightarrow 1$ as $l\longrightarrow \infty$, we know $[Z,\epsilon]_{\mathcal{G}}=O_p(1)$. So, the following relation holds:
\begin{eqnarray*}
[Y,Y]^{\{n\}}_T
&=&[\epsilon,\epsilon]_{\mathcal{G}}-\frac{1}{2}\left([\epsilon,\epsilon]_{\mathcal{G}^{(\min)}}+[\epsilon,\epsilon]_{\mathcal{G}^{(\max)}}\right)+O_p(1)
\end{eqnarray*}
By our assumption, we have the following:
\begin{eqnarray}
	K[\epsilon,\epsilon]_T^{(avg,K)}&=&2\sqrt{n}\left(M_T^{(1)}-M_T^{(3)}\right)+2\sum_{k=1}^K\sum_{t_i\in\mathcal{G}'^{(k)}}g_{t_i}+\sum_{t_i\in\mathcal{G}^{(\min)}}g_{t_i}-\sum_{t_i\in\mathcal{G}^{(\max)}}g_{t_i}+O_p(\sqrt{K})\label{epsilon_avg}\\{}
	[\epsilon,\epsilon]_{\mathcal{G}}&=&2\sqrt{n}\left(M_T^{(1)}-M_T^{(2)}\right)+2\sum^n_{i=0}g_{t_i}+O_p(1)\label{epsilon_all}
\end{eqnarray}
Define the following quantities:
\begin{equation}\label{m_martingales}
	\begin{array}{ll}
	\underline{m}_T^{(1)}\equiv\frac{1}{\sqrt{K}}\sum_{i=1}^K\left(\epsilon^2_{\mathcal{G}_i^{(\min)}}-g_{\mathcal{G}_i^{(\min)}}\right)&
	\underline{m}_T^{(2)}\equiv\frac{1}{\sqrt{K}}\sum_{i=1}^K\epsilon_{\mathcal{G}_{i+1}^{(\min)}}\epsilon_{\mathcal{G}_i^{(\min)}}\\
	\overline{m}^{(1)}_T\equiv\frac{1}{\sqrt{K}}\sum_{i=1}^K\left(\epsilon^2_{\mathcal{G}_i^{(\max)}}-g_{\mathcal{G}_i^{(\max)}}\right)
	&
	\overline{m}^{(2)}_T\equiv\frac{1}{\sqrt{K}}\sum_{i=1}^K\epsilon_{\mathcal{G}_i^{(\max)}}\epsilon_{\mathcal{G}_{i-1}^{(\max)}}
	\end{array}
\end{equation}
Similarly to (\ref{epsilon_all}), 
\begin{eqnarray*}
[\epsilon,\epsilon]_{\mathcal{G}^{(\min)}}
&=&2\sqrt{K}\left(\underline{m}_T^{(1)}-\underline{m}_T^{(2)}\right)+2\sum_{t_i\in\mathcal{G}^{(\min)}}g_{t_i}+O_p(1)\\{}
[\epsilon,\epsilon]_{\mathcal{G}^{(\max)}}
&=&2\sqrt{K}\left(\overline{m}_T^{(1)}-\overline{m}_T^{(2)}\right)+2\sum_{t_i\in\mathcal{G}^{(\max)}}g_{t_i}+O_p(1)
\end{eqnarray*}

Combine \textbf{Lemma \ref{lem_RV}} with these results, the difference between sample-weighted TSRV and the averaged realized variance of the theoretical process $Z$ is
\begin{multline}\label{WTSRV}
\widehat{\langle X,X\rangle}^{(WTSRV,K)}_T-[Z,Z]_T^{(avg,K)}\\=\frac{2\sqrt{n}}{K}\left(M_T^{(2)}-M_T^{(3)}\right)+\frac{1}{\sqrt{K}}\left(\underline{m}_T^{(1)}-\underline{m}_T^{(2)}+\overline{m}_T^{(1)}-\overline{m}_T^{(2)}\right)+O_p\left(\frac{1}{\sqrt{K}}\right)
\end{multline}
therefore
\begin{multline*}
\frac{K}{\sqrt{n}}\left(\widehat{\langle X,X\rangle}^{(WTSRV,K)}_T-[Z,Z]_T^{(avg,K)}\right)
=2\left(M_T^{(2)}-M_T^{(3)}\right)+o_p(1)\overset{\mathcal{L}-s}{\longrightarrow} \mathcal{MN}\left(0,\frac{8}{T}\int_0^Tg_t^2\,\mathrm{d}t\right)
\end{multline*}

The remaining argument discussing the error term due to discretization $[Z,Z]_T^{(avg)}-\langle Z,Z\rangle_T$ to which an identical technique in Appendix A.3 in \cite{zma05} applies, whence we get the claim of \textbf{Theorem \ref{thm1}}.

\end{proof}

\subsection{Proof of Theorem \ref{thm2}}\label{prfthm2}

\begin{proof}
Recall the definition (\ref{m_martingales}), by (\ref{WTSRV}) and the asymptotic behavior of the original TSRV \citep{amz05,lm07}, under the null $H_0$ we have
\begin{multline*}\label{difTSRV}
\widehat{\langle X,X\rangle}_T^{(WTSRV,K)}-\widehat{\langle X,X\rangle}_T^{(TSRV,K)}\\
=\frac{2(K-1)}{K\sqrt{n}}\left(M_T^{(2)}-M_T^{(1)}\right)+\frac{1}{\sqrt{K}}\left(\underline{m}_T^{(1)}-\underline{m}_T^{(2)}+\overline{m}_T^{(1)}-\overline{m}_T^{(2)}\right)+O_p\left(\frac{1}{K}\right)
\end{multline*}
i.e., when $H_0$ holds, we have $N(Y,K)^n_T=\left(\underline{m}_T^{(1)}-\underline{m}_T^{(2)}+\overline{m}_T^{(1)}-\overline{m}_T^{(2)}\right)+o_p(1)$.

To prove the limiting distribution is normal, again, we will use martingale limit central by exploiting the discrete predictable quadratic variations of those in (\ref{m_martingales}).
\begin{eqnarray*}
	\langle\underline{m}_T^{(1)},\underline{m}_T^{(1)}\rangle_T|\mathcal{F}^{(0)}
	&=&\frac{1}{K}\sum^K_{i=0}\left[E\left(\epsilon^4_{\mathcal{G}_i^{(\min)}}|\mathcal{F}^{(1)}_{\mathcal{G}_{i-1}^{(\min)}},\mathcal{F}^{(0)}\right)-g_{\mathcal{G}_i^{(\min)}}^2\right]\\
	&=&\frac{1}{K}\sum^K_{i=0}\left[h_{\mathcal{G}_i^{(\min)}}-g_{\mathcal{G}_i^{(\min)}}^2\right]\longrightarrow h_0-g_0^2
\end{eqnarray*}
By the same calculation, we know $\langle\overline{m}_T^{(1)},\overline{m}_T^{(1)}\rangle_T|\mathcal{F}^{(0)}=\frac{1}{K}\sum^n_{i=0}\left[h_{\mathcal{G}_i^{(\max)}}-g_{\mathcal{G}_i^{(\max)}}^2\right]\longrightarrow h_T-g_T^2$.
Besides,
\begin{eqnarray*}
\langle\underline{m}^{(2)},\underline{m}^{(2)}\rangle_T|\mathcal{F}^{(0)}
&=&\frac{1}{K}\sum^K_{i=1}\epsilon^2_{\mathcal{G}_i^{(\min)}}E\left(\epsilon^2_{\mathcal{G}_{i+1}^{(\min)}}|\mathcal{F}_{\mathcal{G}_i^{(\min)}}^{(1)},\mathcal{F}^{(0)}\right)=\frac{1}{K}\sum^K_{i=1}\epsilon^2_{\mathcal{G}_i^{(\min)}}g_{\mathcal{G}_{i+1}^{(\min)}}\\
&=&\frac{1}{K}\sum^K_{i=1}\left[\epsilon^2_{\mathcal{G}_i^{(\min)}}-g_{\mathcal{G}_i^{(\min)}}\right]g_{\mathcal{G}_{i+1}^{(\min)}}+\frac{1}{K}\sum^K_{i=1}g_{\mathcal{G}_i^{(\min)}}g_{\mathcal{G}_{i+1}^{(\min)}}
\end{eqnarray*}
Since $\mathcal{G}^{(\min)}\longrightarrow0$ is a shrinking sub-grid, so $\frac{1}{K}\sum^K_{i=1}g_{\mathcal{G}_i^{(\min)}}g_{\mathcal{G}_{i+1}^{(\min)}}\longrightarrow g^2_0$. Besides,
\begin{multline*}
E\left(\left(\frac{1}{K}\sum^K_{i=1}(\epsilon^2_{\mathcal{G}_i^{(\min)}}-g_{\mathcal{G}_i^{(\min)}})g_{\mathcal{G}_{i+1}^{(\min)}}\mathbf{1}_{\{\tau_l>T\}}\right)^2|\mathcal{F}^{(0)}\right)=\\
=\frac{1}{K^2}\sum^K_{i=1}E\left(\left(\epsilon^2_{\mathcal{G}_i^{(\min)}}-g_{\mathcal{G}_i^{(\min)}}\right)^2|\mathcal{F}^{(0)}\right)g_{\mathcal{G}_{i+1}^{(\min)}}^2\mathbf{1}_{\{\tau_l>T\}}
\le\frac{1}{K^2}\sum^K_{i=1}M_{(4,l)}\cdot M^2_{(2,l)}=O_p\left(\frac{1}{K}\right)
\end{multline*}
by \textbf{Proposition \ref{prop}} and the fact that $P(\tau_l>T)\overset{\mathbb{P}}{\longrightarrow} 1$ as $l\longrightarrow \infty$, we know $\frac{1}{K}\sum^K_{i=1}(\epsilon^2_{\mathcal{G}_i^{(\min)}}-g_{\mathcal{G}_i^{(\min)}})g_{\mathcal{G}_{i+1}^{(\min)}}\overset{\mathbb{P}}{\longrightarrow}0$, hence $\langle \underline{m}^{(2)},\underline{m}^{(2)}\rangle_T\longrightarrow g_0^2$. Similarly, $\langle \overline{m}^{(2)},\overline{m}^{(2)}\rangle_T\longrightarrow g_T^2$. By martingale central limit theorem, we know $\underline{m}_T^{(1)}$, $\overline{m}_T^{(1)}$, $\underline{m}_T^{(2)}$, $\overline{m}_T^{(2)}$ are asymptotically mixed normal.
Because $\mathcal{G}^{(\min)}$ and $\mathcal{G}^{(\max)}$ are disjoint sets of observation times, so $\underline{m}_T^{(2)}$ and $\overline{m}_T^{(2)}$, or $\underline{m}_T^{(2)}$ and $\underline{m}_T^{(2)}$ are independent conditional on $\mathcal{F}^{(0)}$, so $\underline{m}_T^{(1)}+\overline{m}_T^{(1)}\overset{\mathcal{L}-s}{\longrightarrow}\mathcal{MN}\left(0,h_0-g_0^2+h_T-g_T^2\right)$ and $\underline{m}_T^{(2)}+\overline{m}_T^{(2)}\overset{\mathcal{L}-s}{\longrightarrow}\mathcal{MN}\left(0,g_0^2+g_T^2\right)$.

Furthermore,
\begin{eqnarray*}
\langle \underline{m}^{(1)}, \underline{m}^{(2)}\rangle_T|\mathcal{F}^{(0)}
&=&\frac{1}{K}\sum_{i=1}^KE\left[\left(\epsilon^2_{\mathcal{G}_i^{(\min)}}-g_{\mathcal{G}_i^{(\min)}}\right)\cdot\left(\epsilon_{\mathcal{G}_i^{(\min)}}\epsilon_{\mathcal{G}_{i+1}^{(\min)}}\right)|\mathcal{F}^{(0)},\mathcal{F}^{(1)}_{\mathcal{G}_i^{(\min)}}\right]\\
&=&\frac{1}{K}\sum_{i=1}^K\left[\left(\epsilon^2_{\mathcal{G}_i^{(\min)}}-g_{\mathcal{G}_i^{(\min)}}\right)\epsilon_{\mathcal{G}_i^{(\min)}}\cdot E\Big(\epsilon_{\mathcal{G}_{i+1}^{(\min)}}|\mathcal{F}^{(1)}_{\mathcal{G}_i^{(\min)}}\Big)|\mathcal{F}^{(0)}\right]=0
\end{eqnarray*}
and $\langle \overline{m}^{(1)}, \overline{m}^{(2)}\rangle_T|\mathcal{F}^{(0)}
=\frac{1}{K}\sum_{i=1}^K\epsilon_{\mathcal{G}_{i-1}^{(\max)}}E\left(\epsilon^3_{\mathcal{G}_i^{(\max)}}|\mathcal{F}^{(0)},\mathcal{F}^{(1)}_{\mathcal{G}_{i-1}^{(\max)}}\right)$.
\begin{eqnarray*}
E\left(\left(\langle\overline{m}^{(1)},\overline{m}^{(2)}\rangle_T\right)^2\cdot I_{\{\tau_l>T\}}|\mathcal{F}^{(0)}\right)
&=&\frac{1}{K^2}\sum_{i=1}^Kg_{\mathcal{G}_{i-1}^{(\max)}}\cdot \left(E\left(E(\epsilon^3_{\mathcal{G}_i^{(\max)}}|\omega^{(0)})|\mathcal{F}^{(1)}_{\mathcal{G}_{i-1}^{(\max)}}\right)\right)^2\cdot I_{\{\tau_l>T\}}\\
&\le&\frac{1}{K}\cdot M_{(2,l)}\cdot M^2_{(3,l)}\cdot I_{\{\tau_l>T\}}
\end{eqnarray*}
by \textbf{Proposition \ref{prop}} and the fact that $\mathbb{P}\{\tau_l> T\}\longrightarrow 1$ as $l\longrightarrow \infty$, we know $\langle \overline{m}^{(1)},\overline{m}^{(2)}\rangle_T=O_p\left(\frac{1}{\sqrt{K}}\right)\overset{\mathbb{P}}{\longrightarrow}0$. Thus, $\underline{m}_T^{(1)}$, $\underline{m}_T^{(2)}$, $\overline{m}_T^{(1)}$ and $\overline{m}_T^{(2)}$ are asymptotically independent and mixed normals, and
\begin{equation*}
\underline{m}_T^{(1)}-\underline{m}_T^{(2)}+\overline{m}_T^{(1)}-\overline{m}_T^{(2)}\overset{\mathcal{L}-s}{\longrightarrow}\mathcal{MN}(0,h_0+h_T)
\end{equation*}
\end{proof}

\subsection{Proof of Lemma \ref{lem_epsilon4}}\label{prflem_epsilon4}
\begin{proof}
Since $\{Z_t\}_{t\ge0}$ is an It\^o semimartingale
\begin{multline*}
[Y;4]_{\mathcal{G}}
=\sum^n_{i=1}(\epsilon_{t_i}-\epsilon_{t_{i-1}})^4+4\sum^n_{i=1}(\epsilon_{t_i}-\epsilon_{t_{i-1}})^3(Z_{t_i}-Z_{t_{i-1}})+
4\sum^n_{i=1}(\epsilon_{t_i}-\epsilon_{t_{i-1}})(Z_{t_i}-Z_{t_{i-1}})^3+O_p(1)
\end{multline*}
\begin{multline*}
E\left[\left(\sum^n_{i=1}(\epsilon_{t_i}-\epsilon_{t_{i-1}})^3(Z_{t_i}-Z_{t_{i-1}})\right)^2\mathbf{1}_{\{\tau_l>T\}}|\mathcal{F}^{(0)}\right]=\sum^n_{i=1}E\left[(\epsilon_{t_i}-\epsilon_{t_{i-1}})^6(Z_{t_i}-Z_{t_{i-1}})^2|\mathcal{F}^{(0)}\right]\mathbf{1}_{\{\tau_l>T\}}\\
\le\sum^n_{i=1}E\left[(\epsilon_{t_i}-\epsilon_{t_{i-1}})^{12}|\mathcal{F}^{(0)}\right]^{1/2}\cdot E\left[(Z_{t_i}-Z_{t_{i-1}})^4|\mathcal{F}^{(0)}\right]^{1/2}\mathbf{1}_{\{\tau_l>T\}}\\
\le\left(990\right)^{1/2}\left(\max_{2\le k\le12} M_{(k,l)}\vee 1\right)\mathbf{1}_{\{\tau_l>T\}}\sum^n_{i=1}E\left[(Z_{t_i}-Z_{t_{i-1}})^4|\mathcal{F}^{(0)}\right]^{1/2}=O_p(1)
\end{multline*}
\begin{multline*}
E\left[\left(\sum^n_{i=1}(\epsilon_{t_i}-\epsilon_{t_{i-1}})(Z_{t_i}-Z_{t_{i-1}})^3\right)^2\mathbf{1}_{\{\tau_l>T\}}|\mathcal{F}^{(0)}\right]=\sum^n_{i=1}E\left[(\epsilon_{t_i}-\epsilon_{t_{i-1}})^2(Z_{t_i}-Z_{t_{i-1}})^6|\mathcal{F}^{(0)}\right]\mathbf{1}_{\{\tau_l>T\}}\\
\le\sum^n_{i=1}E\left[(\epsilon_{t_i}-\epsilon_{t_{i-1}})^4|\mathcal{F}^{(0)}\right]^{1/2}\cdot E\left[(Z_{t_i}-Z_{t_{i-1}})^{12}|\mathcal{F}^{(0)}\right]^{1/2}\mathbf{1}_{\{\tau_l>T\}}\\
\le \sqrt{6}\left(\max_{2\le k\le4} M_{(k,l)}\vee 1\right)\mathbf{1}_{\{\tau_l>T\}}\cdot\sum^n_{i=1}E\left[(Z_{t_i}-Z_{t_{i-1}})^{12}|\mathcal{F}^{(0)}\right]^{1/2}=o_p(1)
\end{multline*}
Thus, $[Y;4]_{\mathcal{G}}=\sum^n_{i=1}(\epsilon_{t_i}-\epsilon_{t_{i-1}})^4+O_p(1)$. 
Note that
\begin{equation}
\sum^n_{i=1}(\epsilon_{t_i}-\epsilon_{t_{i-1}})^4=2\sum^n_{i=1}h_{t_i}+6\sum^n_{i=1}g_{t_{i-1}}g_{t_i}+\sqrt{n}\left(2L^{(1)}_T+6L^{(2)}_T-4L^{(3)}_T-4L^{(4)}_T\right)+O_p(1)
\end{equation}
where
\begin{equation*}
\begin{array}{ll}
L^{(1)}_T=\frac{1}{\sqrt{n}}\sum^n_{i=1}\left[\epsilon^4_{t_i}-h_{t_i}\right]&
L^{(2)}_T=\frac{1}{\sqrt{n}}\sum^n_{i=1}\left[\epsilon^2_{t_{i-1}}\epsilon^2_{t_i}-E(\epsilon^2_{t_{i-1}}\epsilon^2_{t_i}|\mathcal{F}^{(0)})\right]\\
L^{(3)}_T=\frac{1}{\sqrt{n}}\sum^n_{i=1}\epsilon^3_{t_{i-1}}\epsilon_{t_i}&
L^{(4)}_T=\frac{1}{\sqrt{n}}\sum^n_{i=1}\epsilon_{t_{i-1}}\epsilon^3_{t_i}
\end{array}
\end{equation*}
We can show that $L^{(1)}_T$, $L^{(2)}_T$, $L^{(3)}_T$ and $L^{(4)}_T$ are mixed normals. Observe that
\begin{eqnarray*}
\frac{1}{n}\sum^n_{i=1}h_{t_i}&=&\frac{1}{T}\sum^n_{i=1}h_{t_i}\frac{T}{n}\longrightarrow\frac{1}{T}\int_0^Th_t\,\mathrm{d}t\\
\frac{1}{n}\sum^n_{i=1}g_{t_{i-1}}g_{t_i}&=&\frac{1}{T}\sum^n_{i=1}g_{t_{i-1}}g_{t_i}\frac{T}{n}\longrightarrow\frac{1}{T}\int_0^Tg^2_t\,\mathrm{d}t
\end{eqnarray*}
then (\ref{eps4}) follows.
\end{proof}

For each $i=1,2,\cdots,r_n$, we define $m_i\equiv m_i^{(1)}-m_i^{(2)}$ where
\begin{equation}\label{m_i}
\begin{array}{ll}
      m_i^{(1)}&\equiv\frac{1}{\sqrt{K_n}}\sum_{k=1}^{K_n}\left(\epsilon^2_{t_{(i-1)K_n+k}}-g_{t_{(i-1)K_n+k}}\right)\\
      m_i^{(2)}&\equiv\frac{1}{\sqrt{K_n}}\sum_{k=1}^{K_n}\epsilon_{t_{(i-1)K_n+k}}\epsilon_{t_{(i-1)K_n+k-1}}
\end{array}
\end{equation}
To prove \textbf{Theorem \ref{thm3}}, we need an additional lemma:
\begin{lem}\label{lemma4} Assume the microstructure noise is stationary, and under the moment assumptions on the noise process $\{\epsilon_t\}_{t\ge 0}$, we have for each $i\in\{1,2,\cdots,r_n\}$,
\begin{eqnarray*}
	E(m^2_i|\mathcal{F}^{(0)})&=&E(\epsilon^4|\mathcal{F}^{(0)})\\
	E(m^4_i|\mathcal{F}^{(0)})&=&6\left[E(\epsilon^4|\mathcal{F}^{(0)})^2-E(\epsilon^4|\mathcal{F}^{(0)})E(\epsilon^2|\mathcal{F}^{(0)})^2+E(\epsilon^2|\mathcal{F}^{(0)})^4\right]+O_p\left(\frac{1}{K}\right)
\end{eqnarray*}
\end{lem}
\begin{proof}
For the ease of notation, let us suppress the notation $K=K_n$, and denote $\epsilon_{(i-1)K_n+k}$ by $\xi_{i,k}$, and $g_{(i-1)K_n+k}$ by $g_{i,k}$ for each $i\in\{1,2,\cdots,r_n\}$ and $k\in\{0,1,2,\cdots,K\}$. Note that under our new notation
\begin{equation*}
\begin{array}{ll}
     m^{(1)}_i\equiv\frac{1}{\sqrt{K}}\sum^K_{k=1}\left(\xi^2_{i,k}-g_{i,k}\right)&
     m^{(2)}_i\equiv\frac{1}{\sqrt{K}}\sum^K_{k=1}\xi_{i,k-1}\xi_{i,k}
\end{array}
\end{equation*}
and
\begin{equation*}
\begin{array}{ll}
      E\left[\left(m_i^{(1)}\right)^2|\mathcal{F}^{(0)}\right]=&
      \frac{1}{K}\sum^K_{k=1}E\left[\xi^4_{i,k}-g^2_{i,k}|\mathcal{F}^{(0)}\right]
      =E(\epsilon^4|\mathcal{F}^{(0)})-E(\epsilon^2|\mathcal{F}^{(0)})^2\\
      E\left[\left(m_i^{(2)}\right)^2|\mathcal{F}^{(0)}\right]=&
      \frac{1}{K}\sum^K_{k=1}E\left[\xi^2_{i,k-1}\xi^2_{i,k}|\mathcal{F}^{(0)}\right]
      =E(\epsilon^2|\mathcal{F}^{(0)})^2\\
      E\left[m_i^{(1)}m_i^{(2)}|\mathcal{F}^{(0)}\right]
      =&\frac{1}{K}\sum^K_{k=1}\sum^K_{j=1}E\left[\left(\xi^2_{i,k}-g_{i,k}\right)\cdot\xi_{i,j-1}\cdot\xi_{i,j}|\mathcal{F}^{(0)}\right]=0
\end{array}
\end{equation*}
thus $E(m^2_i|\mathcal{F}^{(0)})=E(\epsilon^4|\mathcal{F}^{(0)})$. Note that
\begin{eqnarray*}
	m_i^2&=&\left(m_i^{(1)}\right)^2+\left(m_i^{(2)}\right)^2-2m_i^{(1)}m_i^{(2)}\\
	m_i^4&=&\left(m_i^{(1)}\right)^4-4\left(m_i^{(1)}\right)^3m_i^{(2)}+6\left(m_i^{(1)}\right)^2\left(m_i^{(2)}\right)^2-4m_i^{(1)}\left(m_i^{(2)}\right)^3+\left(m_i^{(2)}\right)^4
\end{eqnarray*}
Some calculation yields
\begin{eqnarray*}
E\left[\left(m^{(1)}_i\right)^4|\mathcal{F}^{(0)}\right]&=&\frac{1}{K^2}\left[\sum_{k=1}^KE\left(\xi^2_{i,k}-g_{i,k}\right)^4+6\sum^K_{k=1}\sum_{j\ne k}E\left(\xi^2_{i,k}-g_{i,k}\right)^2E\left(\xi^2_{i,j}-g_{i,j}\right)^2\right]\\
&=&6\left[E(\epsilon^4|\mathcal{F}^{(0)})-E(\epsilon^2|\mathcal{F}^{(0)})^2\right]^2+O_p\left(\frac{1}{K}\right)
\end{eqnarray*}
\begin{eqnarray*}
E\left[\left(m^{(1)}_i\right)^3m^{(2)}_i|\mathcal{F}^{(0)}\right]&=&\frac{3}{K^2}E\left[\left(\sum^K_{k=1}\sum_{j\ne k}(\xi^2_{i,k}-g_{i,k})^2(\xi^2_{i,j}-g_{i,j})\right)\cdot\left(\sum^K_{j=1}\xi_{i,j-1}\xi_{i,j}\right)|\mathcal{F}^{(0)}\right]\\
&=&\frac{6}{K^2}\sum^K_{k=2}E\left[(\xi^2_{i,k-1}-g_{i,k-1})^2\xi_{i,k-1}\cdot(\xi^2_{i,k}-g_{i,k})\xi_{i,k}|\mathcal{F}^{(0)}\right]
=O_p\left(\frac{1}{K}\right)
\end{eqnarray*}
\begin{multline*}
E\left[\left(m^{(1)}_i\right)^2\left(m^{(2)}_i\right)^2|\mathcal{F}^{(0)}\right]=\underbrace{\frac{1}{K^2}E\left[\sum^K_{k=1}\left(\xi^2_{i,k}-g_{i,k}\right)^2\cdot\left(\sum^K_{j=1}\xi_{i,j-1}\xi_{i,j}\right)^2|\mathcal{F}^{(0)}\right]}_{\text{(m21)}}\\
+\underbrace{\frac{1}{K^2}E\left[\sum^K_{k=1}\sum_{j\ne k}\left(\xi^2_{i,k}-g_{i,k}\right)\left(\xi^2_{i,j}-g_{i,j}\right)\cdot\left(\sum^K_{j=1}\xi_{i,j-1}\xi_{i,j}\right)^2|\mathcal{F}^{(0)}\right]}_{\text{(m22)}}
\end{multline*}
Note that
\begin{multline*}
\text{(m21)}=\frac{1}{K^2}E\left[\sum^K_{k=1}\sum_{j\ne k,k+1}\left(\xi^2_{i,k}-g_{i,k}\right)^2\cdot\xi^2_{i,j-1}\xi^2_{i,j}|\mathcal{F}^{(0)}\right]+\frac{1}{K^2}E\left[\sum^K_{k=1}\left(\xi^2_{i,k}-g_{i,k}\right)^2\cdot\xi^2_{i,k-1}\xi^2_{i,k}|\mathcal{F}^{(0)}\right]\\+\frac{1}{K^2}E\left[\sum^{K-1}_{k=1}\left(\xi^2_{i,k}-g_{i,k}\right)^2\cdot\xi^2_{i,k}\xi^2_{i,k+1}|\mathcal{F}^{(0)}\right]
=\left[E(\epsilon^4|\mathcal{F}^{(0)})-E(\epsilon^2|\mathcal{F}^{(0)})^2\right]E(\epsilon^2|\mathcal{F}^{(0)})^2+O_p\left(\frac{1}{K}\right)
\end{multline*}
and
\begin{eqnarray*}
\text{(m22)}&=&\frac{1}{K^2}E\left[\left(\sum^K_{k=1}\sum_{j\ne k}\left(\xi^2_{i,k}-g_{i,k}\right)\left(\xi^2_{i,j}-g_{i,j}\right)\right)\cdot\left(\sum^K_{j=1}\xi^2_{i,j-1}\xi^2_{i,j}+\sum^{K-1}_{j=1}\xi_{i,j-1}\xi^2_{i,j}\xi_{i,j+1}\right)|\mathcal{F}^{(0)}\right]\\
&=&\frac{1}{K^2}E\left[\sum^K_{k=2}\left(\xi^4_{i,k-1}-\xi^2_{i,k-1}g_{i,k-1}\right)\left(\xi^4_{i,k}-\xi^2_{i,k}g_{i,k}\right)|\mathcal{F}^{(0)}\right]\\
&&+\frac{1}{K^2}E\left[\sum^{K-1}_{k=2}\left(\xi^3_{i,k-1}-\xi_{i,k-1}g_{i,k-1}\right)\cdot\xi^2_{i,k}\cdot\left(\xi^3_{i,k+1}-\xi_{i,k+1}g_{i,k+1}\right)|\mathcal{F}^{(0)}\right]
=O_p\left(\frac{1}{K}\right)
\end{eqnarray*}
so we have
\begin{equation*}
E\left[\left(m^{(1)}_i\right)^2\left(m^{(2)}_i\right)^2|\mathcal{F}^{(0)}\right]=\left[E(\epsilon^4|\mathcal{F}^{(0)})-E(\epsilon^2|\mathcal{F}^{(0)})^2\right]E(\epsilon^2|\mathcal{F}^{(0)})^2+O_p\left(\frac{1}{K}\right)
\end{equation*}
\begin{multline*}
E\left[m^{(1)}_i\left(m^{(2)}_i\right)^3|\mathcal{F}^{(0)}\right]=\frac{3}{K^2}\sum^K_{k=2}E\left[\xi^2_{i,k-2}\xi^3_{i,k-1}\xi^3_{i,k}|\mathcal{F}^{(0)}\right]+\frac{3}{K^2}\sum^{K-2}_{k=1}E\left[\xi^3_{i,k}\xi^3_{i,k+1}\xi^2_{i,k+2}|\mathcal{F}^{(0)}\right]\\+\frac{1}{K^2}\sum^K_{k=1}E\left[\left(\xi^5_{i,k}-\xi^3_{i,k}g_{i,k}\right)\cdot\xi^3_{i,k-1}|\mathcal{F}^{(0)}\right]+\frac{1}{K^2}\sum^{K-1}_{k=1}E\left[\left(\xi^5_{i,k}-\xi^3_{i,k}g_{i,k}\right)\cdot\xi^3_{i,k+1}|\mathcal{F}^{(0)}\right]
=O_p\left(\frac{1}{K}\right)
\end{multline*}
\begin{multline*}
E\left[\left(m^{(2)}_i\right)^4|\mathcal{F}^{(0)}\right]=\frac{1}{K^2}\sum^K_{k=1}E\left(\xi^4_{i,k-1}\xi^4_{i,k}|\mathcal{F}^{(0)}\right)+\frac{6}{K^2}\sum^K_{k=1}\sum_{j\ne k}E\left(\xi^2_{i,k-1}\xi^2_{i,k}\xi^2_{i,j-1}\xi^2_{i,j}|\mathcal{F}^{(0)}\right)\\
=6E(\epsilon^2|\mathcal{F}^{(0)})^4+O_p\left(\frac{1}{K}\right)
\end{multline*}
Thus, from the above calculation, we have
\begin{eqnarray*}
E(m_i^4|\mathcal{F}^{(0)})&=&E\left[\left(m_i^{(1)}\right)^4|\mathcal{F}^{(0)}\right]+6E\left[\left(m_i^{(1)}\right)^2\left(m_i^{(2)}\right)^2|\mathcal{F}^{(0)}\right]+E\left[\left(m_i^{(2)}\right)^4|\mathcal{F}^{(0)}\right]+O_p\left(\frac{1}{K}\right)\\
&=&6\left[E(\epsilon^4|\mathcal{F}^{(0)})^2-E(\epsilon^4|\mathcal{F}^{(0)})E(\epsilon^2|\mathcal{F}^{(0)})^2+E(\epsilon^2|\mathcal{F}^{(0)})^4\right]+O_p\left(\frac{1}{K}\right)
\end{eqnarray*}
\end{proof}

\subsection{Proof of Theorem \ref{thm3}}\label{prfthm3}
\begin{proof}
Under the assumption of this theorem, we know $g_t(\omega^{(0)})=E(\epsilon^2|\mathcal{F}^{(0)})(\omega^{(0)})$ has a constant value. By the proof of \textbf{Theorem \ref{thm2}}, we know $m_i\overset{{\mathcal{L}-s}}{\longrightarrow}\mathcal{MN}\left(0,E(\epsilon^4|\mathcal{F}^{(0)})\right)$ under the null hypothesis. By continuous mapping theorem, $m^2_i\overset{{\mathcal{L}-s}}{\longrightarrow}E(\epsilon^4|\mathcal{F}^{(0)})\cdot\chi^2_1$ where $\chi^2_1$ denotes a centered Chi-square distribution with degree of freedom 1 and independent of $\mathcal{F}^{(1)}$. Note that
\begin{equation*}
U(Y,K_n,s_n,2)^n_T
=\frac{1}{r_n-s_n+1}\left[\sum_{i=1}^{r_n-s_n+1}m^2_i+\sum_{i=1}^{r_n-s_n+1}m^2_{i+s_n-1}+2\sum_{i=1}^{r_n-s_n+1}m_im_{i+s_n-1}\right]+O_p\left(\frac{1}{K_n}\bigvee\frac{K_n}{n}\right)
\end{equation*}
we can write
\begin{equation}\label{plim}
\sqrt{r_n-s_n+1}\left(U(Y,K_n,2)^n_T-2E(\epsilon^4|\mathcal{F}^{(0)})\right)=2H_T^{(1)}+2H_T^{(2)}+R_T+O_p\left(\frac{\sqrt{r_n}}{K_n}\bigvee\frac{1}{\sqrt{r_n}}\right)
\end{equation}
where
\begin{equation*}
\begin{array}{ll}
    H_T^{(1)}&=\frac{1}{\sqrt{r_n-s_n+1}}\sum_{i=1}^{r_n-s_n+1}(m_i-E(m_i^2|\mathcal{F}^{(0)}))\\
    H_T^{(2)}&=\frac{1}{\sqrt{r_n-s_n+1}}\sum_{i=1}^{r_n-s_n+1}m_im_{i-s_n+1}\\
    R_T&=\frac{1}{\sqrt{r_n-s_n+1}}\left[\sum_{i=r_n-s_n+2}^{r_n}(m^2_i-E(\epsilon^4|\mathcal{F}^{(0)}))-\sum_{i=1}^{s_n-1}(m_i^2-E(\epsilon^4|\mathcal{F}^{(0)}))\right]
\end{array}
\end{equation*}
Furthermore, note that on the coarser filtered probability space $\left(\Omega^{(1)},\mathcal{F}^{(1)},\{\mathcal{F}^{(1)}_{t_{(i-1)K_n}}\}_{i\in\mathbb{N}},\mathbb{P}^{(1)}\right)$, $H^{(1)}_t$ and $H^{(2)}_t$ are two discrete martingales, and the increments of $H^{(1)}_t$ and $H^{(2)}_t$, namely
\begin{equation*}
\begin{array}{cc}
    \left\{\frac{1}{\sqrt{r_n-s_n+1}}\left(m^2_i-E(m^2_i|\omega^{(0)}\right)\right\}_{n\in\mathbb{N}^{+},i\in\mathbb{N}^{+}\le r_n-s_n+1} &
    \left\{\frac{1}{\sqrt{r_n-s_n+1}}\left(m_im_{i+s_n-1}\right)\right\}_{n\in\mathbb{N}^{+},i\in\mathbb{N}^{+}\le r_n-s_n+1}
\end{array}
\end{equation*}
are two triangular sequences to which we can apply martingale central limit theorem. By the results of \textbf{Lemma \ref{lemma4}}
\begin{eqnarray*}
\langle H^{(1)},H^{(1)}\rangle_T|\mathcal{F}^{(0)}
&=&\frac{1}{r_n-s_n+1}\sum^{r_n-s_n+1}_{i=1}\left[E\left(m^4_i|\mathcal{F}^{(1)}_{t_{(i-1)K}}\right)-E\left(m^2_i|\mathcal{F}^{(0)}\cup\mathcal{F}^{(1)}_{t_{(i-1)K}}\right)^2\right]|\mathcal{F}^{(0)}\\
&=&5E(\epsilon^4|\mathcal{F}^{(0)})^2-6E(\epsilon^4|\mathcal{F}^{(0)})E(\epsilon^2|\mathcal{F}^{(0)})^2+6E(\epsilon^2|\mathcal{F}^{(0)})^4
\end{eqnarray*}
and
\begin{eqnarray*}
\langle H^{(2)},H^{(2)}\rangle_T|\mathcal{F}^{(0)}
&=&\frac{1}{r_n-s_n+1}\sum^{r_n-s_n+1}_{i=1}E(m^2_i|\mathcal{F}^{(0)})\cdot E\left(m^2_{i+s_n-1}|\mathcal{F}^{(0)}\right)\\
&&+\frac{1}{r_n-s_n+1}\sum^{r_n-s_n+1}_{i=1}\left(m^2_i-E(m^2_i|\mathcal{F}^{(0)})\right)\cdot E\left(m^2_{i+s_n-1}|\mathcal{F}^{(0)}\right)
\end{eqnarray*}
since  $P(\tau_l>T)\overset{\mathbb{P}^{(0)}}{\longrightarrow} 1$ as $l\longrightarrow \infty$ and
\begin{multline*}
E\left(\left(\frac{1}{r_n-s_n+1}\sum^{r_n-s_n+1}_{i=1}\left(m^2_i-E(m^2_i|\mathcal{F}^{(0)})\right)\cdot E\left(m^2_{i+s_n-1}|\mathcal{F}^{(0)}\right)\mathbf{1}_{\{\tau_l>T\}}\right)^2|\mathcal{F}^{(0)}\right)\\
=\frac{1}{(r_n-s_n+1)^2}\sum^{r_n-s_n+1}_{k=1}Var\left(m^2_i-E(m^2_i|\mathcal{F}^{(0)})\right)\left(E\left(m^2_{i+s_n-1}|\mathcal{F}^{(0)}\right)\right)^2\mathbf{1}_{\{\tau_l>T\}}\\
\le\frac{1}{(r_n-s_n+1)^2}\sum^{r_n-s_n+1}_{k=1}M_{(4,l)}\cdot M^2_{(2,l)}=O_p\left(\frac{1}{r_n-s_n+1}\right)
\end{multline*}
by \textbf{Proposition \ref{prop}}, we know $\frac{1}{r_n-s_n+1}\sum^{r_n-s_n+1}_{i=1}\left(m^2_i-E(m^2_i|\mathcal{F}^{(0)})\right)\cdot E\left(m^2_{i+s_n-1}|\mathcal{F}^{(0)}\right)\overset{\mathbb{P}}{\longrightarrow}0
$, thus we have
\begin{equation*}
\langle H^{(2)},H^{(2)}\rangle_T|\mathcal{F}^{(0)}\overset{\mathbb{P}}{\longrightarrow} E(\epsilon^4|\mathcal{F}^{(0)})^2
\end{equation*}
Besides,
\begin{equation*}
\langle H^{(1)},H^{(2)}\rangle_T|\mathcal{F}^{(0)}
=\frac{1}{r_n-s_n+1}\sum^{r_n-s_n+1}_{i=1}(m_i^3-m_iE(m^2_i|\mathcal{F}^{(0)}))\cdot E\left(m_{i+s_n-1}|\mathcal{F}^{(0)}\right)=0
\end{equation*}
Therefore, we have the following joint asymptotic distribution for $H^{(1)}_T$ and $H^{(2)}_T$:
\begin{equation}
\left(\begin{array}{c}H^{(1)}_T\\H^{(2)}_T\end{array}\right)\overset{\mathcal{L}-s}{\longrightarrow}\mathcal{MN}\left(\left(\begin{array}{c}0\\0\end{array}\right),\left(\begin{array}{cc}\zeta^2&0\\0&E(\epsilon^4|\mathcal{F}^{(0)})^2\end{array}\right)\right)
\end{equation}
where $\zeta^2=5E(\epsilon^4|\mathcal{F}^{(0)})^2-6E(\epsilon^4|\mathcal{F}^{(0)})E(\epsilon^2|\mathcal{F}^{(0)})^2+6E(\epsilon^2|\mathcal{F}^{(0)})^4$. Lastly, note that $R_T=o_p(1)$, this is because $P(\tau_l>T)\overset{\mathbb{P}^{(0)}}{\longrightarrow} 1$ as $l\to\infty$, and
\begin{eqnarray*}
E(R^2_T\mathbf{1}_{\{\tau_l>T\}}|\mathcal{F}^{(0)})
&=&\frac{1}{r_n-s_n+1}\left[\sum_{i=1}^{s_n-1}E(m^4_i|\mathcal{F}^{(0)})+\sum_{i=r_n-s_n+2}^{r_n}E(m^4_i|\mathcal{F}^{(0)})-2(s_n-1)E(\epsilon^4|\mathcal{F}^{(0)})^2\right]\\
&=&O_p\left(\frac{s_n}{r_n-s_n+1}\right)=o_p(1)
\end{eqnarray*}

Plug in these results into (\ref{plim}), we can get
\begin{equation*}
\sqrt{r_n-s_n+1}\left(U(Y,K_n,s_n,2)^n_T-2E(\epsilon^4|\mathcal{F}^{(0)})\right)
=2\left(H_T^{(1)}+H_T^{(2)}\right)+o_p(1)\overset{\mathcal{L}-s}{\longrightarrow}\mathcal{MN}\left(0,\eta^2\right)
\end{equation*}
where $\eta^2=24 [E(\epsilon^4|\mathcal{F}^{(0)})^2-E(\epsilon^4|\mathcal{F}^{(0)})E(\epsilon^2|\mathcal{F}^{(0)})^2+E(\epsilon^2|\mathcal{F}^{(0)})^4]$.

According to \textbf{Remark \ref{remk.est.epsilon4}} about the consistent estimator of $E(\epsilon^4|\mathcal{F}^{(0)})$, $\widehat{\eta}^2-\eta^2=O_p(1/\sqrt{n})$ when the noise is stationary due to (\ref{eps4}), as well as $\frac{1}{2n}[Y,Y]_{\mathcal{G}}-E(\epsilon^2|\mathcal{F}^{(0)})=O_p(1/\sqrt{n})$, plus the stable convergence for $U(Y,K_n,s_n,2)^n_{\mathcal{T}}$, (\ref{asym.V}) follows. 


\end{proof}

\subsection{Proof of Theorem \ref{thm5.0} and \ref{thm5}}\label{prfthm56}
\begin{proof}
In this proof, we write $K$ and $r$ without the subscript $n$ in order to avoid clustered notation.
We give the proof for \textbf{Theorem \ref{thm5}} first, and dictate how to modify the proof to prove \textbf{Theorem \ref{thm5.0}}.

\subsubsection{The law of large number: the limit quantity}
Under the assumption of \textbf{Theorem \ref{thm5}}, and from \textbf{Lemma \ref{lem_RV}}, we have
\begin{eqnarray*}
\frac{1}{2K}[Y,Y]_{\mathcal{S}_i}&=&\frac{1}{2K}[\epsilon,\epsilon]_{\mathcal{S}_i}+O_p\left(\frac{1}{K}\right)\\
&=&\frac{1}{K}\sum_{t_j\in(T_{i-1},T_i]}g_{t_j}+\frac{1}{\sqrt{K}}\left(m_i^{(1)}-m_i^{(2)}\right)+O_p\left(\frac{1}{K}\right)
\end{eqnarray*}
where $m^{(1)}_i$ and $m^{(2)}_i$ are defined in (\ref{m_i}) 
which are asymptotically mixing normal. Thus,
\begin{eqnarray*}
\frac{1}{2K}[Y,Y]_{\mathcal{S}_{i+1}}-\frac{1}{2K}[Y,Y]_{\mathcal{S}_i}
&=&\underbrace{\frac{1}{K}\sum^K_{j=1}\sum_{l=1}^K\left(g_{t_{(i-1)K+j+l}}-g_{t_{(i-1)K+j+l-1}}\right)}_{\text{(A)}}+\frac{1}{\sqrt{K}}(m_{i+1}-m_i)+O_p\left(\frac{1}{K}\right)
\end{eqnarray*}
notice that:
\begin{eqnarray*}
\text{(A)}^2&=&\frac{1}{K^2}\left[\sum^K_{j=1}(j-1)\Delta g_{t_{(i-1)K+j}}+\sum^{2K}_{j=K+1}(K-(j-1))\Delta g_{t_{(i-1)K+j}}\right]^2\\
&=&\sum^K_{j=1}\frac{(j-1)^2}{K^2}(\Delta g_{t_{(i-1)K+j}})^2+\sum^{2K}_{j=K+1}\frac{(K-(j-1))^2}{K^2}(\Delta g_{t_{(i-1)K+j}})^2+\text{(\Rmnum{1})}+\text{(\Rmnum{2})}+\text{(\Rmnum{3})}
\end{eqnarray*}
where
\begin{equation}\label{discretizationerror}
\begin{array}{ll}
      \text{(\Rmnum{1})}&=\sum^K_{j=1}\sum_{l\ne j}\frac{(j-1)(l-1)}{K^2}\Delta g_{t_{(i-1)K+j}}\Delta g_{t_{(i-1)K+l}}\\
      \text{(\Rmnum{2})}&=\sum^K_{j=1}\sum_{l\ne j}\frac{(K-(j-1))(K-(l-1))}{K^2}\Delta g_{t_{iK+j}}\Delta g_{t_{iK+l}}\\
      \text{(\Rmnum{3})}&=\sum^K_{j=1}\sum^K_{l=1}\frac{(j-1)(K-(l-1))}{K^2}\Delta g_{t_{(i-1)K+j}}\Delta g_{t_{iK+l}}
\end{array}
\end{equation}
are mean-0 martingales. By standard localization procedure, we can strengthen the condition by assuming $\sigma^{(g)}_t\le \sigma^{(g)}_+$, $\forall t \in [0,\mathcal{T}]$, therefore,
\begin{equation*}
E[\text{(\Rmnum{1})}^2]\le \frac{T^2(\sigma^{(g)}_+)^4}{n^2}\sum^K_{j=1}\sum^K_{l=1}\left[\frac{(j-1)(l-1)}{K^2}\right]^2=\frac{T^2(\sigma^{(g)}_+)^4}{n^2}\cdot \sum^K_{j=1}\frac{(j-1)^2}{K^2}\cdot \sum^K_{j=1}\frac{(l-1)^2}{K^2}=O_p\left(\frac{K^2}{n^2}\right)
\end{equation*}
by Chebyshev inequality, $\text{(\Rmnum{1})}=O_p\left(\frac{K}{n}\right)$. Similarly, $\text{(\Rmnum{2})}$, $\text{(\Rmnum{3})}=O_p\left(\frac{K}{n}\right)$. Furthermore, we can know $\text{(A)}=O_p\left(\sqrt{\frac{K}{n}}\right)$. Thus,
\begin{multline*}
	\sum^{r-1}_{i=1}\left(\frac{1}{2K}[Y,Y]_{\mathcal{S}_{i+1}}-\frac{1}{2K}[Y,Y]_{\mathcal{S}_i}\right)^2=\sum^K_{j=1}\frac{(j-1)^2}{K^2}(\Delta g_{t_j})^2+\sum^K_{j=1}\frac{(K-(j-1))^2}{K^2}(\Delta g_{t_{(r-1)K+j}})^2\\
	+\sum^{r-2}_{i=2}\sum^K_{j=1}\frac{(j-1)^2+(K-(j-1))^2}{K^2}(\Delta g_{t_{(i-1)K+j}})^2+\underbrace{\sum^{r-1}_{i=1}[\text{(\Rmnum{1})}+\text{(\Rmnum{2})}+\text{(\Rmnum{3})}]}_{O_p\left(\frac{\sqrt{r}K}{n}\right)=O_p\left(\frac{1}{\sqrt{r}}\right)}+\underbrace{O_p\left(\frac{r}{K}\right)}_{\text{error due to noises}}
\end{multline*}
note that the error due to noises (of the stochastic order $O_p\left(\frac{r}{K}\right)$) approximately equals to $\frac{2r}{\mathcal{T}K}\int_0^\mathcal{T}h_t\,\mathrm{d}t$ by the proof of \textbf{Theorem \ref{thm3}}. Hence,
\begin{multline*}
	\sum^{r-1}_{i=1}\left(\frac{1}{2K}[Y,Y]_{\mathcal{S}_{i+1}}-\frac{1}{2K}[Y,Y]_{\mathcal{S}_i}\right)^2-\frac{2}{3}\sum^n_{j=1}(\Delta g_{t_j})^2-\frac{2r}{\mathcal{T}K}\int_0^\mathcal{T}h_t\,\mathrm{d}t\\
	=\sum_{i=1}^{r-2}\sum_{j=1}^K\left[\frac{1}{3}-2\frac{K-(j-1)}{K}\frac{j-1}{K}\right](\Delta g_{t_{iK+j}})^2
	+\sum^K_{j=1}\left[\frac{(j-1)^2}{K^2}-\frac{2}{3}\right](\Delta g_{t_j})^2\\+\sum^K_{j=1}\left[\frac{(K-(j-1))^2}{K^2}-\frac{2}{3}\right](\Delta g_{t_{(r-1)K+j}})^2+\text{(E1)}+\text{(E2)}
\end{multline*}
where
\begin{equation*}
\begin{array}{ll}
    \text{(E1)}=\sum^{r-1}_{i=1}[\text{(\Rmnum{1})}+\text{(\Rmnum{2})}+\text{(\Rmnum{3})}]=O_p\left(\frac{1}{\sqrt{r}}\right) & \text{ end points of martingale in } \{g_t\}_{t\in[0,\mathcal{T}]}\\
    \text{(E2)}=\frac{2r}{K}\left(\frac{1}{r}\sum_{i=1}^{r-1}m_i^2-\frac{1}{\mathcal{T}}\int_0^\mathcal{T}h_t\,\mathrm{d}t\right)=O_p\left(\frac{\sqrt{r}}{K}\right) & \text{ error due to noise}
\end{array}
\end{equation*}
The order of (E1) will be analyzed later; $\text{(E2)}$ comes from negligible remaining of microstructure noise and its order is obtained from the proof of \textbf{Theorem \ref{thm3}}:
\begin{equation*}
\frac{K}{r}\cdot\sqrt{r}\text{(E2)}\overset{\mathcal{L}-s}{\longrightarrow}\mathcal{MN}\left(0,\frac{24}{\mathcal{T}}\int_0^\mathcal{T}\left[h_t^2-h_tg^2_t+g^4_t\right]\mathrm{d}t\right)
\end{equation*}
Moreover,
\begin{multline*}
\sum_{i=1}^{r-2}\sum_{j=1}^K\left[\frac{1}{3}-2\frac{K-(j-1)}{K}\frac{j-1}{K}\right](\Delta g_{t_{iK+j}})^2+\sum^K_{j=1}\left[\frac{(j-1)^2}{K^2}-\frac{2}{3}\right](\Delta g_{t_j})^2\\
+\sum^K_{j=1}\left[\frac{(K-(j-1))^2}{K^2}-\frac{2}{3}\right](\Delta g_{t_{(r-1)K+j}})^2=O_p\left(\frac{K}{n}\right)=O_p\left(\frac{1}{r}\right)
\end{multline*}
Because $\sum^n_{j=1}(\Delta g_{t_j})^2-\langle g,g\rangle_T=O_p\left(\frac{1}{\sqrt{n}}\right)$\footnote{The proofs can be found in \cite{jp98,mz06}.}, so
\begin{equation*}\label{LLN_N}
\frac{r}{K}U(Y,K,2)_\mathcal{T}^2-\frac{2}{3}\langle g,g\rangle_\mathcal{T}-\frac{2r}{\mathcal{T}K}\int^\mathcal{T}_0h_t\,\mathrm{d}t=\text{(E1)}+\text{(E2)}+O_p\left(\frac{1}{r}\right)
\end{equation*}

\subsubsection{Decomposition of the discretization error process}\label{prfthm5.2}
Followed from (\ref{discretizationerror}),
if we define the following two quantities:
\begin{eqnarray}\label{def_N}
N^{(1)}_\mathcal{T}&\equiv&2\sqrt{r}\sum_{i=0}^{r-1}\sum^K_{j=2}\Delta g_{t_{iK+j}}\cdot\left[\sum_{l=1}^{j-1}\left(1+2\frac{j-1}{K}\frac{l-1}{K}-\frac{j-1}{K}-\frac{l-1}{K}\right)\Delta g_{t_{iK+l}}\right]\\
N^{(2)}_\mathcal{T}&\equiv&\sqrt{r}\sum_{i=1}^{r-1}\sum^K_{j=1}\left(1-\frac{j-1}{K}\right)\Delta g_{t_{iK+j}}\cdot\left(\sum_{l=1}^K\frac{l-1}{K}\Delta g_{t_{(i-1)K+l}}\right)
\end{eqnarray}
then we have
\begin{equation}\label{rep_dis_error}
\text{(E1)}=\frac{1}{\sqrt{r}}N^{(1)}_\mathcal{T}+\frac{1}{\sqrt{r}}N^{(2)}_\mathcal{T}+\underbrace{O_p\left(\frac{K}{n}\right)}_{\text{the edge in (E1)}}
\end{equation}
furthermore, by (\ref{def_N})
\begin{multline*}
\langle N^{(1)},N^{(2)}\rangle_\mathcal{T}=2r\sum^{r-1}_{i=1}\sum^K_{j=2}\left(1-\frac{j-1}{K}\right)\Delta\langle g,g\rangle_{t_{iK+j}}\times\sum^K_{l=1}\frac{l-1}{K}\Delta g_{t_{(i-1)K+l}}\\
\times\sum^{j-1}_{l=1}\left(1+2\frac{j-1}{K}\frac{l-1}{K}-\frac{j-1}{K}-\frac{l-1}{K}\right)\Delta g_{t_{iK+l}}
\end{multline*}
\begin{multline*}
E\left[\langle N^{(1)},N^{(2)}\rangle_\mathcal{T}^2\right]=4r^2\sum^{r-1}_{i=1}\sum^K_{j=2}\left(1-\frac{j-1}{K}\right)^2\left(\Delta\langle g,g\rangle_{t_{iK+j}}\right)^2\times \left(\sum^K_{l=1}\frac{(l-1)^2}{K^2}E\left[(\Delta g_{t_{(i-1)K+l}})^2\right]\right)\\
\times\left(\sum^{j-1}_{l=1}\left(\frac{(j-1)(l-1)}{K^2}+\frac{(K-(j-1))(K-(l-1))}{K^2}\right)^2E\left[(\Delta g_{t_{iK+l}})^2\right]\right)
\end{multline*}
thus we can know $E\left[\langle N^{(1)},N^{(2)}\rangle_\mathcal{T}^2\right]=O_p\left(\frac{r^3K^3}{n^4}\right)=O_p\left(\frac{1}{n}\right)$, So
\begin{equation}\label{decomp_N}
\langle N^{(1)}+N^{(2)},N^{(1)}+N^{(2)}\rangle_\mathcal{T}=\langle N^{(1)},N^{(1)}\rangle_\mathcal{T}+\langle N^{(2)},N^{(2)}\rangle_\mathcal{T}+O_p\left(\frac{1}{\sqrt{n}}\right)
\end{equation}

\subsubsection{Calculating $\langle N^{(1)},N^{(1)}\rangle_\mathcal{T}$}\label{prfthm5.3}
By (\ref{def_N}),
\begin{eqnarray*}
\langle N^{(1)},N^{(1)}\rangle_\mathcal{T}&=&4r\sum^{r-1}_{i=0}\sum^K_{j=2}\Delta \langle g,g\rangle_{t_{iK+j}}\times\left[\sum^{j-1}_{l=1}\left(\frac{(j-1)(l-1)}{K^2}+\frac{(K-(j-1))(K-(l-1))}{K^2}\right)\Delta g_{t_{iK+l}}\right]^2\\
&=&\text{(A1)}+\text{(A2)}
\end{eqnarray*}
where
\begin{eqnarray*}
\text{(A1)}&=&4r\sum^{r-1}_{i=0}\sum^K_{j=2}\Delta \langle g,g\rangle_{t_{iK+j}}\times\left[\sum^{j-1}_{l=1}\left(\frac{2(j-1)-K}{K^2}(l-1)+\frac{K-(j-1)}{K}\right)^2(\Delta g_{t_{iK+l}})^2\right]\\
&=&4r\sum^{r-1}_{i=0}\sum^K_{j=2}\left(\sigma^{(g)}_{t_{iK+j}}\right)^4\Delta^2_n\times\sum^{j-1}_{l=1}\left[\frac{2(j-1)-K}{K^2}(l-1)+\frac{K-(j-1)}{K}\right]^2+O_p\left(\frac{1}{rK^{1/2}}\right)\\
\end{eqnarray*}
the error term appears because $\sigma^{(g)}$ is an It\^o process, and the error due to the local-consistency approximation for $\sigma^{(g)}$ is of a smaller order than $4r\sum^{r-1}_{i=0}\sum^K_{j=2}\sqrt{j}\Delta_n^3\sum^{j-1}_{l=1}\left[\frac{2(j-1)-K}{K^2}(l-1)+\frac{K-(j-1)}{K}\right]^2=O_p\left(\frac{1}{rK^{1/2}}\right)$, besides
\begin{equation*}
\text{(A2)}=8r\sum^{r-1}_{i=0}\sum^K_{j=3}\Delta \langle g,g\rangle_{t_{iK+j}}\cdot \phi_j
\end{equation*}
where
\begin{equation*}
\phi_j=\sum^{j-1}_{l=2}\sum^{l-1}_{k=1}\left[\frac{2(j-1)-K}{K^2}(l-1)+\frac{K-(j-1)}{K}\right]\cdot\left[\frac{2(j-1)-K}{K^2}(k-1)+\frac{K-(j-1)}{K}\right]\Delta g_{t_{iK+l}}\Delta g_{t_{iK+k}}
\end{equation*}
by Burkholder-Davis-Gundy inequality, $\exists C_1\in \mathbb{R}^{+}$ such that 
\begin{multline*}
	\|\phi_j\|^2_2\le C_1^2\|\langle \phi_j,\phi_j\rangle\|_1\\
\le\frac{\mathcal{T}^2C_1^2\left(\sigma^{(g)}_{+}\right)^2}{n^2}\sum^{j-1}_{l=2}\left[\frac{2(j-1)-K}{K^2}(l-1)+\frac{K-(j-1)}{K}\right]^2\times \sum^{l-1}_{k=1}\left[\frac{2(j-1)-K}{K^2}(k-1)+\frac{K-(j-1)}{K}\right]^2
\end{multline*}
thus, $\|\phi_j\|_2=O_p\left(\frac{j^3}{K^2n}\right)$ and $\sum^K_{j=3}\|\phi_j\|_2^2\le \sum^K_{j=3}C^2_1\|\langle\phi_j,\phi_j\rangle\|_1=O_p\left(\frac{K^3}{n^2}\right)$. Define $\text{(A2)}^\prime\equiv8r\sum^{r-1}_{i=0}\sum^K_{j=3}\left(\sigma^{(g)}_{t_{iK}}\right)^2\Delta_n\phi_j$, and apply Burkholder-Davis-Gundy inequality again, but on $\text{(A2)}'$, we get $\|\text{(A2)}'\|^2_2\le 64r^2C^2_2\sum^{r-1}_{i=0}\sum^K_{j=3}\left(\sigma^{(g)}_+\right)^4\Delta_n^2\times \|\langle\phi_j,\phi_j\rangle\|_1=O_p\left(\frac{r^3K^3}{n^4}\right)=O_p\left(\frac{1}{n}\right)$, so
\begin{equation}\label{A2.1}
\text{(A2)}^\prime=O_p\left(\frac{1}{\sqrt{n}}\right)
\end{equation}
by Cauchy-Schwarz inequality,
\begin{multline}\label{A2.2}
\|\text{(A2)}-\text{(A2)}^\prime\|_1
\le8r^2\Delta(\mathcal{G})\sum^K_{j=3}\left\|\sup_{|t-s|\le K\Delta(\mathcal{G})}\left[\left(\sigma^{(g)}_t\right)^2-\left(\sigma^{(g)}_s\right)^2\right]\right\|_2 \cdot\|\phi_j\|_2\\
\le8r^2K\left(\Delta(\mathcal{G})\right)^{\frac{3}{2}}\left(\sum^K_{j=3}\|\phi_j\|^2_2\right)^{\frac{1}{2}}=O_p\left(\frac{1}{\sqrt{r}}\right)
\end{multline}
from (\ref{A2.1}) and (\ref{A2.2}), we can know $\text{(A2)}=o_p(1)$, and more importantly,
\begin{equation*}
\langle N^{(1)},N^{(1)}\rangle_T=4r\sum^{r-1}_{i=0}\sum^K_{j=2}\left(\sigma^{(g)}_{t_{iK+j}}\right)^4\Delta^2_n\times\underbrace{\sum^{j-1}_{l=1}\left[\frac{2(j-1)-K}{K^2}(l-1)+\frac{K-(j-1)}{K}\right]^2}_{(1)}+o_p(1)
\end{equation*}
notice that $(1)=\frac{4}{3}\frac{j^5}{K^4}-\frac{10}{3}\frac{j^4}{K^3}+\frac{13}{3}\frac{j^3}{K^2}-3\frac{j^2}{K}+j+O(1)$, 
so
\begin{eqnarray*}
\langle N^{(1)},N^{(1)}\rangle_T
&=&4r\sum^{r-1}_{i=0}\sum^K_{j=2}(1)\cdot\left[\left(\sigma^{(g)}_{t_{iK}}\right)^2+O_p\left(\sqrt{K\Delta_n}\right)\right]^2\Delta_n^2+o_p(1)\\
&=&4r\sum^{r-1}_{i=0}\left[\left(\sum^K_{j=2}(1)\right)\times\frac{\Delta_n}{K}\left(\sigma^{(g)}_{t_{iK}}\right)^4K\Delta_n\right]+o_p(1)
\end{eqnarray*}
by Faulhaber's formula, wee know
\begin{equation*}
\sum^K_{j=2}(1)=\left(\frac{4}{3}\times\frac{1}{6}-\frac{10}{3}\times\frac{1}{5}+\frac{13}{3}\times\frac{1}{4}-3\times\frac{1}{3}+\frac{1}{2}\right)K^2+O(K)=\frac{5}{36}K^2+O(K)
\end{equation*}
so
\begin{equation*}
\langle N^{(1)},N^{(1)}\rangle_T=T\sum^{r-1}_{i=1}\left[\frac{5}{9}+O_p\left(\frac{1}{K}\right)\right]\left(\sigma^{(g)}_{t_{iK}}\right)^4K\Delta_n+o_p(1)\longrightarrow \frac{5T}{9}\int_0^T\left(\sigma^{(g)}_t\right)^4\mathrm{d}t
\end{equation*}

\subsubsection{Calculating $\langle N^{(2)},N^{(2)}\rangle_T$}\label{prfthm5.4}
By (\ref{def_N}),
\begin{eqnarray*}
\langle N^{(2)},N^{(2)}\rangle_T&=&r\sum^{r-1}_{i=1}\sum^K_{j=1}\frac{(K-(j-1))^2}{K^2}\Delta\langle g,g\rangle_{t_{iK+j}}\times\left(\sum^K_{l=1}\frac{(l-1)}{K}\Delta g_{t_{(i-1)K+l}}\right)^2\\
&=&\text{(B1)}+\text{(B2)}
\end{eqnarray*}
where
\begin{eqnarray*}
\text{(B1)}&=&r\sum^{r-1}_{i=1}\sum^K_{j=1}\frac{(K-(j-1))^2}{K^2}\Delta\langle g,g\rangle_{t_{iK+j}}\times \sum^K_{l=1}\frac{(l-1)^2}{K^2}(\Delta g_{t_{(i-1)K+l}})^2\\
&=&r\sum^{r-1}_{i=1}\sum^K_{j=1}\frac{(K-(j-1))^2}{K^2}\left(\sigma^{(g)}_{t_{iK+j}}\right)^4\Delta_n^2\times \sum^K_{l=1}\frac{(l-1)^2}{K^2}+O_p\left(\frac{1}{rK^{1/2}}\right)
\end{eqnarray*}
the error term just above comes from the local-constancy approximation on $(\sigma^{(g)})^2$, it is of the stochastic order of $O_p(r\sum^{r-1}_{i=1}\sum^K_{j=1}\frac{(K-(j-1))^2}{K^2}\sqrt{K}\Delta_n^3\times\sum^K_{l=1}\frac{(l-1)^2}{K^2})=O_p\left(\frac{1}{rK^{1/2}}\right)$. Besides,
\begin{equation*}
\text{(B2)}=2r\sum^{r-1}_{i=1}\sum^K_{j=1}\frac{(K-(j-1))^2}{K^2}\Delta \langle g,g\rangle_{t_{iK+j}}\psi_i
\end{equation*}
where
\begin{equation*}
\psi_i=\sum^K_{l=2}\sum^{l-1}_{k=1}\frac{l-1}{K}\frac{k-1}{K}\Delta g_{t_{(i-1)K+l}}\Delta g_{t_{(i-1)K+k}}
\end{equation*}
Apply Burkholder-Davis-Gundy on $\psi_i$, since $(\psi_i)_t\equiv\sum^K_{l=2}\sum^{l-1}_{k=1}\frac{l-1}{K}\frac{k-1}{K}\Delta g_{t_{(i-1)K+k}}\int^{t_{(i-1)K+l}\wedge t}_{t_{(i-1)K+l-1}}\mathrm{d}g_t$ is a continuous martingale by assumption of the \textbf{Theorem \ref{thm5}},
\begin{eqnarray*}
\|\psi_i\|^2_2&\le&D_1^2\|\langle \psi_i,\psi_i\rangle\|_1=D^2_1E\sum^K_{l=2}\frac{(l-1)^2}{K^2}\Delta \langle g,g\rangle_{t_{(i-1)K+l}}\times\left(\sum^{l-1}_{k=1}\frac{k-1}{K}\Delta g_{t_{(i-1)K+k}}\right)^2\\
&\le&D^2_1\left(\sigma^{(g)}_+\right)^4(\Delta(\mathcal{G}))^2\times \sum^K_{l=2}\frac{(l-1)^2)}{K^2}\sum^{l-1}_{k=1}\frac{(k-1)^2}{K^2}=O_p\left(\frac{K^2}{n^2}\right)
\end{eqnarray*}
so $\|\psi_i\|^2_2\le D_1\|\langle \psi_i,\psi_i\rangle\|_1=O_p\left(\frac{1}{r^2}\right)$. Define $\text{(B2)}'\equiv 2r\sum^{r-1}_{i=1}\sum^K_{j=1}\frac{(K-(j-1))^2}{K^2}\left(\sigma^{(g)}_{t_{(i-1)K}}\right)^2\Delta_n\psi_i$, apply Burkholder-Davis-Gundy inequality again on $\text{(B2)}'$,
\begin{eqnarray*}
\|\text{(B2)}^\prime\|_2^2&\le&4r^2D_2^2\sum^{r-1}_{i=1}\sum^K_{j=1}\frac{(K-(j-1))^4}{K^4}\left(\sigma^{(g)}_t\right)^4\Delta_n^2\times \|\langle \psi_i,\psi_i\rangle\|_1\\
&=&O_p\left(\frac{r^2}{n^2}\right)\times \sum^{r-1}_{i=1}\sum^K_{j=1}\frac{(K-(j-1))^4}{K^4}\times O_p\left(\frac{1}{r^2}\right)=O_p\left(\frac{1}{n}\right)
\end{eqnarray*}
therefore,
\begin{equation}\label{B2.1}
\text{(B2)}^\prime=O_p\left(\frac{1}{\sqrt{n}}\right)
\end{equation}
By Cauchy-Schwarz inequality,
\begin{multline}\label{B2.2}
\|\text{(B2)}-\text{(B2)}^\prime\|_1
\le2r\sum^{r-1}_{i=1}\sum^K_{j=1}\frac{(K-(j-1))^2}{K^2}\left\|\Delta\langle g,g\rangle_{t_{iK+j}}-\left(\sigma^{(g)}_{t_{(i-1)K}}\right)^2\Delta_n\right\|_2\cdot\left\|\psi_i\right\|_2\\
\le2r^2K\Delta(\mathcal{G})\cdot\left\|\sup_{|t-s|\le 2K\Delta(\mathcal{G})}\left[\left(\sigma^{(g)}_t\right)^2-\left(\sigma^{(g)}_s\right)^2\right]\right\|_2\cdot\sup_i\left\|\psi_i\right\|_2=O_p\left(\frac{1}{\sqrt{r}}\right)
\end{multline}
combine (\ref{B2.1}) and (\ref{B2.2}), we can get $\text{(B2)}=o_p(1)$. More importantly,
\begin{eqnarray*}
\langle N^{(2)},N^{(2)}\rangle_T
&=&r\sum^{r-1}_{i=1}\sum^K_{j=1}\frac{(K-(j-1))^2}{K^2}\left[\left(\sigma^{(g)}_{t_{iK}}\right)^4+O_p(\sqrt{K\Delta_n})\right]\times\left(\frac{K}{3}+O_p(1)\right)\times\Delta_n^2+o_p(1)\\
&=&r\sum^{r-1}_{i=1}\frac{K^2}{9}\left(\sigma^{(g)}_{t_{iK}}\right)^4\Delta_n^2+o_p(1)
\longrightarrow\frac{T}{9}\int_0^T\left(\sigma^{(g)}_t\right)^4\mathrm{d}t
\end{eqnarray*}

\subsubsection{Proof of the stable convergence}
Based on subsection \ref{prfthm5.2}, \ref{prfthm5.3} and \ref{prfthm5.4},
\begin{equation}\label{QV_E1}
\langle \sqrt{r}\text{(E1)},\sqrt{r}\text{(E1)}\rangle=\langle N^{(1)},N^{(1)}\rangle_T+\langle N^{(2)},N^{(2)}\rangle_T+o_p(1)=\frac{2T}{3}\int^T_0\left(\sigma^{(g)}_t\right)^4\mathrm{d}t+o_p(1)
\end{equation}
Following the similar method as that in the proof of \textbf{Theorem \ref{thm3}}, we know
\begin{equation}\label{QV_E2}
\langle \sqrt{r}\text{(E2)},\sqrt{r}\text{(E2)}\rangle_T=\frac{24r^2}{TK^2}\int_0^T\left[h_t^2-h_tg_t^2+g_t^4\right]\mathrm{d}t+o_p\left(\frac{r^2}{K^2}\right)=O_p\left(\frac{r^2}{K^2}\right)
\end{equation}
We need a technical condition on the filtration $\{\mathcal{F}_t\}_{t\ge0}$ to which all the relevant processes are adapted:
\begin{assum} 
	\textbf{(Condition on the Filtration)} There are Brownian motions $W^{(1)}, W^{(2)}, \cdots, W^{(p)}$ that generate the filtration $\{\mathcal{F}_t\}_{t\ge0}$.
\end{assum}
Consider the normalized error process,
\begin{eqnarray*}
\sqrt{r}\text{(E)}&=&\sqrt{r}\text{(E1)}+\sqrt{r}\text{(E2)}=N^{(1)}_T+N^{(2)}_T+\underbrace{O_p\left(\frac{1}{\sqrt{r}}\right)}_{\text{the edge in }N^{(1)}_T+N^{(2)}_T}+\underbrace{\sqrt{r}\text{(E2)}}_{O_p\left(\frac{r}{K}\right)}\\
&=&2\sqrt{r}\sum^{r-1}_{i=0}\sum^K_{j=2}\Delta g_{t_{iK+j}}\times\sum^{j-1}_{l=1}\left(1+2\frac{j-1}{K}\frac{l-1}{K}-\frac{j-1}{K}-\frac{l-1}{K}\right)\Delta g_{t_{iK+l}}\\
&&+\sqrt{r}\sum^{r-1}_{i=1}\sum^K_{j=1}\frac{K-(j-1)}{K}\Delta g_{t_{iK+j}}\times\sum^K_{l=1}\frac{l-1}{K}\Delta g_{t_{(i-1)K+l}}+O_p\left(\frac{1}{\sqrt{r}}\vee\frac{r}{K}\right)
\end{eqnarray*}
Define
\begin{multline*}
N^n_t=2\sqrt{r}\sum^{r-1}_{i=0}\sum^K_{j=2}\left[\sum^{j-1}_{l=1}\left(1+2\frac{j-1}{K}\frac{l-1}{K}-\frac{j-1}{K}-\frac{l-1}{K}\right)\Delta g_{t_{iK+l}}\right]\cdot\Delta g_{t_{iK+j}\wedge t}\\
+\sqrt{r}\sum^{r-1}_{i=1}\sum^K_{j=1}\left(\sum^K_{l=1}\frac{l-1}{K}\Delta g_{t_{(i-1)K+l}}\right)\cdot\frac{K-(j-1)}{K}\Delta g_{t_{iK+j}\wedge t}
\end{multline*}
then $\sqrt{r}\text{(E)}_t=N^n_T+o_p(1)$. Suppose $t_{iK+j-1}=\max\{t_k, k=0,1,\cdots,n, t_k\le t\}$, then
\begin{multline*}
\mathrm{d}\langle N^n,W^{(i)}\rangle_t=2\sqrt{r}\left[\sum^{j-1}_{l=1}\left(1+2\frac{j-1}{K}\frac{l-1}{K}-\frac{j-1}{K}-\frac{l-1}{K}\right)\Delta g_{t_{iK+l}}\right]\mathrm{d}\langle g,W^{(i)}\rangle_t\\
+\sqrt{r}\left[\sum^K_{l=1}\frac{l-1}{K}\Delta g_{t_{(i-1)K+l}}\cdot\frac{K-(j-1)}{K}\right]\mathrm{d}\langle g,W^{(i)}\rangle_t
\end{multline*}
for $i=1,2,\cdots,p$, by Kunita-Watanabe inequality,
\begin{equation*}
\left|\langle g,W^{(i)}\rangle_{t+h}-\langle g,W^{(i)}\rangle_t\right|\le\sqrt{\langle g,g\rangle_{t+h}-\langle g,g\rangle_t}\cdot\sqrt{\langle W^{(i)},W^{(i)}\rangle_{t+h}-\langle W^{(i)},W^{(i)}\rangle_t}\le
\sigma^{(g)}_+h
\end{equation*}
so $\Delta \langle g,W^{(i)}\rangle_{t_k}\le \sigma^{(g)}_+\Delta(\mathcal{G})$. We have
\begin{multline*}
\langle N^n,W^{(i)}\rangle_T
=\underbrace{2\sqrt{r}\sum^{r-1}_{i=0}\sum^K_{j=2}\left[\sum^{j-1}_{l=1}\frac{(j-1)(l-1)+(K-(j-1))(K-(l-1))}{K^2}\Delta g_{t_{iK+l}}\right]\Delta\langle g,W^{(i)}\rangle_{t_{iK+j}}}_{\text{(NW1)}}\\
+\underbrace{\sqrt{r}\sum^{r-1}_{i=1}\sum^K_{j=1}\left[\frac{K-(j-1)}{K}\sum^K_{l=1}\frac{l-1}{K}\Delta g_{t_{(i-1)K+l}}\right]\Delta\langle g,W^{(i)}\rangle_{t_{iK+J}}}_{\text{(NW2)}}
\end{multline*}
note that
\begin{eqnarray*}
E(\text{NW1})^2&\le&\frac{4\left(\sigma^{(g)}_+\right)^4r}{n^3}\sum^{r-1}_{i=0}\sum^K_{j=2}\sum^{j-1}_{l=1}\left[\frac{2(j-1)-K}{K^2}(l-1)+\frac{K-(j-1)}{K}\right]^2=O_p\left(\frac{r^2K^2}{n^3}\right)=O_p\left(\frac{1}{n}\right)\\
E(\text{NW2})^2&\le&\frac{\left(\sigma^{(g)}_+\right)^4r}{n^3}\sum^{r-1}_{i=1}\sum^K_{j=1}\frac{(K-(j-1))^2}{K^2}\sum^K_{l=1}\frac{(l-1)^2}{K^2}=O_p\left(\frac{r^2K^2}{n^3}\right)=O_p\left(\frac{1}{n}\right)
\end{eqnarray*}
the first equality in the first line follows the calculation of $\langle N^{(1)},N^{(1)}\rangle_T$ in \textbf{Subsection \ref{prfthm5.3}}.

Hence, $\langle N^n,W^{(i)}\rangle_T=O_p\left(\frac{1}{\sqrt{n}}\right)$, combine the result for $\langle N^{(1)},N^{(1)}\rangle_T$ and $\langle N^{(2)},N^{(2)}\rangle_T$, \textbf{Theorem \ref{thm5}} follows.

\subsubsection{The sketchy proof for \textbf{Theorem \ref{thm5.0}}}
The first step in proving \textbf{Theorem \ref{thm5.0}} is to write the error term
$$E_{\mathcal{T}}^n\equiv\sqrt{\frac{n}{K}}\left(\frac{n}{sK^2}U(Y,K,s,2)^n_{\mathcal{T}}-E^{(1)}_n-E^{(2)}_n-E^{(3)}_n\right)$$
as a summation of independent items (as (\ref{rep_dis_error}) did). One intermediary  (key) step is to write
\begin{equation*}
	E_{\mathcal{T}}^n=A_{\mathcal{T}}^n+B_{\mathcal{T}}^n+(I)_{\mathcal{T}}^n+(II)_{\mathcal{T}}^n+(III)_{\mathcal{T}}^n+(IV)_{\mathcal{T}}^n+o_p(1)
\end{equation*}
where
\begin{eqnarray*}
	A_{\mathcal{T}}^n&=&\frac{1}{s}\left(\frac{n^{1/2}}{K}\right)^3\left(\frac{1}{r-s+1}\sum_{i=1}^{r-s+1}(m_i+m_{i+s})^2-\frac{2}{\mathcal{T}}\int_{0}^{\mathcal{T}}h_t\,\mathrm{d}t\right)\\
	B_{\mathcal{T}}^n&=&\sum_{i=1}^{r-s+1}\sum_{l=1}^{K}\frac{(l-1)^2}{K^2}(\Delta g_{t_{(i-1)K+l}})^2+\sum_{i=1}^{r-s+1}\sum_{l=1}^{K}\frac{(K-(l-1))^2}{K^2}(\Delta g_{t_{(i+s-2)K+l}})^2\\
	&&+\sum_{i=1}^{r-s+1}\sum_{j=1}^{s-2}\sum_{l=1}^{K}\left[2\frac{j-1}{s-2}+\left(2\frac{l-1}{(s-2)K}-1\right)\right]^2(\Delta g_{t_{(i+j-1)K+l}})^2-E^{(1)}_n-E^{(2)}_n
\end{eqnarray*}
besides, $(I)^n_\mathcal{T}=\sum_{i=1}^{r-s+1}(I)^n_i$, $(II)^n_\mathcal{T}=\sum_{i=1}^{r-s+1}(II)^n_i$, $(III)^n_\mathcal{T}=\sum_{i=1}^{r-s+1}(III)^n_i$, and  $(IV)^n_\mathcal{T}=\sum_{i=1}^{r-s+1}(IV)^n_i$ where
\begin{eqnarray*}
	(I)_i^n&=&\frac{2r^{1/2}}{s}\sum_{l=2}^{K}\frac{l-1}{K}\Delta g_{t_{(i-1)K+l}}\times\sum_{h=1}^{l-1}\frac{h-1}{K}\Delta g_{t_{(i-1)K+h}}\\
	&&\hspace{10mm}+\frac{2r^{1/2}}{s}\sum_{l=2}^{K}\frac{K-(l-1)}{K}\Delta g_{t_{(i+s-2)K+l}}\times\sum_{h=1}^{l-1}\frac{K-(h-1)}{K}\Delta g_{t_{(i+s-2)K+h}}\\
	&&\hspace{20mm}-\frac{r^{1/2}}{s}\sum_{l=1}^{K}\frac{K-(l-1)}{K}\Delta g_{t_{(i+s-2)K+l}}\times\sum_{l=1}^{K}\frac{l-1}{K}\Delta g_{t_{(i-1)K+l}}\\
	(II)_i^n&=&\frac{r^{1/2}}{s}\sum_{l=1}^{K}\frac{K-(l-1)}{K}\Delta g_{t_{(i+s-2)K+l}}\times\sum_{j=1}^{s-2}\sum_{l=1}^{K}\left(\frac{2(l-1)}{(s-2)K}+\frac{2(j-1)}{s-2}-1\right)\Delta g_{t_{(i+j-1)K+l}}\\
	&&\hspace{10mm}-\frac{r^{1/2}}{s}\sum_{j=1}^{s-2}\sum_{l=1}^{K}\left(\frac{2(l-1)}{(s-2)K}+\frac{2(j-1)}{s-2}-1\right)\Delta g_{t_{(i+j-1)K+l}}\times\sum_{l=1}^{K}\frac{l-1}{K}\Delta g_{t_{(i-1)K+l}}\\
	(III)_i^n&=&\frac{2r^{1/2}}{s}\sum_{j=1}^{s-2}\sum_{l=2}^{K}\left(\frac{2(l-1)}{(s-2)K}+\frac{2(j-1)}{s-2}-1\right)\Delta g_{t_{(i+j-1)K+l}}\\
	&&\hspace{50mm}\times\sum_{h=1}^{l-1}\left(\frac{2(h-1)}{(s-2)K}+\frac{2(j-1)}{s-2}-1\right)\Delta g_{t_{(i+j-1)K+h}}\\
	(IV)_i^n&=&\frac{2r^{1/2}}{s}\sum_{j=2}^{s-2}\sum_{l=1}^{K}\left(\frac{2(l-1)}{(s-2)K}+\frac{2(j-1)}{s-2}-1\right)\Delta g_{t_{(i+j-1)K+l}}\\
	&&\hspace{50mm}\times\sum_{k=1}^{j-1}\sum_{l=1}^{K}\left(\frac{2(l-1)}{(s-2)K}+\frac{2(k-1)}{s-2}-1\right)\Delta g_{t_{(i+k-1)K+l}}
\end{eqnarray*}

The next task is to show the joint asymptotics of $A_\mathcal{T}^n$, $B_\mathcal{T}^n$, $(I)_\mathcal{T}^n$, $(II)_\mathcal{T}^n$, $(III)_\mathcal{T}^n$ and $(IV)_\mathcal{T}^n$ (which is a similar task of subsection \ref{prfthm5.3} and \ref{prfthm5.4}). By \textbf{Theorem \ref{thm3}}, $A_\mathcal{T}^n=O_p\left(\frac{1}{s}\cdot\frac{n}{K^2}\right)$. Some calculation (omitted here) yields $B_\mathcal{T}^n=O_p(s/r)$, $(I)_\mathcal{T}^n=O_p(1/s)$, $(II)_\mathcal{T}^n=O_p(1/s^{1/2})$, $(III)_\mathcal{T}^n=O_p(1/s^{1/2})$, and 
$$(IV)_\mathcal{T}^n\overset{\mathcal{L}-s}{\longrightarrow}\mathcal{MN}\left(0,\frac{2\mathcal{T}}{9}\int_{0}^{\mathcal{T}}(\sigma_t^{(g)})^4\,\mathrm{d}t\right)$$
hence \textbf{Theorem \ref{thm5.0}} follows.
\end{proof}

\subsection{Proof of Lemma \ref{lemmareg}}\label{prflemreg}
\begin{proof}
	Followed from (\ref{lin1}), $\widehat{g}_t+(g_t-\widehat{g}_t)=\widehat{\beta}_m\widehat{\sigma}_t^2+\widehat{\alpha}_m+(\beta\sigma_t^2-\widehat{\beta}_m\widehat{\sigma}_t^2)+(\alpha-\widehat{\alpha}_m)+\zeta_t$, so we have
	\begin{eqnarray}\label{regeta}
	\eta^{(m)}_t&=&\widehat{g}_t-\widehat{\beta}_m\widehat{\sigma}^2_t-\widehat{\alpha}_m\nonumber\\
	&=&(\widehat{g}_t-g_t)+\beta(\sigma_t^2-\widehat{\sigma}_t^2)+\beta\widehat{\sigma}_t^2-\widehat{\beta}_m\widehat{\sigma}_t^2+(\alpha-\widehat{\alpha}_m)+\zeta_t
	\end{eqnarray}	
	Note that $\widehat{\sigma}^2_t\overset{\mathbb{P}}{\to}\sigma^2_t$ and $\widehat{g}_t\overset{\mathbb{P}}{\to}g_t$, thus plugging (\ref{regeta}) into (\ref{lin2}), we get
	\begin{equation*}
	\widehat{g}_t=\beta\widehat{\sigma}_t^2+\alpha+\zeta_t+o_p(1)
	\end{equation*}
	so the estimates obtained from linear regression on the pairs $(\widehat{\sigma}_t^2,\widehat{g}_t)$'s are consistent, i.e., $\widehat{\beta}_n$ converges to $\beta$ and $\widehat{\alpha}_n$ converges to $\alpha$ in the in-fill asymptotic setting, provided (\ref{lin1}) holds.
\end{proof}


\newpage
\bibliographystyle{apa}
\bibliography{Reference}
\end{document}